\documentclass[aos]{imsart}

\usepackage{amsthm}
\usepackage{amsmath,amsfonts,amssymb}
\usepackage[authoryear]{natbib} 
\usepackage[colorlinks,citecolor=blue,urlcolor=blue,hypertexnames=false]{hyperref} 

\usepackage{xr}  
\usepackage{bbm}
\usepackage{graphicx,psfrag}
\usepackage{enumerate} 
\usepackage{url} 
\usepackage{bm}
\usepackage[bottom]{footmisc}
\usepackage{booktabs} 
\usepackage{color}
\usepackage{mathtools}
\usepackage{bbm}
\usepackage{multirow} 
\usepackage{multibib} 
           
\newtheorem{lemma}{Lemma}
\newtheorem{corollary}{Corollary}
\newtheorem{proposition}{Proposition}
\newtheorem{theorem}{Theorem}
\newtheorem{definition}{Definition}
\theoremstyle{definition}
\newtheorem{remark}{Remark}
\newtheorem{example}{Example}

\newtheorem{assumption}{Assumption}
\newtheorem{condition}{Condition}

\DeclareMathOperator*\diag{diag}

\DeclarePairedDelimiter{\norm}{\lVert}{\rVert}

\newcommand{\ind}{\mathbbm{1}}

\newcommand{\cA}{\mathcal{A}}
\newcommand{\sA}{\mathsf{A}}

\newcommand{\sB}{\mathsf{B}}

\newcommand{\sC}{\mathsf{C}}

\newcommand{\cE}{\mathcal{E}}

\newcommand{\bbE}{\mathbb{E}}

\newcommand{\cF}{\mathcal{F}}

\newcommand{\cG}{\mathcal{G}}

\newcommand{\bI}{\bm{I}}

\newcommand{\bK}{\bm{K}}

\newcommand{\cM}{\mathcal{M}}

\newcommand{\cN}{\mathcal{N}}

\newcommand{\bbN}{\mathbb{N}}

\newcommand{\cO}{\mathcal{O}}

\newcommand{\bP}{\bm{P}}

\newcommand{\bQ}{\bm{Q}}

\newcommand{\bR}{\bm{R}}

\newcommand{\bbR}{\mathbb{R}}
\newcommand{\fR}{\mathfrak{R}}

\newcommand{\cS}{\mathcal{S}}
\newcommand{\sS}{\mathsf{S}}

\newcommand{\cT}{\mathcal{T}}

\newcommand{\cX}{\mathcal{X}}
\newcommand{\sX}{\mathsf{X}}

\newcommand{\sY}{\mathsf{Y}}

\DeclarePairedDelimiter{\TV}{\lVert}{\rVert_{\mathrm{TV}}}

\def\adds{\cN_{\mathrm{a}}}  
\def\dels{\cN_{\mathrm{d}}} 
\def\swaps{\cN_{\mathrm{s}}}  
\def\nstar{\cN_{\star}}
\def\wadd{w_{\mathrm{a}}}
\def\wdel{w_{\mathrm{d}}}
\def\wswap{w_{\mathrm{s}}}
\def\wstar{w_{\star}}
\def\Zadd{Z_{\mathrm{a}}}
\def\Zdel{Z_{\mathrm{d}}}
\def\Zswap{Z_{\mathrm{s}}}
\def\Zstar{Z_{\star}}
\def\hadd{h_{\mathrm{a}}}
\def\hdel{h_{\mathrm{d}}}
\def\hswap{h_{\mathrm{s}}}
\def\hstar{h_{\star}}
\def\true{\gamma_{\rm{true}}}

\def\PR{B}     
\def\PP{\pi_n}
\def\LOT{LIT-MH}  
\def\RW{RW-MH}
\def\LIB{LB-MH}
\def\KRW{\bK_{\mathrm{rw}}}

\def\KLOT{\bK_{\mathrm{lit}}}
\def\Tmix{T_{\mathrm{mix}}}
\def\MHP{\bm{P}_{\rm{lit}}}

\def\KZ{\bK_{\mathrm{lb}}}

\def\KLIB{\bK_{\mathrm{lb}}}
\def\acc{\mathrm{acc}}
\def\fR{\mathfrak{R}}

\def\ev{\sigma^2_z}
\def\X{X}
\def\ty{y_{\rm{s}}}
\def\tX{\tilde{\mathsf{X}}}
\def\PJ{P} 
\def\OPJ{P^{\perp}} 
\def\Cbeta{\beta_{\rm{min}}}
\def\Ga{\gamma^* \setminus \gamma}
\def\Gb{\gamma \setminus \gamma^*}
\def\tgamma{\tilde{\gamma}}
\def\YWJ{1}

\def\N{\mathrm{MN}}
\def\P{\mathbb{P}}
\def\E{\mathbb{E}}
\def\CC{\mathfrak{C}}
\newcommand{\cond}[1]{(\YWJ{}\ref{#1})}

\newcites{s}{References}
\def\TITLE{Dimension-free Mixing for High-dimensional Bayesian Variable Selection}

\begin{document}

\begin{frontmatter} 
\title{\TITLE{}}
\runtitle{Dimension-free mixing of LIT-MH} 

\begin{aug} 
\author[A]{\fnms{Quan} \snm{Zhou}\ead[label=e1]{quan@stat.tamu.edu}}, 
\author[B]{\fnms{Jun} \snm{Yang}\ead[label=e2]{jun.yang@stats.ox.ac.uk}}, 
\author[C]{\fnms{Dootika} \snm{Vats}\ead[label=e3]{dootika@iitk.ac.in}}, 
\author[D]{\fnms{Gareth O.} \snm{Roberts}\ead[label=e4]{Gareth.O.Roberts@warwick.ac.uk}}

\and 
\author[E]{\fnms{Jeffrey S.} \snm{Rosenthal}\ead[label=e5]{jeff@math.toronto.edu}}
\address[A]{Department of Statistics, Texas A\&M University}
\address[B]{Department of Statistics, University of Oxford}
\address[C]{Department of Mathematics and Statistics, Indian Institute of Technology Kanpur}
\address[D]{Department of Statistics, University of Warwick}
\address[E]{Department of Statistical Sciences, University of Toronto}
\end{aug}

\begin{abstract} 
\citet{yang2016computational} proved that the symmetric random walk Metropolis--Hastings algorithm for Bayesian variable selection is rapidly mixing under mild high-dimensional assumptions. 
We propose a novel MCMC sampler using an informed proposal scheme, which we prove achieves a much faster mixing time that is independent of the number of covariates, under the assumptions of~\citet{yang2016computational}. To the best of our knowledge, this is the first high-dimensional result which rigorously shows that the mixing rate of informed MCMC methods can be fast enough to offset the computational cost of local posterior evaluation. 
Motivated by the theoretical analysis of our sampler, we further propose a new approach called ``two-stage drift condition'' to studying convergence rates of Markov chains on general state spaces, which can be useful for obtaining tight complexity bounds in high-dimensional settings.
The practical advantages of our algorithm are illustrated by both simulation studies and real data analysis. 
\end{abstract}

\end{frontmatter}

\section{Introduction} \label{sec:intro}
Consider a variable selection problem where we observe an $n \times p$ design matrix $\X$ and a response vector $y$; each column of $\X$ represents a covariate. 
The goal is to identify the set of all ``influential'' covariates which have non-negligible effects on $y$; we denote this set by $\gamma$. 
We are mostly interested in a high-dimensional setting where $p$ is much larger than the sample size $n$ but most of the covariates have either zero or negligible effects. 
Due to this sparsity assumption, we can choose some threshold $s_0$, which may grow with $n$, and assume that the unknown parameter $\gamma$ takes value in the space 
\begin{align*}
\cM(s_0) = \{  \gamma \subseteq \{1, 2, \dots, p\} \colon |\gamma| \leq s_0  \}, 
\end{align*}
where $|\cdot|$ denotes the cardinality of a set. 
By assigning a prior distribution on $\cM(s_0)$ and then updating it using the data, we can compute the posterior distribution of $\gamma$, denoted by $\PP(\gamma)$~\citep{chipman2001practical}. 
One advantage of the Bayesian approach is that we can make inferences by averaging over $\PP$, a property known as  model averaging~\citep{kass1995bayes}.  This is different from methods such as penalized regression,  where we aim to find a single best model that minimizes some loss function.   
For theoretical results on Bayesian variable selection in high-dimensional settings, see~\citet{johnson2012bayesian, narisetty2014bayesian, castillo2015bayesian, jeong2021unified}, among many others. 

\subsection{Background and main contributions of this work}\label{sec:contributions}
The calculation of $\PP$ is usually performed by  Markov chain Monte Carlo (MCMC) sampling, including both Metropolis-Hastings (MH) and Gibbs algorithms~\citep{george1993variable, george1997approaches, brown1998multivariate, guan2011bayesian}; see~\citet{o2009review} for a review. 
For problems with extremely large $p$, the efficiency of the MCMC sampler largely depends on how we propose the next state given current state $\gamma$.  
\citet{zanella2020informed} considered the so-called ``locally informed'' proposal schemes on general discrete state spaces, which assign a proposal weight to each neighboring state $\gamma'$ using some function of $\PP(\gamma') / \PP(\gamma)$. Though variable selection was not discussed explicitly in~\citet{zanella2020informed}, similar ideas are utilized in most state-of-the-art MCMC methods for variable selection. 
Examples include the tempered Gibbs sampler of~\citet{zanella2019scalable} and the ASI (adaptively scaled individual adaptation) proposal of~\citet{griffin2021search}, both of which require calculating $\PP(\gamma') $ (up to the normalizing constant) for each $\gamma' \in \cN_1 (\gamma) = \{ \gamma' \colon |\gamma \triangle \gamma' | = 1 \}$, where $\triangle$ denotes the symmetric set difference. 
In a similar spirit, the Hamming ball sampler of~\citet{titsias2017hamming} performs an exact sampling according to $\PP$ within a randomly selected subset of the neighborhood of $\gamma$. 
For non-MCMC algorithms, we note that the design of the shotgun stochastic search~\citep{hans2007shotgun, shin2018scalable} bears a striking resemblance to informed proposals. 

Intuitively, informed proposals rely on the following idea: avoid visiting states with low posterior probabilities by carefully tuning the proposal probabilities. Though this seems very appealing, it naturally comes at the computational cost of evaluating the local posterior landscape around the current state.
For instance, an informed proposal that draws the next state from $\cN_1 (\cdot)$ has complexity linear in $p$. 
Whether such local evaluation of $\PP$ is worthwhile is theoretically unclear, and convergence analysis of informed sampling algorithms (for variable selection) is very challenging because the landscape of $\PP$ is hard to characterize, especially in high-dimensional asymptotic regimes. Indeed, even for the ``uninformed'' random-walk MH algorithm (denoted by \RW{} henceforth), its mixing rate has only been obtained recently by~\citet{yang2016computational} under mild high-dimensional assumptions. The order of their upper bound on the mixing time is approximately $p n s_0^2 \log p$ (see Remark~\ref{rmk:mix}), which shows that \RW{} is rapidly mixing (i.e., the mixing time is polynomial in $n$ and $p$). 
Then, the question is whether informed MCMC methods can achieve sufficiently fast mixing rates that can at least offset the additional computation cost.  

In this work, we rigorously derive a positive answer to the above question. 
We consider a novel informed MH algorithm, named \LOT{}  (Metropolis--Hastings with Locally Informed and Thresholded proposal distributions), which assigns bounded proposal weights to the standard add-delete-swap moves. Under the high-dimensional assumptions made in~\citet{yang2016computational}, \LOT{} achieves a mixing time that does not depend on $p$. To the best of our knowledge, this is the first dimension-free mixing time result for an informed MCMC algorithm in a ``general'' high-dimensional setting. (There exist similar results for special cases where the posterior distribution has independent coordinates or the design matrix is orthogonal, which are not very useful for real high-dimensional problems; see, e.g.,~\citet{zanella2019scalable} and~\citet{griffin2021search}.) To prove the mixing rate of \LOT{}, unlike most existing approaches based on path methods, we propose a ``two-stage drift condition'' method, which provides theoretical insights into the behavior of MCMC methods for variable selection. General results for the two-stage drift condition are derived, which can be useful to other problems where multiple drift conditions hold on different parts of the state space. 
Simulation studies show that \LOT{} can efficiently explore the posterior distribution under various settings. A real data example is also provided, where five genetic variants associated with cup-to-disk ratio are identified. 

\subsection{Motivation for the \LOT{} algorithm}\label{sec:intro.lit} 
One may expect that by using an informed proposal scheme that assigns larger proposal probabilities to states with larger posterior, the resulting MH algorithm requires less iterations than \RW{} to find high posterior regions. 
This is not always true, and surprisingly, it is even possible that such an informed MH algorithm is slowly mixing while \RW{} is rapidly mixing.   

Consider MH algorithms for variable selection that always propose the next state from $\cN_1(\cdot)$; that is, we can either add or remove a covariate (we will consider swap moves later in Section~\ref{sec:model}).  
Suppose we assign proposal weight $ \PP(\gamma')^\nu$ to each $\gamma' \in \cN_1(\gamma)$ for some constant $\nu \geq  0$. That is, we can express the proposal matrix $\bK_\nu$ as 
\begin{equation}\label{eq:def.Kv}
\bK_\nu (\gamma, \gamma') = \frac{  \PP(\gamma')^\nu  }{ \sum_{\tgamma \in \cN_1(\gamma)} \PP(\tgamma)^\nu }  \ind_{\cN_1(\gamma)}(\gamma'),
\end{equation} 
where $\ind$ denotes the indicator function. 
When $\nu = 0$, $\bK_\nu(\gamma, \cdot)$ becomes the uniform distribution on the set $\cN_1(\gamma)$, which is uninformed. It seems desirable to choose some $\nu > 0$  so that with high probability we propose adding an influential covariate or removing a non-influential one. 
We give a toy low-dimensional example below, which shows that for any $\nu > 0$,  the MH algorithm using $\bK_\nu$ as the proposal can fail to work well when the sample size is sufficiently large. 

\begin{example}\label{ex1} 
Suppose that there are only two influential covariates, $X_1$ and $X_2$, and $\PP(\{1, 2\}) \gg \PP(\{i\}) \gg \PP(\emptyset) \gg \PP(\{j\})$ for $i = 1, 2$ and $3 \leq j \leq p$. 
Thus, if we start an MH algorithm at the null model, we want the chain to first move to $\{1\}$ or $\{2\}$ and then move to $\{1, 2\}$. 
By using some $\nu > 0$, we can make the proposal probability $\bK_\nu(\emptyset, \{1\} \cup \{2\})$  close to 1. 
Let $\bP_\nu$ denote the transition matrix of the MH algorithm with proposal $\bK_\nu$ given in~\eqref{eq:def.Kv}.   
To bound the transition probability from   $\emptyset$ to $\{1\}$, observe that $\bK_\nu (\{1\}, \emptyset) \leq \PP(\emptyset)^\nu / \PP( \{1, 2\})^\nu$,  since $\{1, 2\}$ is a neighbor of $\{1\}$.  
It then follows from the Metropolis rule that 
\begin{equation}\label{eq:ex1}
\begin{aligned} 
   \bP_\nu (\emptyset, \{1\}) =\;& \bK_\nu (\emptyset, \{1\})  \min \left\{ 1 ,    \,  \frac{ \PP(\{1\})}{ \PP(\emptyset )}  \frac{ \bK_\nu (\{1\}, \emptyset) }{ \bK_\nu (\emptyset, \{1\})}  \right\}\\
   \leq \;& 
   \frac{ \PP(\{1\})}{ \PP(\emptyset )}   \bK_\nu (\{1\}, \emptyset) \leq 
   \left\{ \frac{ \PP( \{1\} ) }{ \PP(\emptyset)} \right\}^{1 - \nu}  \left\{ \frac{ \PP( \{1\}) }{ \PP( \{1, 2\}) } \right\}^{ \nu}. 
\end{aligned} 
\end{equation}  
An analogous bound holds for $\bP_\nu (\emptyset, \{2\})$.  
It is clear from~\eqref{eq:ex1} that if $\nu > 1$, $\bP_\nu (\emptyset, \{1\})$ can  be exceedingly small. 

Next, we construct a concrete example to show that even if $\nu \in (0, 1]$, $\bP_\nu$ may still have very poor mixing. 
Fix some $\nu \in (0, 1]$. 
Let $X_j$ denote the $j$-th column of $\X$. 
Suppose  the design matrix satisfies $X_j^\top X_j = n$ for each $j \in [p]$, $X_1^\top X_2 = (\nu - 1) n$, and $X_i^\top X_j = 0$ for any other $i < j$.  
Assume that the response vector $y$ is generated by $y = X_1 + X_2 + z $  where $z$ is a deterministic error vector such that $z^\top z  = n$ and $\X^\top z = 0$. 
Choose some $s_0 \geq 2$, and let the posterior distribution be given by~\eqref{eq:post} with hyperparameters $\kappa, g > 0$ (see Section~\ref{sec:model} for details). Fix $p, \nu, \kappa, g$ and let $n$ tends to infinity. 
In Section~\ref{proof.ex1} in the supplement, we show that 
\begin{align*}
    \bK_\nu (\emptyset, \{1\} \cup \{2\}) = 1 - O(e^{- a_1 n}), \quad \quad   
  \bP_\nu (\emptyset, \{1\} \cup \{2\})  = O(e^{ -a_2 n}). 
\end{align*}
where $a_1, a_2 > 0$ are some constants that only depend on $\nu$ and $g$. Hence, the chain must be slowly mixing since $\bP_\nu(\emptyset,  \emptyset) = 1 - O(e^{ -a_1 n}) - O(e^{ -a_2 n})$. That is,  the informed MH chain can get stuck at the null model for exponentially many iterations, where we keep proposing adding $X_1$ or $X_2$ but getting rejected.  
In contrast, one can use the path method of~\citet{yang2016computational} to show that \RW{} is rapidly mixing (proof is omitted). 
\end{example} 

This toy example reveals that the real challenge in developing informed MH algorithms is to bound the acceptance probability of informed proposals. From~\eqref{eq:ex1}, we can see that in order to make $\bP_\nu (\emptyset, \{1\})$ large, we need $\bK_\nu (\emptyset, \{1\})$ to be sufficiently large and $\bK_\nu (\{1\}, \emptyset) $ not to be too small. 
This motivates us to use proposal weights that are bounded both from above and from below so that the proposal probability of any neighboring state is bounded as well. 
Further, we partition the neighborhood of each $\gamma$ according to the proposal type (e.g. addition, deletion or swap) and then perform proposal weighting in each subset separately, which also helps control the acceptance probability of informed proposal moves. 

\subsection{Two-stage drift condition}\label{sec:intro.drift}
Drift-and-minorization methods have been used to show rapid mixing of various MCMC algorithms \citep{rosenthal1995minorization, fort2003geometric,roy2007convergence, vats2017geometric,johndrow2020scalable, yang2017complexity,qin2019convergence}; see~\citet{jones2001honest} for a review. These methods are particularly useful for studying Gibbs sampling algorithms on continuous state spaces. 
One possible reason is that to establish the drift condition, we need to bound the expected change in the drift function in the next MCMC iteration, which is easier if the next sample is drawn from a smooth full conditional distribution.  
For problems like high-dimensional variable selection,  the posterior landscape is highly irregular and difficult to characterize.  The convergence analysis becomes even more challenging for informed MH algorithms since the proposal distribution usually involves normalizing constants that do not admit simple expressions. 

We prove the dimension-free mixing rate of the \LOT{} algorithm using a novel drift condition. 
But unlike traditional drift-and-minorization methods which only involve a single drift condition, we establish two drift conditions on two disjoint subsets of the state space separately. 
Our method is motivated by the forward-backward stepwise  selection~\citep{an2008stepwise} and the insights obtained in~\citet{yang2016computational}. 
Let $\gamma^*$ denote the model consisting of all influential covariates. 
We say a model $\gamma$ is overfitted if $\gamma^* \subseteq \gamma$; otherwise, we say $\gamma$ is underfitted. 
Under certain mild conditions, we expect that the posterior probability mass will concentrate on $\gamma^*$. 
It is then tempting to use  a single drift condition that measures the distance between $\gamma$ and $\gamma^*$.  Unfortunately, this approach may not work. 
The most important reason is that for an underfitted $\gamma$,  non-influential covariates may appear to be influential due to the correlation with some truly influential covariate(s)  missing in $\gamma$ (and similarly, some influential covariates may appear to be non-influential).
Nevertheless, as in the stepwise variable selection, once the model becomes overfitted, we expect that all non-influential covariates can be easily removed.  
This observation suggests that we can partition $\cM(s_0)$ into underfitted and overfitted models. 
On the set of overfitted models, we may construct a drift function using the distance from $\gamma^*$, and the corresponding drift condition should reflect that the chain tends to move towards $\gamma^*$ by removing non-influential covariates. 
On the set of underfitted models, we need a different drift condition capturing the tendency of the chain to add (possibly truly non-influential) covariates, in order to explain the variation in the response variable. 

We propose to use this two-stage drift condition as a general method for convergence analysis of Markov chains; all related results will be derived for general state spaces in Section~\ref{sec:main}.
The flexibility of this approach could be useful to other problems where the state space has a complex topological structure. 
To derive a bound on the mixing time using the two-stage drift condition, we use regeneration theory as in the classical drift-and-minorization methods~\citep{roberts1999bounds}, but it is more difficult in our case to bound the tail probability of the regeneration time. 
In our proof, we split the path of the Markov chain into disjoint segments using an auxiliary sequence of geometric random variables and then apply a union bound argument of~\citet{rosenthal1995minorization}. 
 
The use of the two-stage drift condition is critical to proving the dimension-free mixing of \LOT{}. 
In~\citet{yang2016computational},  the convergence rate of \RW{} is analyzed by using canonical paths~\citep{sinclair1992improved}, a method widely used for Markov chains on discrete spaces~\citep[Chapter 14]{levin2017markov}. 
A key step of their proof is to identify, for any $\gamma \neq \gamma^*$,  a ``high-probability'' path from $\gamma$ to $\gamma^*$ (``high-probability'' means that each step of the path has a sufficiently large transition probability).  
A potential limitation of this approach is that for some $\gamma \neq \gamma^*$, there may exist a large number of ``high-probability'' paths leading to $\gamma^*$, and if we only consider one of them,  the resulting mixing time bound may be loose.  
This is indeed the case for our \LOT{} algorithm. 
In order to obtain a sharp bound on the mixing time, we need to invoke the drift condition to take into account all possible moves, and the method of canonical paths will fail to yield a dimension-free estimate for the mixing time of \LOT{}.  
 
\subsection{Organization of the paper}
In Section~\ref{sec:model} we formally introduce the Bayesian variable selection problem. 
Key results of~\citet{yang2016computational} for the \RW{} algorithm are reviewed in Section~\ref{sec:ywj}, and our \LOT{} algorithm is introduced in Section~\ref{sec:informed}. 
In Section~\ref{sec:ss.two}, we construct two drift conditions for \LOT{} and then derive the mixing time bound  in Theorem~\ref{th:mh.mix}. 
In Section~\ref{sec:main}, we consider the two-stage drift condition in a general setting, for which the main result is presented in Theorem~\ref{th:drift}. 
Simulation studies are presented in Section~\ref{sec:sim}, with some results provided in Section~\ref{supp:sim} in the supplement. 
A real data example is provided in Section~\ref{sec:gwas}, where we apply the \LOT{} algorithm to genome-wide association studies on glaucoma. 
Section~\ref{sec:discuss} concludes the paper with some discussion on the implementation and generalization of \LOT{} and its differences from other MCMC methods. 
All technical proofs are relegated to the supplement.

\section{\RW{} and \LOT{} algorithms for  variable selection}\label{sec:selection} 
We first define some notation. Let $[p] = \{1, 2, \dots, p\}$. 
For  $\gamma \subseteq [p]$, let $\X_\gamma$  denote the submatrix of $X$ with columns indexed by $\gamma$, and $\beta_\gamma$  denote the subvector with entries indexed by $\gamma$.
Recall  that $|\cdot|$ denotes the cardinality of a set. 
  
\subsection{Model, prior and local proposals}\label{sec:model}
Consider a sparse linear regression model,
\begin{align*}
y = \X_\gamma \beta_\gamma + e,  \quad e \sim \N(0, \, \phi^{-1} I_n), 
\end{align*} 
where  $\N$ denotes the multivariate normal distribution and $I_n$ is the identity matrix. 
Hence, $\gamma$ can be understood as the set of nonzero entries of $\beta$. 
We follow~\citet{yang2016computational} to consider the following prior: 
\begin{equation} \label{eq:priors}
\begin{array}{ll}
\text{(g-prior)}      &  \beta_\gamma \mid \gamma \sim \N(0, \, g \phi^{-1} (\X_\gamma^\top \X_\gamma)^{-1} ), \\
\text{(precision prior)}  \quad \quad    & \pi_0(\phi) \propto \phi^{-1}, \\
\text{(sparsity prior)}      &  \pi_0(\gamma) \propto p^{- \kappa_0 |\gamma|} \ind_{\cM(s_0)} (\gamma),  \\
\text{(choice of $g$)}      &  1 + g =  p^{2 \kappa_1}, 
\end{array}
\end{equation}
where $\kappa_0, \kappa_1 > 0$ are hyperparameters,  $\pi_0$ denotes the prior probability density/mass function, and we recall $s_0$ is the maximum model size we allow. 
After integrating out $\beta$, the marginal posterior probability of $\gamma \subseteq [p]$ can be computed by  
\begin{equation}\label{eq:post}
\PP( \gamma) \propto p^{- \kappa |\gamma|} 
\left( g^{-1} y^\top y  +  y^\top  \OPJ_\gamma   y   \right)^{-n/2} \ind_{\cM(s_0)} (\gamma), 
\end{equation}
where $\kappa = \kappa_0 + \kappa_1$ and  $\OPJ_\gamma $ denotes the projection matrix: 
\begin{equation*} 
 \OPJ_\gamma = I_n - \PJ_\gamma,  \quad \quad  \PJ_\gamma = \X_\gamma (\X_\gamma^\top \X_\gamma)^{-1} \X_\gamma^\top. 
\end{equation*}
For two models $\gamma, \gamma'$, let $\PR(\gamma, \gamma')$ denote their posterior probability ratio.
It follows from~\eqref{eq:post} that 
\begin{equation}\label{eq:ppr}
\PR(\gamma, \gamma') = \frac{\PP(\gamma' )}{ \PP(\gamma)} = p^{\kappa (|\gamma| - |\gamma'|)} \left\{ 1 + \frac{y^\top (\PJ_{\gamma} - \PJ_{\gamma'})y}{g^{-1} y^\top y + y^\top \OPJ_\gamma y}  \right\}^{-n/2}. 
\end{equation} 

For MH algorithms on the space $\cM(s_0)$, the most common approach is to use a proposal scheme consisting of three types of local moves, ``addition'', ``deletion'' and ``swap'',  which induces an irreducible Markov chain on $\cM(s_0)$. Explicitly, for every $\gamma \in \cM(s_0)$, define the addition, deletion and swap neighborhoods of $\gamma$ by 
\begin{equation}\label{eq:def.neighbor}
\begin{aligned} 
\adds(\gamma) =\;&  \{ \gamma' \in \cM(s_0)  \colon  \gamma' = \gamma \cup \{j\} \text{ for some } j \notin \gamma  \},   \\
\dels(\gamma) = \;&   \{ \gamma' \in \cM(s_0)  \colon  \gamma' = \gamma \setminus \{k\} \text{ for some } k \in \gamma  \}, \\
\swaps(\gamma) =\;&  \{  \gamma' \in \cM(s_0)  \colon \gamma' = (\gamma \cup \{j\} ) \setminus \{k\}  \text{ for some } j \notin\gamma, \, k \in \gamma \}. 
\end{aligned}
\end{equation} 
The three sets are disjoint, and for each $\gamma \in \cM(s_0)$, $|\adds(\gamma) \cup \dels(\gamma)| \leq p$ (the equality holds if $|\gamma| < s_0$) and $|\swaps(\gamma)| \leq p s_0$. The definitions of these neighborhoods can be generalized by allowing changing more covariates at one time. 
Let $\cN(\gamma) = \adds(\gamma) \cup \dels(\gamma) \cup \swaps(\gamma)$ denote the union of the three neighborhoods.  

\subsection{Rapid mixing of the \RW{} algorithm }\label{sec:ywj}  
Consider a proposal scheme defined by a transition probability matrix $\KRW\colon \cM(s_0) \times \cM(s_0) \rightarrow [0, 1]$ such that 
\begin{equation}\label{eq:KRW}
\KRW(\gamma, \gamma') =   \frac{\hadd(\gamma) \ind_{\adds(\gamma) } (\gamma')    }{ |\adds(\gamma) | } + 
 \frac{\hdel(\gamma) \ind_{ \dels(\gamma)} (\gamma') }{ |\dels(\gamma)| } + \frac{ \hswap(\gamma) \ind_{\swaps(\gamma)}(\gamma') }{  |\swaps(\gamma)| },
\end{equation}
where $\hadd(\gamma), \hdel(\gamma), \hswap(\gamma)$ are non-negative constants that sum to $1$; thus, $\hadd(\gamma)$ is the probability of proposing an addition move given current state $\gamma$.  
When $\hadd, \hdel, \hswap$ are all constants independent of $\gamma$, we refer to the resulting MH algorithm as asymmetric \RW{}. 
If $\hadd(\gamma) = | \adds(\gamma) | / 2p, \hdel (\gamma) = |\dels(\gamma)| / 2p$ and $\hswap(\gamma) = 1/2$, the resulting MH algorithm is called symmetric \RW{}, since  $\KRW(\gamma, \gamma') = \KRW(\gamma', \gamma)$ for any $\gamma, \gamma' \in \cM(s_0)$.  
We will now explain the main idea of the proof for the rapid mixing of \RW{} given in~\citet{yang2016computational}, which will be useful later for studying the mixing time of \LOT{}. Note that though~\citet{yang2016computational} only considered symmetric \RW{},  their argument can be easily extended to prove the rapid mixing of asymmetric \RW{}. 

A key step in the mixing time analysis of \RW{} is to characterize the shape of $\PP$. To this end, suppose that the true model is given by $y = X \beta^* + z$ where $z \sim \N(0,  \ev I_n)$ and define  
\begin{equation}\label{eq:beta.min}
\gamma^* = \{ j \in [p]: \, | \beta^*_j | \geq  \Cbeta \}, \quad \quad s^* = |\gamma^*|,  
\end{equation}
where $\Cbeta > 0$ is some threshold.  Covariates in $\gamma^*$ are called ``influential'', and we assume $s^* \leq s_0$ so that identification of $\gamma^*$ is possible. 
We are interested in a high-dimensional asymptotic regime where $p$ can grow much faster than $n$ but $s_0 \log p = O(n)$;  $s_0, s^*, \gamma^*, \Cbeta$ are also allowed to depend on $n$. 

Assume the entries of $\beta^*$ corresponding to non-influential covariates are very close to zero so that $y - \X_{\gamma^*} \beta^*_{\gamma^*}$ can be effectively treated as the noise. Denote by $\ty = \X_{\gamma^*} \beta^*_{\gamma^*}$  the signal part of $y$.  
For an overfitted model $\gamma$, 
the variation of the signal $\ty$ is fully explained, so non-influential covariates in $\gamma$ tend to be ``useless'':  they may happen to explain some noise, but the degree can be controlled by concentration inequalities and mild eigenvalue conditions on the design matrix $\X$. 
Provided that the penalty on the model size is sufficiently large, we should be able to remove  non-influential covariates from an overfitted model. 
If $\gamma$ is underfitted (i.e., $\gamma^* \not\subseteq \gamma$), the analysis becomes much more complicated due to the correlation among the covariates. 
Suppose for some $j \in \gamma^*$, $|\beta^*_j|$ is very large but $X_j$  is not selected in the model $\gamma$. 
Then, non-influential covariates in $\gamma^c \setminus \gamma^*$ that are slightly correlated with $X_j$ may be added to the model, and similarly, 
influential covariates in $\gamma^* \cap \gamma$ may be removed from $\gamma$ due to correlation with $X_j$. 
However, a known result from forward-backward stepwise selection~\citep{an2008stepwise} guarantees that an underfitted model will eventually become overfitted in order to fully explain the signal,  as long as the sample size $n$ and $\Cbeta$ in~\eqref{eq:beta.min} are sufficiently large relative to the collinearity in $\X$. 

The above reasoning implies that the following condition, under certain mild assumptions,  would hold true with high probability for a sufficiently large sample size,  which was proved in~\citet[Lemma 4]{yang2016computational}.  

\begin{condition}\label{cond:ywj}
There exist $\gamma^* \in \cM(s_0)$ and constants $c_0, c_1 > 0$ (not depending on $\gamma$) such that the following three conditions are satisfied. (We say $\gamma$ is overfitted if $\gamma^* \subseteq \gamma$, and underfitted if $\gamma^* \not\subseteq \gamma$.)
\begin{enumerate}[({1}a)] 
\item For any overfitted $\gamma \in \cM(s_0)$ and $j \notin \gamma$,  $\PR(\gamma, \gamma\cup\{j\}) \leq p^{- c_0}$.  \label{c1}
\item For any underfitted $\gamma \in \cM(s_0)$,  there exists some $j \in \Ga$, which may not be unique, such that $\PR(\gamma, \gamma\cup\{j\}) \geq p^{c_1}$.  \label{c2}
\item For any underfitted $\gamma$ with $|\gamma| = s_0$,  there exist some $j \in \Ga$ and $k \in \Gb$, which may not be unique, such that $\PR(\gamma, (\gamma \cup \{j\}) \setminus \{k\}) \geq p^{c_1}$.   \label{c3}  
\end{enumerate}
\end{condition}

Though in the statement of Lemma 4 of~\citet{yang2016computational}, they set $c_0 = 2$ and $c_1 = 3$, their argument actually proved Condition~\ref{cond:ywj} for at least $c_0 = 2$ and $c_1 = 4$, which suffices for the analysis to be conducted in later sections. Indeed, by modifying the universal constants  in their assumptions,  the  same argument can  prove the claim for  any fixed $c_0, c_1 > 0$; see Section~\ref{sec:cond.ywj} in the supplement, where we state this result as Theorem~\ref{th:ywj} and provide a sketch of the proof. 
In the proof, we treat the design matrix as fixed and do not make assumptions on how the columns of $\X$ are generated. In particular, the design matrix may include interaction terms which can account for potentially non-linear relationship between the  response and explanatory variables. 

If Condition~\ref{cond:ywj} holds, given any $\gamma \neq \gamma^*$, we can increase the posterior probability by a local move: if $\gamma$ is overfitted, we can remove a non-influential covariate by Condition~\cond{c1}; if $\gamma$ is underfitted, we can find an addition or swap move according to Condition~\cond{c2} and~\cond{c3}. Thus, Condition~\ref{cond:ywj} essentially assumes that $\PP$ is unimodal on $\cM(s_0)$ with respect to the add-delete-swap neighborhood relation; see the definition below.  

\begin{definition}
Given a function $\cN \colon \cM(s_0) \rightarrow 2^{\cM(s_0)}$,  we say $\gamma$ is a local mode (w.r.t. $\cN$) if $\PP(\gamma) \geq \max_{\gamma' \in \cN(\gamma)} \PP(\gamma)$, and we say $\PP$ is unimodal (w.r.t. $\cN$)  if there is only one local mode w.r.t. $\cN$. 
\end{definition}

Another important consequence of Condition~\ref{cond:ywj} is that, as long as  $c_1 > 2 c_0 > 2$, tails of $\PP$ ``decay fast'', since for any integer $k \geq 1$, we have $\PP( \cS_k ) \leq p^{1 - c_0} \PP(\cS_{k - 1} )$ where $\cS_k = \{ \gamma \in \cM(s_0) \colon  | \gamma \triangle \gamma^*| = k \}$ denotes the set of all models that have a Hamming distance of $k$ from $\gamma^*$.  This fact is a byproduct of the rapid mixing proof and implies that as $n$ tends to infinity, $\PP(\gamma^*) \rightarrow 1$ in probability with respect to the true data-generating probability measure, a property that is often known as  strong model selection consistency and has been proved for other spike-and-slab priors~\citep{narisetty2014bayesian}. 
To prove the rapid mixing of \RW{},  we define an operator $\cT \colon \cM(s_0) \rightarrow \cM(s_0)$ such that $\cT(\gamma^*) = \gamma^*$, and for any $\gamma \in \cM(s_0) \setminus \{ \gamma^* \}$, $\cT(\gamma) \in  \cN(\gamma)$ and $\PR(\gamma,  \cT(\gamma)) \geq p^{ c_0 \wedge c_1}$.  
Then, as shown in~\citet{yang2016computational}, one can construct a canonical path ensemble, which yields a bound on the spectral gap of the transition matrix of the  \RW{} chain~\citep{diaconis1991geometric, sinclair1992improved}. 
It is noteworthy that $\PP$ can still be highly ``irregular'' under Condition~\ref{cond:ywj} in the sense that its $p$ coordinates may have a very complicated dependence structure due to the collinearity in the design matrix.

\subsection{The \LOT{} algorithm}\label{sec:informed}
The proposal distribution of the \RW{} algorithm is not ``informed'' in the sense that it is constructed without using  information from $\PP$. 
But as explained in Section~\ref{sec:intro.lit},  a naive informed proposal scheme may lead to worse performance due to exceedingly small acceptance probabilities.  

We consider a more general setup where the proposal weighting can be performed for each type of proposal separately.  By  modifying the transition matrix in~\eqref{eq:KRW}, define $\KLOT \colon \cM(s_0) \times \cM(s_0) \rightarrow [0, 1]$ by 
\begin{equation}\label{eq:KI}
\begin{aligned}
   \KLOT(\gamma, \gamma') =\;& \sum_{\star = \text{`a', `d', `s'}} \frac{ \hstar(\gamma)  \wstar(\gamma' \mid \gamma) }{ \Zstar(\gamma)}   \ind_{\nstar(\gamma)}(\gamma') , \\
   \Zstar(\gamma) =\;& \sum_{\tgamma \in \nstar(\gamma) }  \wstar(\tgamma \mid \gamma), 
\end{aligned}
\end{equation}
where $\wstar(\gamma' \mid \gamma) \in [0, \infty)$ denotes the proposal weight of $\gamma' \in \nstar(\gamma)$ given current state $\gamma$. 
In words, we first sample the type of move with probabilities given by $\hadd(\gamma), \hdel(\gamma)$ and $\hswap(\gamma)$. 
If an addition move is to be proposed, we sample a state  $\gamma' \in \adds(\gamma)$ with weight $\wadd(\gamma' \mid \gamma)$. 
We propose to use 
\begin{equation}\label{eq:def.weight}
\wstar (\gamma' \mid \gamma) =    p^{\ell_\star}  \vee \PR(\gamma, \gamma')  \wedge p^{L_\star}, \quad  \text{ for }\star = \text{`a', `d', `s'}, 
\end{equation}
where $-\infty \leq  \ell_\star \leq  L_\star \leq \infty$ are some constants that may depend on the type of move. 
This proposal scheme has two desirable properties. First, states with larger posterior probabilities are more likely to be proposed. Second, for any $\gamma' \in \nstar(\gamma)$, we can bound its proposal probability from below by 
\begin{align*}
\KLOT(\gamma, \gamma') =  \frac{ \hstar(\gamma) \wstar(\gamma' \mid \gamma ) }{\Zstar(\gamma)}
\geq \frac{   \hstar(\gamma) }{  |\nstar(\gamma)| }p^{\ell_\star - L_\star}. 
\end{align*}   
More generally, these two properties still hold if we replace $\PR(\gamma, \gamma')$ in~\eqref{eq:def.weight} with $f( \PR(\gamma, \gamma') )$ for any monotone function $f \colon (0, \infty) \rightarrow (0, \infty)$; some related discussion will be given in Section~\ref{sec:zanella}.  

Calculating  the normalizing constant $\Zstar$ requires evaluating $\PR(\gamma, \gamma')$ for every $\gamma' \in \nstar(\gamma)$. We use the method described in~\citet[Supplement B]{zanella2019scalable}, and we note that the Cholesky decomposition of $\X_\gamma^\top \X_\gamma$ can be obtained by efficient updating algorithms, which only have complexity $O(|\gamma|^2)$~\citep{smith1996nonparametric, george1997approaches}.  
Assuming that $\X^\top \X$ and $\X^\top y$ are precomputed, the complexity of each addition or deletion move has complexity $O(p|\gamma|^2)$ for \LOT{} and complexity $O(|\gamma|^2)$ for \RW{}. We will discuss the implementation of swap moves in Section~\ref{sec:swap}.   
For extremely large $p$, one may first use marginal regression (i.e., simple linear regression of $y$ against $X_j$ for each $j$) to select a subset of  potentially influential covariates~\citep{fan2008sure}. Denoting this subset by $S$,  we then replace the weighting function $\wadd$ in~\eqref{eq:def.weight} by 
\begin{equation}\label{eq:def.twadd}
\begin{aligned}
\tilde{w}_{\rm{a}} (\gamma' \mid \gamma) = \left\{ 
\begin{array}{cc}
 p^{\ell_{\rm{a}} } \vee \PR(\gamma, \gamma')   \wedge p^{L_{\rm{a}}},  \quad &  \text{ if } \gamma' \in \adds(\gamma), \, \gamma'\setminus \gamma \in  S, \\  
 p^{\ell_{\rm{a}}}, &  \text{ if } \gamma' \in \adds(\gamma), \, \gamma'\setminus \gamma \notin  S. \\  
\end{array}
\right. 
\end{aligned}
\end{equation} 
The function $\tilde{w}_{\rm{s}}$ can be  defined similarly. Note that the calculation of $\wdel$ is much easier since $|\dels(\gamma)| = |\gamma| \leq s_0$. 
In Section~\ref{sec:gwas}, we will see that such a practical implementation of the \LOT{} algorithm works well for a real data set with $p = 328,129$. 

\section{Dimension-free mixing of \LOT{} }\label{sec:ss.two}
In this section, we prove that, if the parameters of \LOT{} are properly chosen, the algorithm can achieve a dimension-free mixing rate  when $\PP$ satisfies Condition~\ref{cond:ywj} (the actual data-generation mechanism and the interpretation of $\gamma^*$ as the set of all influential covariates as given in~\eqref{eq:beta.min} are  irrelevant to our proof as long as Condition~\ref{cond:ywj} is satisfied).  
To simplify the analysis, we  only consider swaps when $|\gamma| = s_0$. 
For any $\gamma$ with $|\gamma| < s_0$, with probability $1/2$ we propose to add a covariate, and with probability $1/2$ we propose to remove one; that is, we let $\hadd(\gamma) = \hdel(\gamma) = 1/2$,  $\hswap(\gamma) = 0$ for  the proposal matrix $\KLOT$ given in~\eqref{eq:KI}.
If $|\gamma| = s_0$, we let $\hswap(\gamma) = \hdel(\gamma) = 1/2$ and $\hadd(\gamma) = 0$. 
Thus, $\KLOT$ can be written as 
\begin{equation}\label{eq:def.K}
\begin{aligned}
\KLOT(\gamma,    \gamma') = \frac{\wadd( \gamma \mid \gamma')}{2 \Zadd(\gamma)} \ind_{  \adds(\gamma)  }(\gamma')   +  \frac{ \wdel( \gamma \mid \gamma') }{2 \Zdel(\gamma)} \ind_{  \dels(\gamma)  }(\gamma')   \quad &  \text{ if } |\gamma| < s_0,  \\
\KLOT(\gamma,  \gamma')  = \frac{ \wswap( \gamma \mid \gamma') }{2 \Zswap(\gamma)} \ind_{  \swaps(\gamma)  }(\gamma')   +  \frac{ \wdel( \gamma \mid \gamma') }{2 \Zdel(\gamma)} \ind_{  \dels(\gamma)  }(\gamma') ,  \quad&   \text{ if } |\gamma| = s_0. 
\end{aligned}
\end{equation}
This is different from  the symmetric \RW{} algorithm, where the probability of proposing a deletion move is only $O(s_0 / p)$ since states in $\adds(\gamma) \cup \dels(\gamma)$ are proposed randomly with equal probability. 

For the weighting functions $\wadd, \wdel, \wswap$, we assume 
\begin{equation}\label{eq:def.weights}
\begin{aligned}
\wadd(\gamma' \mid \gamma) = \;&  \PR(\gamma, \gamma')  \wedge p^{c_1},  \\ 
\wdel(\gamma' \mid \gamma)  = \;&   1 \vee \PR(\gamma, \gamma')   \wedge p^{c_0},  \\  
\wswap(\gamma' \mid \gamma)  = \;&  p s_0 \vee \PR(\gamma, \gamma')   \wedge p^{c_1}, 
\end{aligned}
\end{equation} 
where $\PR$ is the posterior probability ratio given in~\eqref{eq:ppr} and $c_0, c_1$ are constants given in Condition~\ref{cond:ywj}. Clearly, they are special cases of the general form given in~\eqref{eq:def.weight}. 
Other choices of the threshold values may also yield the same mixing rate. 
For example, one may use $\wadd(\gamma' \mid \gamma) = \PR(\gamma, \gamma')$,  and the proof will be essentially the same.  However, for $\wdel$ and $\wswap$, the use of  two-sided thresholds is critical. 
 
\subsection{Two-stage drift condition for \LOT{}}\label{sec:drifts}
Let $\KLOT$ be as defined in~\eqref{eq:def.K} and~\eqref{eq:def.weights}  and $\MHP$ denote the corresponding transition matrix, which is given by 
\begin{equation}\label{eq:MH}
\MHP( \gamma, \gamma') = \left\{\begin{array}{cc}
\KLOT(\gamma, \gamma') \min \left\{ 1, \, \frac{\PP(\gamma' ) \KLOT(\gamma', \gamma)}{ \PP(\gamma) \KLOT(\gamma, \gamma')} \right\}, & \text{ if } \gamma' \neq \gamma,  \\
1 - \sum_{\gamma' \neq \gamma} \MHP(\gamma, \gamma'),  & \text{ if } \gamma' = \gamma. 
\end{array}
\right.
\end{equation}
For any function $f$,  let $(\MHP  f )(\gamma) = \sum_{\gamma'} f(\gamma') \MHP(\gamma, \gamma') $. 
If for some set $\sA \subset \cM(s_0)$, function $V\colon  \cM(s_0) \rightarrow [1, \infty)$ and constant $\lambda \in (0, 1)$,
\begin{equation}\label{eq:def.drift.lit}
  (\MHP V)(\gamma) \leq \lambda V(\gamma), \quad \forall \,  \gamma \in \sA. 
\end{equation}
we say the \LOT{} chain satisfies a drift condition on $\sA$, which implies that  the entry time of the \LOT{} chain into $\sA^c$ has a ``thin-tailed'' distribution (see Lemma~\ref{lm:drift0}). 
 
To analyze the convergence rate of the \LOT{} algorithm, we will establish two drift conditions, one for underfitted models and the other for overfitted models.  The two conditions jointly imply that, if initialized at an underfitted model, the \LOT{} chain tends to first move to some overfitted model and then move to  $\gamma^*$. 
Then,  general results for the two-stage drift condition to be proved in Section~\ref{sec:main} (see Theorem~\ref{th:drift} and Corollary~\ref{coro:hd})  can be used to derive a bound on the mixing time of \LOT{}. 
Define 
\begin{equation*}\label{eq:def.O}
\cO = \cO(\gamma^*, s_0) = \{ \gamma \in \cM(s_0) \colon  \gamma^* \subseteq \gamma \}, 
\end{equation*} 
which denotes the set of all overfitted models in $\cM(s_0)$. 
The two drift functions we choose are given by 
\begin{equation}\label{eq:v12}
V_1 (\gamma) =    \left( 1 +  \frac{    y^\top  \OPJ_\gamma   y }{  g^{-1} y^\top y } \right)^{1/ \log(1 + g)},  \quad \quad 
V_2 (\gamma) =  e^{| \gamma \setminus \gamma^*| / s_0}, 
\end{equation}
where we recall   $1 + g = p^{2 \kappa_1}$ defined in~\eqref{eq:priors}.  
If the current model $\gamma \notin \cO$, we expect that $V_1(\gamma)$ tends to decrease in the next iteration  since some covariates can be added to explain the variation of the signal. 
If $\gamma \in \cO \setminus \{\gamma^*\}$, $V_2(\gamma)$ tends to decrease   since non-influential covariates in $\gamma$ can be removed. For convenience, we introduce the notation 
\begin{equation*} 
R_i(\gamma, \gamma') = \frac{V_i(\gamma')}{V_i(\gamma)} - 1, \quad i = 1, 2. 
\end{equation*}
We summarize the properties of functions $V_1, V_2, R_1, R_2$ in the following lemma.
\begin{lemma}\label{lm:R1}
Assume $s_0 \geq 1$. For any $\gamma, \gamma' \in \cM(s_0)$, the following statements hold. 
\begin{enumerate}[(i)]
\item $1 \leq V_1 (\gamma) \leq e$ and $1 \leq V_2 (\gamma) \leq e$. \label{lm1.1}
\item For any  $j \notin \gamma$,  $R_1(\gamma, \gamma\cup\{j\}) \leq 0$; 
for any $k \in \gamma$, $R_1(\gamma, \gamma \setminus \{k\}) \geq 0$.  \label{lm1.2}
\item For any $j \in (\gamma \cup \gamma^*)^c$ and $k \in \Gb$,  \label{lm1.3}
\begin{align*}
R_2(\gamma, \gamma \cup \{j\}) \leq \frac{2}{s_0}, \quad  R_2(\gamma, \gamma \setminus \{k\}) \leq - \frac{1}{2 s_0}. 
\end{align*} 
\end{enumerate}
\end{lemma}
\begin{proof}
See Section~\ref{proof.lemmas}. 
\end{proof}

Since $\MHP(\gamma, \gamma) = 1 - \sum_{\gamma' \neq \gamma} \MHP(\gamma, \gamma')$, some algebra yields that, for $i = 1, 2$, 
\begin{equation}\label{eq:pv2}
\begin{aligned}
\frac{ (\MHP V_i) (\gamma) }{ V_i(\gamma)}  
=\;&    1 + \sum_{\gamma' \neq \gamma} R_i (\gamma, \gamma') \MHP(\gamma, \gamma'),  \\
=\;&   1 + \sum_{\star = \text{`a', `d', `s'}} \, \sum_{\gamma' \in \nstar(\gamma)} R_i (\gamma, \gamma') \MHP(\gamma, \gamma'). 
\end{aligned}
\end{equation}
Therefore, we only need to bound the sum of  $R_i (\gamma, \gamma') \MHP(\gamma, \gamma')$ for three types of proposals separately. 
Since by~\eqref{eq:def.K}, the proposal probability of any move is bounded by $1/2$,  we have, for any $\gamma' \neq \gamma$,  
\begin{equation}\label{eq:mh2}
\MHP( \gamma, \gamma')  
= \min \left\{ \KLOT(\gamma, \gamma'),  \,  \PR(\gamma, \gamma') \KLOT(\gamma', \gamma) \right\} \leq  \PR(\gamma, \gamma')/2. 
\end{equation}  

Consider the case of overfitted models first.   Let $\gamma \in \cM(s_0)$ be overfitted. 
By Condition~\cond{c1}, if we remove any non-influential covariate from $\gamma$, we will get a model with a much larger posterior probability; by Condition~\cond{c2}, if we remove any influential covariate from $\gamma$, we will get an underfitted model with a much smaller posterior probability. 
As a result,  when a deletion move is proposed, the covariate to be removed will be non-influential with high probability. The resulting change in $V_2$ can be bounded using Lemma~\ref{lm:R1}\eqref{lm1.3}. Note that we also need to bound the probability of adding this covariate back so that we can show the acceptance probability of the desired deletion move is large. 
The case of adding a non-influential covariate is easier to handle; one just needs to use Lemma~\ref{lm:R1}\eqref{lm1.3} and the inequality in~\eqref{eq:mh2}. 
The argument for swap moves is essentially a combination of those for addition and deletion moves. 
The bounds we find for the summation term in~\eqref{eq:pv2} are given in Lemma~\ref{lm:overfit.all}, from which we obtain the drift condition for overfitted models. 

\begin{proposition}\label{th:overfit} 
Suppose Condition~\ref{cond:ywj} holds for some $c_0 \geq 2$ and $c_1 \geq 1$. 
For any overfitted model $\gamma$ such that $\gamma \neq \gamma^*$ and $|\gamma| \leq s_0$, 
\begin{align*}
\frac{ (\MHP V_2) (\gamma) }{ V_2(\gamma)}   = 1 - \frac{1}{4 s_0} + O\left(  \frac{1}{p s_0} \right).
\end{align*}
\end{proposition}
\begin{proof} 
It follows from~\eqref{eq:pv2} and the bounds provided in Lemma~\ref{lm:overfit.all}. 
\end{proof}

\begin{remark}\label{rmk:over.sharp}
By Proposition~\ref{th:overfit} and Lemma~\ref{lm:R1}\eqref{lm1.1},  if we consider the \LOT{} chain restricted to the set $\cO$, the mixing time has order at most $s_0$. 
The order of this bound is sharp. 
Consider the worst case where $\gamma^* = \emptyset$ and $| \gamma | = s_0$. Then we need approximately $2 s_0$ steps to remove all the covariates in $\gamma$.   
\end{remark}
 
Next, consider the set of all underfitted models. 
Comparing the expression of $V_1$ in~\eqref{eq:v12} with that of $\PP$ in~\eqref{eq:post}, we see that a lower bound on $\PR(\gamma, \gamma')$ can yield an upper bound on $R_1(\gamma, \gamma')$. This is proven in Lemma~\ref{lm:ineq.R1} in Section~\ref{sec:proof.sel}. 
Just like in the analysis of overfitted models, we will bound $R_1 (\gamma, \gamma') \MHP(\gamma, \gamma')$ for three types of proposals separately.  
In particular, by Lemma~\ref{lm:R1}\eqref{lm1.2}, we need to bound the increase in $V_1$ when we remove any covariate and show that the expected decrease in $V_1$ is sufficiently large when we use the addition move (or swap move). 
However, the calculation is much more complicated than in the overfitted case. 
By Condition~\cond{c2}, we know that there exists at least one model in $\adds(\gamma)$ which has a much smaller value of $V_1$; 
denote this model by $\gamma \cup \{j^*\}$.
But in an extreme case, we may have $\PR(\gamma,  \gamma \cup \{j\} ) \geq p^{c_1}$ for every $j \notin \gamma$. This happens when, for some $j^* \in [p]$,  $|\beta^*_{j^*}|$ is extremely large and every non-influential covariate is slightly correlated with $X_{j^*}$. 
Hence,  the proposal probability,  $\KLOT(\gamma, \gamma \cup \{j^*\})$, may be as small as $O ( p^{-1} )$, and if we only consider the best addition move, the  bound on the mixing time will have a factor of $p$. 
This is the main reason why the path method used by~\citet{yang2016computational} is unable to yield a dimension-free mixing time bound for~\LOT{}. 
To overcome this problem, we will directly bound the sum of $R_1 (\gamma, \gamma') \MHP(\gamma, \gamma')$ over all possible addition moves and take into account   ``good'' moves other than $\gamma \cup \{j^*\}$. 
The same technique is also needed for the analysis of swap moves. 
The following proposition gives the drift condition for underfitted models,  where we recall $\kappa = \kappa_0 + \kappa_1$. 

\begin{proposition}\label{th:underfit}  
Suppose that  $n = O(p)$, $s_0 \log p = O(n)$, $\kappa = O(s_0)$,  and Condition~\ref{cond:ywj} holds for some  $c_1$ such that 
$$  (c_0 + 1) \vee 4 \leq c_1 \leq n \kappa_1 - \kappa. $$ 
For any  underfitted model $\gamma \in \cM(s_0)$, 
\begin{align*}
\frac{ (\MHP V_1) (\gamma) }{ V_1(\gamma)}   \leq  1 - \frac{ c_1  }{8 n \kappa_1}  + o \left( \frac{1}{n \kappa_1} \right). 
\end{align*} 
\end{proposition}
\begin{proof} 
It follows from~\eqref{eq:pv2} and Lemma~\ref{lm:underfit.all}. 
\end{proof}
  
\subsection{Mixing time of the \LOT{} algorithm}\label{sec:mix}
The remaining challenge is to find a mixing time bound for the \LOT{} chain by combining the two drift conditions derived in Propositions~\ref{th:overfit} and~\ref{th:underfit}. 
This is a very interesting problem in its own right and will be investigated in full generality in the next section. 
Applying Corollary~\ref{coro:hd} (which will be presented in Section~\ref{sec:hd}), we find the following mixing time bound for the \LOT{} algorithm.  

\begin{theorem}\label{th:mh.mix}
Consider the Markov chain \LOT{} defined by~\eqref{eq:def.K},~\eqref{eq:def.weights} and~\eqref{eq:MH} with stationary distribution $\PP$ given in~\eqref{eq:post}. 
Define the mixing time of \LOT{} by
\begin{align*}
\Tmix = \sup_{\gamma \in \cM(s_0)} \min\{t \geq 0:  \TV{ \MHP^t( \gamma, \cdot) -   \PP (\cdot ) }  \leq 1/4 \}, 
\end{align*}
where $\TV{\cdot}$ denotes the total variation distance. 
Suppose that $n = O(p)$, $s_0 \log p = O(n)$ and $\kappa = O(s_0)$.  If Condition~\ref{cond:ywj} holds for $c_0 = 2$ and $4 \leq c_1 \leq n \kappa_1 - \kappa$,  then, for sufficiently large $n$, we have 
$$\Tmix \leq  800   \max\left\{ \frac{    n \kappa_1  }{ c_1   }, \; 3  s_0      \right\}.$$
\end{theorem}

\begin{proof}
See Section~\ref{proof.mix}. 
\end{proof}

\begin{remark}\label{rmk:mix}
The assumptions we have made in Theorem~\ref{th:mh.mix} are mild and are essentially the same  as those of~\citet{yang2016computational}. 
First, it is known that $s_0 \log p = O(n)$ is a necessary condition for estimation consistency in high-dimensional sparse regression models; \citet{yang2016computational} assumed that $s_0 \log p \leq n / 32$. 
Second, the prior parameter choice $\kappa = O(s_0)$ is very reasonable. Indeed, if $\kappa$ grows faster than $n / \log p$, for the data-generating model considered in Section~\ref{sec:ywj}, the threshold $\Cbeta$ in~\eqref{eq:beta.min}  needs to go to infinity for consistent model selection, which would be of little  interest in most applications.  
Note that we can always let $\kappa_1$ be a fixed positive constant, and then the mixing time bound in  Theorem~\ref{th:mh.mix} is  at most $O(n)$. 
One can even prove Condition~\ref{cond:ywj} for some $c_1$ growing with $n$ (e.g. by letting $\Cbeta$ in~\eqref{eq:beta.min} be sufficiently large), 
in which case it is possible for our bound to only grow at rate $s_0$. 
For comparison, the upper bound on the mixing time of the symmetric \RW{} algorithm given in~\citet[Theorem 2]{yang2016computational} is $O( p   s_0^2 (n \kappa_1 + s_0 \kappa_1 +  s_0 \kappa_0 ) \log p )$.  
\end{remark}

\section{General results for the two-stage drift condition}\label{sec:main}
For the \LOT{} algorithm, we have established two drift conditions, one for underfitted models and the other for overfitted models. In this section, we derive some general results for using such a two-stage drift condition to bound the mixing time of a Markov chain (not necessarily the \LOT{} chain),  which we denote by $(\sX_t)_{t \in \bbN}$ where $\bbN = \{0, 1, 2, \dots \}$.   
We only need to require the following assumption on  $(\sX_t)_{t \in \bbN}$, and note that the underlying state space may not be discrete. 

\renewcommand{\theassumption}{\Alph{assumption}}
\begin{assumption}\label{ass:mc}
$(\sX_t)_{t \in \bbN}$ is a Markov chain defined on a state space $(\cX, \cE)$ where the $\sigma$-algebra $\cE$ is countably generated. The transition kernel $\bP$  is reversible with respect to a stationary distribution $\pi$, and $\bP$ has non-negative spectrum.  
\end{assumption}

\begin{remark}
There is little loss of generality by assuming reversibility and non-negative spectrum for MCMC algorithms. 
First, both Metropolis--Hastings and random-scan Gibbs algorithms (in the classical sense) are always reversible, though some non-reversible versions have been proposed in recent years~\citep{bierkens2016non, fearnhead2018piecewise, bouchard2018bouncy, gagnon2019non, bierkens2019zig}. 
Second, for any transition kernel $\bP$,  its lazy version $\bP_{\rm{lazy}} = (\bP + \bI)/2$ always has non-negative spectrum. 
As noted in~\citet{baxendale2005renewal}, these two assumptions can yield better bounds on the convergence rates. 
\end{remark}

For any non-negative measurable function $f$,  let $(\bP^t f )(x) = \E_x[  f(\sX_t) ]$, where $\E_x$ denotes the expectation with respect to the probability measure for $(\sX_t)_{t \in \bbN}$ with   $\sX_0  = x$. 
For a non-empty measurable set $\sC \subset \cX$, we say  $(\sX_t)$ satisfies a drift condition on $\cX \setminus \sC$ if 
\begin{equation}\label{eq:def.drift}
  (\bP V)(x) \leq \lambda V(x), \quad \forall \,  x \notin \sC, 
\end{equation}
for some  function $V\colon  \cX \rightarrow [1, \infty)$ and constant $\lambda \in (0, 1)$. 
If  $\sC = \{ x^* \}$ is a singleton set, \eqref{eq:def.drift} will be referred to as a ``single element'' drift condition.   This happens when the state $x^*$ has a large stationary probability mass and the chain has a tendency to move towards $x^*$.   

\subsection{Convergence rates with the two-stage drift condition}\label{sec:two}
Motivated by the variable selection problem, we consider a setting where $(\sX_t)_{t \in \bbN}$ satisfies two ``nested'' drift conditions. Let $\sA$ be a measurable subset of $\cX$ and $x^*$ be a point in $\sA$. 
First, we use a drift condition on  $\sA^c$ to describe the tendency of the chain to move towards $\sA$, if it is currently outside $\sA$. Second, we assume once the chain enters $\sA$, it drifts towards $x^*$, which can be described by a single element drift condition on $\sA \setminus \{x^*\}$.  
We refer to such a construction as a two-stage drift condition, for which the main result is provided in the following theorem. 
  
\begin{theorem}\label{th:drift}
Let $(\sX_t)_{t \in \bbN}, \cX, \bP, \pi$ be as given in Assumption~\ref{ass:mc}. 
Suppose that there exist two drift functions $V_1, V_2\colon \cX \rightarrow [1, \infty)$,  
constants $\lambda_1, \lambda_2 \in (0, 1)$,  a set $\sA  \in \cE$ and a point $x^* \in \sA$ such that 
\begin{enumerate}[(i)]
    \item $\bP V_1 \leq \lambda_1 V_1$ on $\sA^c$, and \label{d1}
    \item $\bP V_2 \leq \lambda_2 V_2$ on $\sA\setminus \{x^*\}.$ \label{d2}
\end{enumerate}
Further, suppose that   $\sA$ satisfies the following conditions for some finite constants $M \geq 2$ and $K\geq 1$. 
\begin{enumerate}[(i)]\setcounter{enumi}{2}
\item For any $x \in \sA$,  $V_1(x) = 1$, and if $\bP(x, \sA^c) > 0$, $\E_x [ V_1(\sX_1) \mid  \sX_1 \in \sA^c ]  \leq  M/2$. \label{d3}
\item For any $x \in \sA$,  $V_2(x) \leq K$, and if $\bP(x, \sA^c) > 0$, $\E_x [ V_2(\sX_1) \mid \sX_1 \in \sA^c  ] \geq V_2(x)$.  \label{d4}
\item  For any $x \in \sA$,   $\bP(x, \sA^c)  \leq q$ for some constant $q <  \min\{ 1 - \lambda_1, \, (1 - \lambda_2) / K \}$.  \label{d5}
\end{enumerate}
Then, for every $x \in \cX$ and $t \in \bbN$, we have 
\begin{equation*}\label{eq:tv}
\TV{ \bP^t (x, \cdot ) - \pi}   \leq 4  \alpha^{  t + 1  }     \left( 1  + \frac{  V_1(x) }{  M } \right), 
\end{equation*}
where $\alpha$ is a constant in $(1 - q/4, 1)$ and can be computed by 
\begin{align*}
\alpha = \frac{1 + \rho^r}{2} = \frac{1 + M^r / u}{2},  \quad \rho = \frac{q K}{1 - \lambda_2},  \quad   u = \frac{1}{1 - q/2}, \quad r = \frac{ \log u}{ \log (M / \rho) }. 
\end{align*}
\end{theorem}
\begin{proof}
See Section~\ref{sec:proof}. 
\end{proof}

\begin{remark}\label{rmk:t1}
To interpret the two drift conditions, (i) and (ii), it may help to  think of $ \log V_1$  as the   ``distance''  to the set $\sA$, and $ \log V_2$  as the  ``distance''  to the point $x^*$. 
Both conditions (iii) and (iv) then become natural. 
Indeed, there is no loss of generality by assuming that $V_1 = 1$ on $\sA$.
Given any other drift function $V_1'$ which satisfies (i) on $\sA^c$, we can always define $V_1$ by letting $V_1 = 1$ on $\sA$ and $V_1 = V_1'$ on $\sA^c$,  which still satisfies all the assumptions made in the theorem. 
The constant $M$ can be simply chosen to be $2 \sup_{x \in \cX} V_1(x)$, if it is finite. 
\end{remark}

\begin{remark}\label{rmk:absorb}
Consider the distribution of the hitting time $\tau^* = \min\{t \geq 0 \colon \sX_t = x^*\}$, which, by Theorem~\ref{th:j2}, can be used to bound the mixing time of the chain. 
The sample path from an arbitrary point $x \in \sA^c$ to $x^*$ can be broken into disjoint segments in $\sA^c$ and $\sA$. Though  the length of each segment has a finite expectation due to the two drift conditions,  the number of these segments largely depends on the parameter $q$, and $\E[ \tau^*]$ may be infinite if the chain can easily escape from the set $\sA$. This is why we need condition (v).  
Consider some $x \in \sA$ and $y \in \sA^c$ such that $\bP(x, y) > 0$.
For Markov chains that move locally, $\bP(x, y) > 0$ implies that $x$ and $y$ are very ``close''. 
By the reversibility of the chain, we have $\bP(x, y) \leq \pi(y) / \pi(x)$. Hence, to check condition (v), it suffices to bound the ratio $\pi(y) / \pi(x)$, which is often straightforward for neighboring states $x, y$.
When applying Theorem~\ref{th:drift}, we should be careful with the choice of $q$. 
Even if $\bP(x, \sA^c) = 0$ for all $x \in \sA$, we should try other positive values of $q$ so that $1 - \alpha$ can be maximized. 
As a rule of thumb, we can choose some $q$  that has the same order as $\min\{ 1 - \lambda_1, \, (1 - \lambda_2) / K \}$; 
see Corollary~\ref{th:hd}. 
\end{remark}

\subsection{Convergence rates in the high-dimensional setting}\label{sec:hd}
For MCMC algorithms,  $\cX$ is the parameter space and its dimension is conventionally denoted by $p$. 
For high-dimensional problems, $p = p(n)$ grows to infinity and typically,  we have the drifting parameters $\lambda_1, \lambda_2 \uparrow 1$ and the convergence rate $(1 - \alpha) \downarrow 0$, where $\lambda_1, \lambda_2, \alpha$ are as given in Theorem~\ref{th:drift}. 
To show the chain is rapidly mixing, we need to find a finite constant $c > 0$ such that $p^{-c} = O(1 - \alpha)$. 
The following result extends Theorem~\ref{th:drift} to the high-dimensional setting. 

\begin{corollary}\label{th:hd}
Consider a sequence of Markov chains  where each $(\sX_t)_{t \in \bbN}$ (implicitly indexed by $n$) satisfies the assumptions in Theorem~\ref{th:drift}. 
Assume that $\lambda_1, \lambda_2 \rightarrow 1$  and  $ q \leq \min\{ 1 - \lambda_1,   (1 - \lambda_2) / C K\}$ for some universal constant  $C > 1$. 
Then, we have 
\begin{equation*} 
\TV{ \bP^t (x, \cdot ) - \pi}   \leq 4  \alpha^{ t + 1}     \left(1  + \frac{   V_1(x) }{  M } \right), 
\end{equation*}
for some $\alpha$ such that  ($\sim$ denotes asymptotic equivalence)
\begin{equation*}
1 - \alpha \sim  \frac{ (1 - \lambda^*) \log C }{4  \log (M C)}, \quad \text{ where }1 - \lambda^* =  \min \left\{ 1 - \lambda_1,  \; \frac{ 1 - \lambda_2}{CK}  \right\}. 
\end{equation*}
\end{corollary}

\begin{proof}
Observe that without loss of generality we can assume  $1 - \lambda_1 = (1 - \lambda_2)/CK = 1 - \lambda^*$ and $q = 1 - \lambda^* = o(1)$. 
Then, the constants defined in Theorem~\ref{th:drift} satisfy $\rho = C^{-1}$ and $r \sim q / (2 \log (MC) )$,
from which the result follows. 
\end{proof}

\begin{corollary}\label{coro:hd}
For $\epsilon \in (0, 1/2)$,  define the $\epsilon$-mixing time of the chain $(\sX_t)_{t \in \bbN}$ by 
$$\Tmix(\epsilon) = \sup_{x \in \cX} \min\{t \geq 0:  \TV{ \bP^t( x, \cdot) -   \pi (\cdot ) }  \leq \epsilon \}.$$ 
In the setting of Corollary~\ref{th:hd} with $M = 2 \sup_{x \in \cX} V_1(x)$,  for sufficiently large $n$, we have 
\begin{align*}
\Tmix(\epsilon) \lesssim \frac{  4 \log (6 / \epsilon)   }{ \log C }  \log (C M)
\max\left\{ \frac{1}{1 - \lambda_1}, \; \frac{C K}{1 - \lambda_2}  \right\}. 
\end{align*}
\end{corollary}

\begin{proof} 
This follows from a  straightforward calculation using    $- \log( \alpha) \sim 1 - \alpha$.  
\end{proof}

\begin{remark}
In Corollaries~\ref{th:hd} and~\ref{coro:hd},  we do not make assumptions on the growth rates of $M$ and $K$.  In particular, if $M = p^c$ for some constant $c \geq 0$, it will only introduce an additional $\log p$ factor to the the mixing time.  
\end{remark}

\subsection{Comparison with drift-and-minorization methods}
The two-stage drift condition can be seen as a generalization of the single element drift condition since eventually the chain will arrive at the central state $x^*$.  
But, from a different angle, it  also resembles the classical drift-and-minorization methods, which assume that 
  there exist a drift function $V \colon \cX \rightarrow [1, \infty)$, a ``small'' set $\sS \in \cE$, a probability measure $\psi$ on $(\cX, \cE)$, constants $\lambda \in (0, 1)$, $\xi  > 0$ and $b < \infty$ such that 
\begin{align*}
\text{(drift condition) } \quad  & \bP V  \leq \lambda V \ind_{\sS^c} + b \ind_{\sS},   \\ 
\text{ (minorization condition) }  \quad  & \bP(x,  \cdot) \geq \xi \psi(\cdot) \, \text{ for } x \in \sS.   
\end{align*}
Both coupling arguments and regeneration theory can be used to compute a bound on $\TV{ \bP^t(x, \cdot) - \pi }$; see, for example,~\citet{rosenthal1995minorization, roberts1999bounds}. 
By the minorization condition,  each time the chain is in $\sS$, we can let the whole process regenerate according to $\psi$ with probability $\xi$. By the drift condition,   the return times into $\sS$ have geometrically decreasing tails. Jointly, the two conditions imply that the first time that the chain regenerates has a ``thin-tailed'' distribution.  
The proof for our result with the two-stage drift condition uses a  similar  idea. 
The set $\sA$ in Theorem~\ref{th:drift} can be seen as the small set, and each time the chain enters $\sA$, there is some positive probability  that the chain will hit $x^*$ and thus regenerates before leaving $\sA$. 
Essentially,  we have  replaced the minorization condition on the small set with another drift condition, which is still used to bound the regeneration probability when the chain visits the small set. 

The above comparison between the two-stage drift condition and drift-and-minorization method suggests that one may want to consider the following more general setting.  
Suppose there exist a sequence of sets $\cX = \sA_0 \supseteq \sA_1 \supseteq \cdots \supseteq \sA_k$ such that for each $i = 0, 1, \dots, k - 1$, a drift condition holds on $\sA_i \setminus \sA_{i+1}$ showing that the chain tends to drift from  $\sA_i \setminus \sA_{i+1}$ into $\sA_{i + 1}$. 
When $\sA_k$ is a singleton set, one can mimic the proof of Theorem~\ref{th:drift} to combine the $k$ drift conditions and derive a quantitative bound on the mixing time. 
When $\sA_k$ is a set on which we can establish a minorization condition or we have a mixing time bound for the Markov chain restricted to $\sA_k$, the main idea of our proof still applies, though some details may need nontrivial modification. 
For a concrete example, consider the posterior distribution of $\beta$ in our variable selection problem described in Section~\ref{sec:model}.  
To sample from $\PP(\beta)$, we just need to modify any MH algorithm targeting $\PP(\gamma)$ by sampling  from the conditional posterior distribution $\PP(\beta \mid \gamma)$ at the end of each iteration. 
Let $\norm{\beta}_0$ denote the number of nonzero entries of $\beta$. 
Then, assuming Condition~\ref{cond:ywj} holds, we can define $\cA_0 = \{ \beta \in \bbR^p \colon \norm{\beta}_0 \leq s_0 \}$, 
$\cA_1 = \{\beta \in \cA_0 \colon \forall  j \in \gamma^*, \, \beta_j \neq 0  \}$ (the set of all possible values of $\beta$ for an overfitted model), 
and $\cA_2 = \{ \beta \in \cA_1 \colon \forall j \notin \gamma^*, \, \beta_j = 0  \}$ (the set of all possible values of $\beta$ for the model $\gamma^*$). 
The two drift conditions proved in Section~\ref{sec:drifts} show that the MH algorithm tends to drift from $\cA_0 \setminus \cA_1$ into $\cA_1$  and from $\cA_1 \setminus \cA_2$ into $\cA_2$. By combining them with standard results in the literature for the mixing time of an MH algorithm targeting $\PP(\beta \mid \gamma = \gamma^*)$ (which is just a multivariate normal distribution), one can derive the mixing time bound for the MH chain. 

\subsection{Applications of the two-stage drift condition}\label{sec:disc.two.stage}
Our use of the two-stage drift condition in Section~\ref{sec:ss.two} is largely motivated by Condition~\ref{cond:ywj}, which characterizes the different behaviors of underfitted and overfitted models. Though this makes the two-stage drift condition look very specific to the variable selection problem, there are actually many discrete-state-space problems other than variable selection where the mixing time of MH algorithms can be conveniently analyzed by using multiple drift conditions. 
 
First, many model selection problems can be written as a set of sparse linear regression models. For example, in structure learning problems, the goal is to infer the underlying Bayesian network (i.e., directed acyclic graph) of a $p$-variate distribution. These problems are often formulated as structural equation models where the causal relationships among $p$ coordinate variables are described by $p$ sparse linear regression models. 
When the ordering is known, structure learning of Bayesian networks becomes very similar to variable selection, and one can extend and prove Condition~\ref{cond:ywj} in a way completely analogous to that in~\citet{yang2016computational}. The application of the two-stage drift condition and the construction of \LOT{} are also straightforward. 
When the ordering is unknown, the problem becomes much more complicated due to the existence of Markov equivalent Bayesian networks, but it is still possible to generalize Condition~\ref{cond:ywj}~\citep{zhou2021complexity}, and the two-stage drift condition can be constructed accordingly. 

Moreover, in recent years,  general high-dimensional consistency results have been obtained for a large variety of Bayesian model selection problems. \citet{gao2020general} proved optimal posterior contraction rates for a general class of structured linear models, including as special cases stochastic block model,   multi-task learning, dictionary learning and wavelet estimation. All these examples share common features with the variable selection problem; for example, one can naturally define a model to be underfitted or overfitted according as it has the best model nested within it.  It seems very promising that the methodology developed in~\citet{gao2020general} can be used to prove results similar to Condition~\ref{cond:ywj} in general settings, from which we may further establish a two-stage drift condition, one for underfitted models and the other for overfitted ones. 
 
In addition to ``underfitted/overfitted'' schemes, for some problems, we may partition the state space using a different strategy. For example, consider a change-point detection problem where we need to infer both the number and locations of change points. Suppose that we use an informative prior which favors models with equal-sized segments. Then, if we incorrectly infer the number of change points, the locations of change points cannot be accurately estimated either due to the prior. In such cases, a possible approach to mixing time analysis is to construct one drift condition showing that the chain first drifts towards models with true number of change points and another drift condition showing that once the number of change points is correctly inferred, the chain is able to tune the locations of change points towards their true values. 
We note that similar ideas may also be applied to more complicated spatial clustering models, such as those based on spanning trees~\citep{luo2021bayesian, lee2021t}.
Compared with using a single drift condition on the whole space, the two-stage approach often leads to an easier and more constructive proof.

\section{Simulation studies}\label{sec:sim}
For our simulation studies, we implement the \RW{} and \LOT{} algorithms as follows. Assume the proposal scheme has the form given in~\eqref{eq:KI}.   
In each iteration we propose an addition, deletion or swap move with fixed probabilities $0.4, 0.4$ and $0.2$,  respectively; that is, we set $\hadd(\gamma) = \hdel(\gamma) = 0.4$, $\hswap(\gamma) = 0.2$ in~\eqref{eq:KI}. This is slightly different from the setting in Section~\ref{sec:ss.two}, but our mixing time bound still applies up to some constant factor. 
We consider four choices of the weighting functions $\wadd$ and $\wdel$.  
\begin{align*}
\begin{array}{lll}
\text{\RW{}:}   &  \wadd( \gamma' \mid \gamma) = 1, & \wdel(\gamma' \mid \gamma) = 1. \\  
\text{\LOT{}-1:}   &  \wadd( \gamma' \mid \gamma)  = p^{-1}  \vee \PR(\gamma, \gamma') \wedge p, &  \wdel(\gamma' \mid \gamma) =  p^{-1}  \vee \PR(\gamma, \gamma') \wedge 1 , \\  
\text{\LOT{}-2:}  \; &  \wadd( \gamma' \mid \gamma)  = p^{-2}  \vee \PR(\gamma, \gamma') \wedge p^2, \; &  \wdel(\gamma' \mid \gamma) =  p^{-2}  \vee \PR(\gamma, \gamma') \wedge p, \\
\text{\LIB{}-1:}   &  \wadd( \gamma' \mid \gamma)  = \sqrt{  \PR(\gamma, \gamma')}, &  \wdel(\gamma' \mid \gamma) =  \sqrt{  \PR(\gamma, \gamma')}. 
\end{array}
\end{align*} 
\LOT{}-2 is   more aggressive   than \LOT{}-1 in the sense that the proposal distribution is more concentrated on the neighboring states with very large posterior probabilities.  
The last one is inspired by the locally balanced proposals of~\citet{zanella2020informed}. Since the proposal weights are unbounded in \LIB{}-1, we expect its acceptance probability to be smaller than that of \LOT{} algorithms.  
For comparison, we also consider the original locally balanced MH algorithm of~\citet{zanella2020informed} with proposal 
\begin{align*}
    \KZ(\gamma, \gamma') = \frac{ \sqrt{  \PR(\gamma, \gamma')}  }{ Z(\gamma) } \ind_{\adds(\gamma) \cup \dels(\gamma)}(\gamma'), \text{ where } Z(\gamma) =  \sum_{\tgamma \in \adds(\gamma) \cup \dels(\gamma) } \sqrt{  \PR(\gamma, \tgamma)}.
\end{align*}
Denote this algorithm by \LIB{}-2. It differs from \LIB{}-1 in that we do not distinguish types of proposal moves when calculating proposal weights. 
We will discuss \LIB{}-1 and \LIB{}-2 in Sections~\ref{sec:zanella} and~\ref{sec:compare.lb}.

The use of the parameter $s_0$ is unnecessary in our simulation studies since all sampled models have size much smaller than $n$. 
For computational convenience, we do not consider swap moves for \LIB{}-2 and, for the other algorithms, we implement swap moves by compounding one addition and one deletion move, which makes the swap proposal only ``partially informed''; details are given in Section~\ref{sec:swap}. When we describe the modality of $\PP$ in simulation results, we are always referring to the ``single-flip'' neighborhood relation $\adds(\cdot) \cup \dels(\cdot)$. 
The code is written in \texttt{C++} in order to maximize the computational efficiency; see Section~\ref{sec:data}. 

\subsection{Finding models with high posterior probabilities}\label{sec:sim1}
For the first simulation study, we consider the settings used in~\citet{yang2016computational} with random design matrices. Let all rows of $\X$ be i.i.d., and the $i$-th row vector, $x_{(i)}$, be generated in the following two ways. 
\begin{align*}
\text{Independent design: }  &  x_{(i)} \overset{i.i.d.}{\sim} \N(0, I_p), \\
\text{Correlated  design: }  &  x_{(i)} \overset{i.i.d.}{\sim} \N(0, \Sigma),  \, \Sigma_{jk} = e^{-|j - k|}. 
\end{align*}
The response vector $y$ is simulated by $y = X \beta^* + z$ with $z \sim \N(0, I_n)$. 
The first 10 entries of $\beta^*$ are given by 
\begin{align*}
\beta^*_{[10]} = \mathrm{SNR} \sqrt{ \frac{\log p}{ n } } (2, -3, 2, 2, -3, 3, -2, 3, -2, 3),  
\end{align*}
where $\mathrm{SNR} > 0$ denotes the signal-to-noise ratio.  All the other entries of $\beta^*$ are set to zero. 
The un-normalized posterior probability of a model $\gamma$ is calculated using~\eqref{eq:post} with $\kappa_0 = 2$ and $\kappa_1 = 3/2$. 
Simulation experiments are conducted for  $\rm{SNR} = 1, 2, 3$ and $(n, p) = (500, 1000)$ or $(1000, 5000)$. 
For each setting, we simulate 100 data sets, and for each data set, we run \RW{} for $10^5$ iterations and each informed algorithm for $2,000$ iterations. All algorithms are initialized with a randomly generated model $\gamma^{(0)}$ with $|\gamma^{(0)}| = 10$. 
Let $\true = \{1, 2, \dots, 10\}$ denote the true set of covariates with nonzero regression coefficients, and let $\hat{\gamma}_{\rm{max}}$ be the model with the largest posterior probability that has been sampled by any of the five algorithms. 
If an algorithm has never sampled $\hat{\gamma}_{\rm{max}}$, the run is counted as a failure. 

\begin{table} 
\caption{\label{table1} Simulation study I.  For each setting, we simulate 100 data sets. ``Time'' is the average wall time usage measured in seconds.  ``Success'' is the number of successful runs; a run is successful if $\hat{\gamma}_{\rm{max}}$ has been sampled (see the main text). $H_{\rm{max}}$ is the median number of iterations needed to sample $\hat{\gamma}_{\rm{max}}$; the number in the parenthesis is the 95\% quantile. $t_{\rm{max}}$ is the median wall time (in seconds) needed to sample $\hat{\gamma}_{\rm{max}}$. 
The number in parentheses shown below the SNR value is the number of replicates out of 100 where $\hat{\gamma}_{\rm{max}} = \true$.}
\centering
{\scriptsize
\begin{tabular}{cccccccc}
\toprule
 &  &  &  \RW{} & \LOT{}-1 & \LOT{}-2 & \LIB{}-1 & \LIB{}-2 \\
\multicolumn{3}{c}{Number of iterations} & 100,000 & 2,000 & 2,000 & 2,000 & 2,000 \\
\midrule 
\multirow{12}{*}{\shortstack{$n = 500$,\\ $p = 1000$, \\ independent \\ design}} & \multirow{4}{*}{ \shortstack{ SNR = 3 \\ (100) } }  & Time  & 7.8 & 1.3 & 1.3 & 1.2 & 1.2\\
  &  & Success  & 100 & 100 & 100 & 100 & 12\\
  &  & $H_{\rm{max}}$  & 8004(13773) & 20(30) & 210(578) & 39(53) & 2000(2000+)\\
  &  & $t_{\rm{max}}$  & 0.62 & 0.012 & 0.13 & 0.023 & 1.1\\
 \cmidrule{2-8}
 & \multirow{4}{*}{ \shortstack{ SNR = 2 \\ (100) } }  & Time  & 7.8 & 1.3 & 1.3 & 1.2 & 1.6\\
  &  & Success  & 100 & 100 & 100 & 100 & 95\\
  &  & $H_{\rm{max}}$  & 9316(17884) & 20(30) & 60(385) & 38(53) & 94(1972)\\
  &  & $t_{\rm{max}}$  & 0.73 & 0.013 & 0.038 & 0.023 & 0.077\\
 \cmidrule{2-8}
 & \multirow{4}{*}{ \shortstack{ SNR = 1 \\ (0) } }  & Time  & 5.3 & 0.96 & 0.95 & 0.9 & 1.2\\
  &  & Success  & 100 & 100 & 100 & 100 & 100\\
  &  & $H_{\rm{max}}$  & 33(9326) & 22(35) & 21(39) & 21(36) & 9(15)\\
  &  & $t_{\rm{max}}$  & 0.0014 & 0.010 & 0.0095 & 0.0088 & 0.0060 \\
\midrule
\multirow{12}{*}{\shortstack{$n = 500$,\\ $p = 1000$, \\ correlated \\ design}} & \multirow{4}{*}{ \shortstack{ SNR = 3 \\ (98) } }  & Time  & 7.7 & 1.3 & 1.3 & 1.2 & 1.2\\
  &  & Success  & 89 & 99 & 98 & 98 & 4\\
  &  & $H_{\rm{max}}$  & 17230($10^5+$) & 29(726) & 50(390) & 45(689) & 2000(2000+)\\
  &  & $t_{\rm{max}}$  & 1.3 & 0.018 & 0.032 & 0.028 & 1.1\\
 \cmidrule{2-8}
 & \multirow{4}{*}{ \shortstack{ SNR = 2 \\ (42) } }  & Time  & 7.4 & 1.2 & 1.2 & 1.1 & 1.2\\
  &  & Success  & 57 & 79 & 81 & 85 & 57\\
  &  & $H_{\rm{max}}$  & 62308($10^5+$) & 94(2000+) & 72(2000+) & 124(2000+) & 1200(2000+)\\
  &  & $t_{\rm{max}}$  & 4.5 & 0.061 & 0.046 & 0.083 & 0.84\\
 \cmidrule{2-8}
 & \multirow{4}{*}{ \shortstack{ SNR = 1 \\ (0) } }  & Time  & 3.7 & 0.8 & 0.79 & 0.76 & 2\\
  &  & Success  & 100 & 100 & 100 & 100 & 100\\
  &  & $H_{\rm{max}}$  & 26(4722) & 22(32) & 21(30) & 22(34) & 9(10)\\
  &  & $t_{\rm{max}}$  & 0.00074 & 0.0085 & 0.0084 & 0.0081 & 0.010\\
\midrule
\multirow{12}{*}{\shortstack{$n = 1000$,\\ $p = 5000$, \\ independent \\ design}} & \multirow{4}{*}{ \shortstack{ SNR = 3 \\ (100) } }  & Time  & 39 & 9.2 & 9.3 & 8.9 & 13\\
  &  & Success  & 99 & 100 & 100 & 100 & 4\\
  &  & $H_{\rm{max}}$  & 35291(65746) & 20(32) & 36(53) & 43(53) & 2000(2000+)\\
  &  & $t_{\rm{max}}$  & 14 & 0.091 & 0.17 & 0.19 & 12\\
 \cmidrule{2-8}
 & \multirow{4}{*}{ \shortstack{ SNR = 2 \\ (100) } }  & Time  & 39 & 9.2 & 9.4 & 8.9 & 13\\
  &  & Success  & 98 & 100 & 91 & 100 & 93\\
  &  & $H_{\rm{max}}$  & 39746(70458) & 20(31) & 469(2000+) & 33(48) & 160(2000+)\\
  &  & $t_{\rm{max}}$  & 15 & 0.089 & 2.2 & 0.15 & 1.1\\
 \cmidrule{2-8}
 & \multirow{4}{*}{ \shortstack{ SNR = 1 \\ (0) } }  & Time  & 34 & 8.5 & 8.6 & 8.3 & 8.3\\
  &  & Success  & 97 & 100 & 100 & 100 & 100\\
  &  & $H_{\rm{max}}$  & 20140(76960) & 21(44) & 20(44) & 20(43) & 12(22)\\
  &  & $t_{\rm{max}}$  & 7.5 & 0.089 & 0.086 & 0.081 & 0.056\\
\midrule
\multirow{12}{*}{\shortstack{$n = 1000$,\\ $p = 5000$, \\ correlated \\ design}} & \multirow{4}{*}{ \shortstack{ SNR = 3 \\ (100) } }  & Time  & 41 & 10 & 10 & 9.9 & 8.9\\
  &  & Success  & 94 & 100 & 99 & 100 & 0\\
  &  & $H_{\rm{max}}$  & 51178($10^5+$) & 20(32) & 65(1181) & 41(54) & 2000(2000+)\\
  &  & $t_{\rm{max}}$  & 21 & 0.10 & 0.33 & 0.20 & 8.7\\
 \cmidrule{2-8}
 & \multirow{4}{*}{ \shortstack{ SNR = 2 \\ (84) } }  & Time  & 41 & 10 & 10 & 9.8 & 10\\
  &  & Success  & 42 & 89 & 94 & 97 & 20\\
  &  & $H_{\rm{max}}$  & $10^5+$($10^5+$) & 37(2000+) & 64(2000+) & 50(1194) & 2000(2000+)\\
  &  & $t_{\rm{max}}$  & 38 & 0.19 & 0.33 & 0.24 & 8.5\\
 \cmidrule{2-8}
 & \multirow{4}{*}{ \shortstack{ SNR = 1 \\ (0) } }  & Time  & 20 & 7.2 & 7.2 & 7 & 16\\
  &  & Success  & 100 & 100 & 100 & 100 & 100\\
  &  & $H_{\rm{max}}$  & 25(13474) & 23(36) & 22(33) & 22(32) & 9(10)\\
  &  & $t_{\rm{max}}$  & 0.0041 & 0.082 & 0.078 & 0.076 & 0.069\\
\bottomrule
\end{tabular}
}
\end{table}

Results are summarized in Table~\ref{table1}. 
We first note that when $\rm{SNR} = 3$, \LIB{}-2 performs much worse than \RW{} and almost always gets stuck at some sub-optimal local mode. This is consistent with the observation made in Example~\ref{ex1}, which will be further discussed in Section~\ref{sec:zanella}. Due to its poor performance, we exclude \LIB{}-2 from all the remaining numerical experiments. 
For $\rm{SNR}=2$ and $3$, informed algorithms always find the model $\hat{\gamma}_{\rm{max}}$ much faster than \RW{}. Remarkably, the median wall time needed for \LOT{}-1 to sample $\hat{\gamma}_{\rm{max}}$ (denoted by $t_{\rm{max}}$ in the table) is less than $0.2$ second in all scenarios. 
When $\rm{SNR}=1$, the best model is often the null model, in which case \RW{} can also find $\hat{\gamma}_{\rm{max}}$ easily (since we fix $\hdel(\gamma) = 0.4$, it only takes \RW{} about 25 iterations to propose removing the 10 covariates in $\gamma^{(0)}$). 
When the signal-to-noise ratio is either very strong ($\rm{SNR} = 3$) or very weak ($\rm{SNR} = 1$), all algorithms  except \LIB{}-2 can identify $\hat{\gamma}_{\rm{max}}$ in most runs. Similar findings were made in~\citet{yang2016computational}, and this is because when SNR is very large,  Condition~\ref{cond:ywj} is likely to be satisfied  with $\gamma^* = \true$, and when SNR is very small, Condition~\ref{cond:ywj} is likely to be satisfied  with $\gamma^* = \emptyset$.  
For correlated designs with $\rm{SNR} = 2$, the posterior landscape tends to be  multimodal, and all algorithms may get stuck at local modes. But the informed algorithms still have much better performance than \RW{}: each informed algorithm,  except \LIB{}-2, is able to sample $\hat{\gamma}_{\rm{max}}$ in $\geq 80\%$ of the runs, while \RW{} has a much larger failure rate and finds $\hat{\gamma}_{\rm{max}}$ much more slowly. 
We also observe that in most settings, \LOT{}-1 is (significantly) more efficient than \LOT{}-2 and \LIB{}-1. This suggests that it is helpful to truncate the weighting function $\wadd$ and $\wdel$  to a relatively small range, which is consistent with our theory. 
Since $|\gamma^{(0)} | = |\true| = 10$ (and the two sets are disjoint in most cases),  it takes at least 20 addition and deletion proposals to move from $\gamma^{(0)}$ to $\true$. 
According to Table~\ref{table1}, in high SNR settings,  $\hat{\gamma}_{\rm{max}}$ usually coincides with  $\true$, and the median number of iterations needed by \LOT{}-1 to reach $\hat{\gamma}_{\rm{max}}$ is about $20$ (except in the correlated design case with $n=500$, $p=1000$),  suggesting that the performance of \LOT{}-1 is close to being ``optimal'' in the sense that any other local MH sampler cannot find $\hat{\gamma}_{\rm{max}}$ in a smaller number of iterations.

Since to implement informed MCMC algorithms, we need to evaluate $\PP$ for every possible addition or deletion move in each iteration, we can use the MCMC sample paths to empirically study to what extent Condition~\ref{cond:ywj} is satisfied. 
As predicted by the theory developed in~\citet{yang2016computational}, we find that Condition~\ref{cond:ywj} is more likely to be violated when SNR is small or the design matrix contains highly correlated covariates. See Sections~\ref{supp:sim.local} and~\ref{supp:cond.ywj} in the supplement for details.

\subsection{Rao-Blackwellization for \LOT{} }\label{sec:sim1.rb}
During MCMC, given the current model $\gamma$, we can estimate $\beta$ using the conditional posterior mean $\hat{\beta}(\gamma) = \bbE[  \beta \mid \gamma, y ]$. 
For \LOT{} algorithms, we can obtain a Rao-Blackwellized estimator, $\hat{\beta}_{\rm{RB}}(\gamma)$,  as follows. 
For each $j \in [p]$, let $\gamma_j = \ind_{ \gamma } (j)$ indicate whether covariate $j$ is selected in $\gamma$, and let $\gamma_{-j} = (\gamma_1, \dots, \gamma_{j-1}, \gamma_{j+1}, \dots, \gamma_p)$ denote the status of all the other $p-1$ covariates.  
By the law of total expectation, $\bbE [\beta_j \mid y] = \bbE[  \, \bbE[ \beta_j \mid \gamma_{-j}, y ] \,  ]$. 
For informed proposal schemes, we can get $ \bbE[ \beta_j \mid \gamma_{-j}, y ]$ for every $j \in [p]$ with little additional computational cost, since 
\begin{align*}
\bbE[ \beta_j \mid \gamma_{-j}, y ]  =    \frac{ \PP(  \gamma \cup \{j\} ) \, \bbE[ \beta_j \mid \gamma_{-j}, y, \gamma_j = 1]   }{   \PP(  \gamma \cup \{j\} )  +  \PP(  \gamma \setminus \{j\} )  }, 
\end{align*} 
and all the three terms on the right-hand side (one of them is just $\PP(\gamma)$)  have already been obtained when we calculate the normalizing constants,  $\Zadd(\gamma)$ and $\Zdel(\gamma)$. 
The estimator $\hat{\beta}_{\rm{RB}}(\gamma)$ is then obtained by estimating  each entry  using $ \bbE[ \beta_j \mid \gamma_{-j}, y ]$.  

Let $\mathrm{MSE}(\hat{\beta}) = p^{-1}  \norm{\hat{\beta} - \beta^*}_2^2$  denote the error of an estimator $\hat{\beta}$. For the simulation study described in Section~\ref{sec:sim1}, we observe that $\mathrm{MSE}(\hat{\beta}_{\rm{RB}}(\gamma))$   always decreases to a nearly optimum level within a few iterations. 
See Figure~\ref{figure1} for the trajectories of $\mathrm{MSE}(\hat{\beta}_{\rm{RB}}(\gamma))$ of  \LOT{}-1 and $\mathrm{MSE}(\hat{\beta}(\gamma))$ of \RW{}, averaged over 100 data sets. (Trajectories of \LOT{}-2 and \LIB{}-1 are omitted since they are very similar to that of \LOT{}-1.) 
The advantage of  \LOT{} over \RW{} becomes even more substantial. More investigation is needed to justify the use of $\hat{\beta}_{\rm{RB}}(\gamma)$, but our analysis at least shows that the computation of the proposal weights can be made use of in multiple ways. 

\begin{figure} 
\begin{center}
\includegraphics[width=0.9\linewidth]{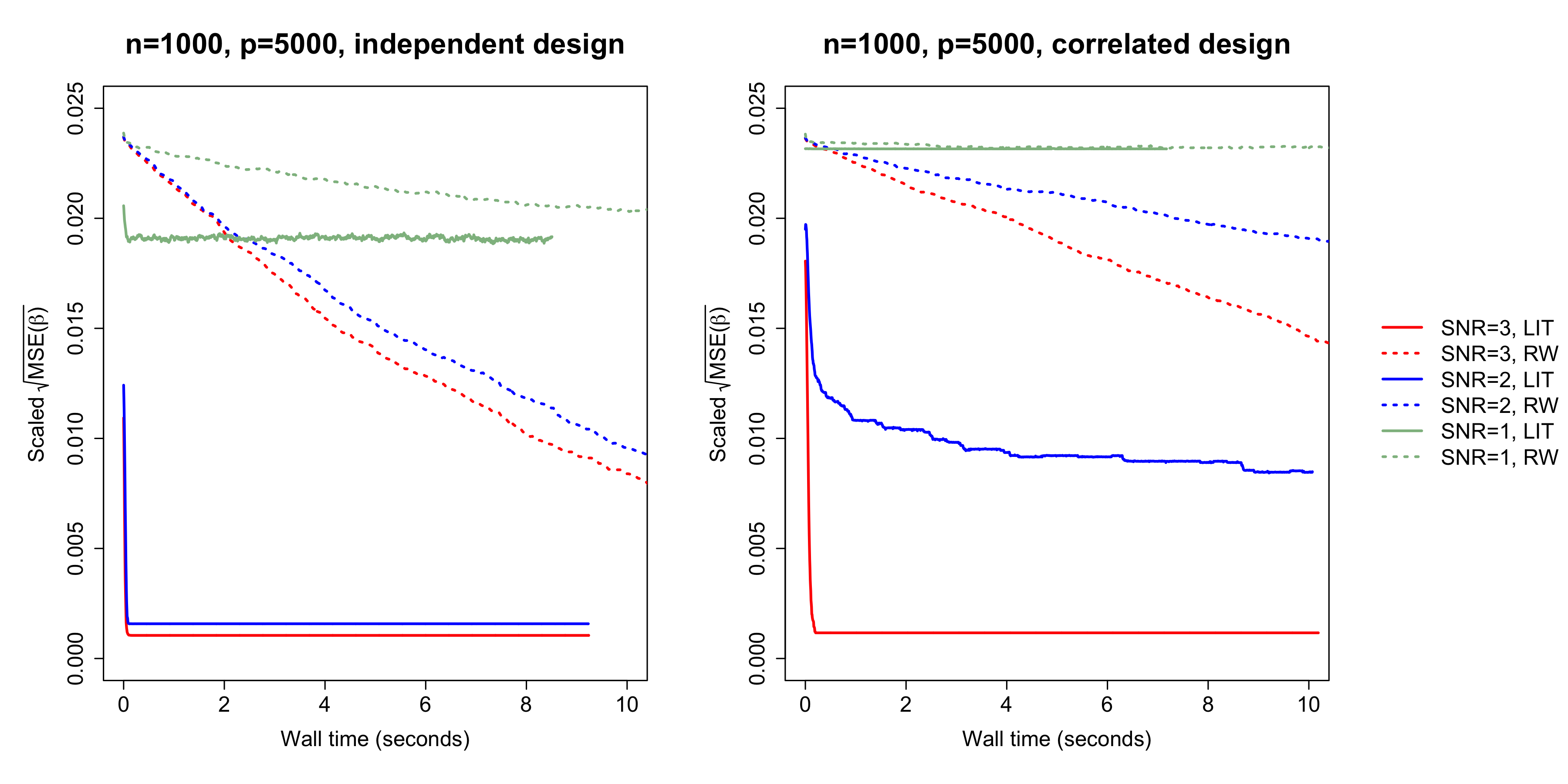}  
\caption{\label{figure1} Trajectories of $\mathrm{MSE}(\hat{\beta})$.  Solid lines represent the \LOT{}-1 algorithm, which uses $\hat{\beta} = \hat{\beta}_{\rm{RB}}(\gamma)$; dotted lines represent the \RW{} algorithm, which uses $\hat{\beta} = \hat{\beta} (\gamma)$. The y-axis is $  \{ \rm{MSE} ( \hat{\beta} ) \}^{1/2} / \rm{SNR}$ averaged over 100 data sets.}
\end{center}
\end{figure}

\subsection{Exploring multimodal posterior distributions}\label{sec:sim.ess}

\begin{table} 
 \caption{\label{table2} Simulation study II.  ``Mean model size'' is the mean size of $\gamma$ sampled by \LOT{}-1 (which is almost the same as that for the other three samplers). 
 ``Time'' is the wall time usage measured in seconds. ``Local modes'' is the number of unique local modes  sampled in the MCMC; we say $\gamma$ is a local mode if   $\PP(\gamma) > \PP(\gamma')$ for any $\gamma'  \in \adds(\gamma) \cup \dels(\gamma)$. 
 ``Acc. rate'' is the acceptance rate.  ESS($T_1$)  and ESS($T_2$) are the estimated effective sample sizes of $T_1$ and $T_2$ (see the main text for details). All statistics are averaged over 20 data sets. }
 {\scriptsize
\begin{tabular}{cccccc}
\toprule
  &  &  \RW{} & \LOT{}-1 & \LOT{}-2 & \LIB{}-1   \\
& Number of iterations &  200,000 & 2,000 & 2,000 & 2,000 \\ 
\midrule
 \multirow{5}{*}{ \shortstack{ $\sigma_\beta = 0.1$ \\ Mean model size = 6.1 } }& Time  & 78.1 & 9.95 & 10.0 & 9.30\\ 
  & Local modes  & 1.85 & 2.30 & 2.15 & 2.00\\ 
  & Acc. Rate  & 0.012 & 0.495 & 0.566 & 0.228\\ 
  & ESS($T_1$)/Time  & 0.706 & 16.8 & 15.1 & 6.46\\ 
  & ESS($T_2$)/Time  & 4.83 & 34.5 & 29.3 & 11.6\\ 
\cmidrule{1-6} 
 \multirow{5}{*}{ \shortstack{ $\sigma_\beta = 0.2$ \\ Mean model size = 26.4 } }& Time  & 79.1 & 16.0 & 15.9 & 14.6\\ 
  & Local modes  & 2.60 & 6.20 & 6.25 & 3.90\\ 
  & Acc. Rate  & 0.0060 & 0.602 & 0.580 & 0.320\\ 
  & ESS($T_1$)/Time  & 0.414 & 4.72 & 3.76 & 2.40\\ 
  & ESS($T_2$)/Time  & 4.67 & 19.9 & 18.5 & 12.7\\ 
\cmidrule{1-6} 
 \multirow{5}{*}{ \shortstack{ $\sigma_\beta = 0.3$ \\ Mean model size = 50.2 } }& Time  & 80.4 & 27.9 & 27.6 & 24.7\\ 
  & Local modes  & 2.40 & 5.05 & 4.45 & 3.65\\ 
  & Acc. Rate  & 0.0037 & 0.578 & 0.571 & 0.296\\ 
  & ESS($T_1$)/Time  & 0.360 & 2.49 & 2.82 & 1.48\\ 
  & ESS($T_2$)/Time  & 3.57 & 19.8 & 18.1 & 9.79\\ 
\cmidrule{1-6} 
 \multirow{5}{*}{ \shortstack{ $\sigma_\beta = 0.4$ \\ Mean model size = 63.9 } }& Time  & 81.2 & 37.0 & 36.8 & 32.6\\ 
  & Local modes  & 2.00 & 3.85 & 5.20 & 3.65\\ 
  & Acc. Rate  & 0.0027 & 0.541 & 0.546 & 0.261\\ 
  & ESS($T_1$)/Time  & 0.333 & 2.57 & 1.85 & 1.02\\ 
  & ESS($T_2$)/Time  & 3.02 & 17.5 & 14.3 & 7.92\\ 
\cmidrule{1-6} 
 \multirow{5}{*}{ \shortstack{ $\sigma_\beta = 0.5$ \\ Mean model size = 71.6 } }& Time  & 81.8 & 42.5 & 42.7 & 36.8\\ 
  & Local modes  & 1.80 & 2.75 & 2.60 & 2.65\\ 
  & Acc. Rate  & 0.0021 & 0.485 & 0.536 & 0.217\\ 
  & ESS($T_1$)/Time  & 0.526 & 3.62 & 3.28 & 1.36\\ 
  & ESS($T_2$)/Time  & 2.45 & 15.0 & 15.1 & 5.78\\ 
\bottomrule
\end{tabular}
}
\end{table}

Condition~\ref{cond:ywj} represents the ideal case where the posterior distribution is unimodal, but in reality multimodality is the norm. 
In our second simulation study, we consider a more realistic simulation scheme which gives rise to multimodal posterior distributions. 
The design matrix $\X$ is still assumed to have i.i.d. rows, but each row is sampled from $\N(0, \Sigma_{d, p})$ where  $\Sigma_{d, p} = \diag (\Sigma_d, \dots, \Sigma_d)$ is block-diagonal. 
Each block $\Sigma_d$ has dimension $d \times d$, and $(\Sigma_d)_{jk} = e^{-|j - k| / 3}$. 
We fix $n = 1000$, $p = 5000$ and $d = 20$. 
The response $y$ is still simulated by $y = X \beta^* + z$ with $z \sim \N(0, I_n)$. But we generate $\beta^*$ by first sampling $\gamma^*$ from the uniform distribution on the set $\{ \gamma \subset [p] \colon |\gamma| = 100\}$ and then sampling $\beta^*_{\gamma^*} \sim \N(0,  \sigma^2_\beta I_{100} )$. 
We use $\sigma_\beta = 0.1, 0.2, 0.3, 0.4, 0.5$ to simulate posterior distributions with varying degrees of multimodality. For the hyperparameters, we choose $\kappa_0 = 1$ and $\kappa_1 = 1/2$.  We observe that the posterior multimodality is most severe for $\sigma_\beta=0.2$.

For each setting, we simulate 20 data sets, and for each data set, we run \RW{} for $2 \times 10^5$ iterations and each informed algorithm  for $2,000$ iterations.  
All four algorithms are initialized with the model obtained by forward-backward stepwise selection. We use effective sample size (ESS) to measure the sampling efficiency. To calculate ESS, we consider two one-dimensional ``summary statistics''.  Let $(\gamma^{(k)}, \beta^{(k)})$ denote the sample collected in the $k$-th MCMC iteration, where $\beta^{(k)}$ is drawn from the conditional posterior distribution of $\beta$ given $\gamma^{(k)}$.  Let $T_{1}^{(k)} = \norm{ \hat{\beta}(\gamma^{(k)}) - \beta^* }_2^2$ where $\hat{\beta}(\gamma) = \bbE[ \beta \mid \gamma, y],$ and let $T_{2}^{(k)} = \norm{ X \beta^{(k)} }_2^2$.  
Note that  $T_{1}^{(k)}$ only depends on $\gamma^{(k)}$, and thus the ESS of $T_1$ roughly reflects the efficiency  of sampling $\gamma$, while the ESS of $T_2$ reflects the efficiency of sampling $\beta$. 
The ESS estimates presented in Table~\ref{table2} are calculated by using the \texttt{R} package \texttt{coda}~\citep{plummer2006coda}, and we provide nonparametric ESS estimates in Section~\ref{supp:sim.ess} in the supplement. 
Given that the interest is in estimating the posterior means of $\beta$ and $\gamma$, following~\citet{vats2019multivariate}, we may calculate a multivariate ESS by appealing to a multivariate central limit theorem using $\{ \gamma^{(k)} \}$ or $\{ \beta^{(k)} \}$. However, the resulting Monte Carlo covariance matrix is quite large, and the estimation of such high-dimensional matrices is an ongoing problem in the literature; see~\citet{jin2021fast}. 
In Section~\ref{supp:sim.ess} in the supplement, we propose a method for constructing a low-dimensional summary of $\gamma$, and present corresponding multivariate ESS estimation results. 

From Table~\ref{table2}, we see that given a fixed total number of iterations,  the acceptance rate of \RW{} decreases quickly for larger $\sigma_\beta$, while it remains roughly unchanged around $0.5$ for \LOT{}  algorithms.  
In all scenarios, \LOT{}-1 and \LOT{}-2 have much larger effective sample sizes (per second) of both statistics $T_1$ and $T_2$ than \RW{}, indicating that \LOT{} algorithms  can explore multimodal  distributions and collect posterior samples much more efficiently than \RW{}. 
Comparing $\rm{ESS}(T_1)$ of \LOT{} and that of \RW{}, we see  that the advantage of using \LOT{} for sampling $\gamma$ is huge for small values of $\sigma_\beta$. 
For \LIB{}-1, we note that it always has smaller acceptance rate and effective sample sizes than the two \LOT{} algorithms. 
This is probably due to the use of an unbounded weighting function, which will be further discussed in Section~\ref{sec:compare.lb}.

\section{Analysis of real GWAS data}\label{sec:gwas}
We have obtained access to two GWAS (genome-wide association study) data sets on glaucoma from dbGaP (the database of Genotypes and Phenotypes) with accession no. phs000308.v1.p1 and phs000238.v1.p1. 
Both data sets only contain individuals of Caucasian descent, and they were generated using the same genotyping array.  
We remove individuals whose self-reported race is Hispanic Caucasian and those with abnormal intraocular pressure or cup-to-disk ratio (CDR) measurements. 
We choose the response variable $y$ to be the standardized CDR measurement averaged over two eyes. 
After merging the two data sets, we discard variants with minor allele frequency less than $0.05$ or missing rate greater than $ 0.01$  and variants that fail the Hardy-Weinberg equilibrium test (p-value less than $10^{-6}$ in control samples). Finally, we use PLINK to prune variants with pairwise correlation $>0.75$ and end up with $n = 5,418$ and $p = 328,129$. Each entry of the matrix $\X$ takes value in $\{0, 1, 2\}$, representing the number of copies of the minor allele.  For the hyperparameters, we choose $\kappa_0$ and $\kappa_1$ such that $g = 100$ and $\pi_0(\gamma) \propto (20 / p)^{|\gamma|}$.  
These choices are motivated by practical considerations. First, the prior on $\gamma$ reflects that a priori we believe there are about 20 variants associated with $y$. For a complex trait such as CDR, this is a very conservative estimate. Second, assuming a causal variant $X_j$ has minor allele frequency $0.5$, $g = 100$ implies that the prior effect size $( g / X_j^\top X_j)^{1/2} \approx 0.2$, which is recommended for Bayesian analysis of GWAS data~\citep{stephens2009bayesian}.

\begin{table} 
 \caption{\label{table3} CDR analysis using GWAS data. $|S_\delta|$ is the number of variants for which we need to evaluate the posterior probabilities when proposing addition moves (see the main text for details). 
``Iterations'' is the number of MCMC iterations for each run; for each algorithm, we conduct 5 independent runs.  ``Time'' is the average wall time usage measured in minutes. ``Acc. rate'' is the average acceptance rate. 
ESS($T_2$) is the effective sample size calculated using the statistic $T_2$ as in Table~\ref{table2}.  
}
\centering
\begin{tabular}{cccccc}
\toprule
Algorithm & $|S_\delta|$ & Iterations & Time & Acc. rate  & ESS($T_2$)/Time \\ 
\midrule 
\RW{}  &  NA &  $1,000,000$  & 428 & 0.273  & 1.95 \\ 
\LOT{}-1 ($\delta = 0.0001$)  & 7255  & 8,000  & 48.7  & 0.714 & 8.52 \\ 
\LOT{}-1 ($\delta = 0.0005$) & 1410  & $40,000$  & 56.3  & 0.635 & 26.8 \\ 
\LOT{}-1 ($\delta = 0.001$)  & 715  & $80,000$ & 82.6   & 0.603  & 33.5 \\ 
\bottomrule
\end{tabular}
\end{table}

\begin{table} 
 \caption{\label{table4} Top 10 signals in the CDR analysis. ``Location'' is the cytogenetic location of the variant in human genome. ``PIP'' is the posterior inclusion probability estimate averaged over all \LOT{}-1 runs. 
``Known hit'' indicates whether the variant is known to be associated with ocular traits; if yes, a reference is provided in the last column.}
\begin{tabular}{ccccc}
\toprule
Variant name &  Location &  PIP &  Known hit  & Reference \\
\midrule
rs1063192 &  9p23.1 & 0.989 &  Yes & \citet{osman2012genome} \\
rs653178 &  12q24.12 & 0.972 & No & \\
rs10483727 & 14q23.1 & 0.888 & Yes  & \citet{bailey2016genome} \\
rs319773 & 17q11.2 & 0.532 & No  & \\
rs2275241 & 9q33.3  & 0.531 & Yes & \citet{craig2020multitrait} \\ 
rs4557053 & 20p12.3 & 0.222 &  No & \\
rs10491971 & 12p13.32 & 0.144 & No & \\
rs4901977 & 14q23.1 & 0.112 &  Yes & \citet{springelkamp2014meta} \\ 
rs587409 & 13q34 & 0.111 & Yes &  \citet{khawaja2018genome} \\
rs314300 & 7q22.1 & 0.107 & No & \\
\bottomrule
\end{tabular}
\end{table}

\begin{table} 
 \caption{\label{table5} Posterior inclusion probabilities of the 5 hits in Table~\ref{table4}. ``PIP'': posterior inclusion probability estimate. The last column gives the minimum and maximum PIP estimates obtained from 15 \LOT{}-1 runs.  }
\begin{tabular}{ccccccc}
\toprule 
 \multirow{2}{*}{Variant name} & \multicolumn{5}{c}{PIPs of \RW{} runs }  & \multirow{2}{*}{\shortstack{PIP range in  15 \\ \LOT{}-1 runs} }  \\ 
 \cmidrule{2-6}
   & Run 1 & Run 2  & Run 3  & Run 4  & Run 5  & \\
\midrule
rs1063192 & 0  & 0.729 & 0.309 & 0.753 & 0.330 & [0.839, 1 ] \\ 
rs10483727 & 0.483 & 0.147 & 0 & 0.852 & 0  & [0.308, 1] \\ 
rs2275241 & 0.383 & 0 & 0  & 0  & 0.426 & [0.491, 0.569] \\ 
rs4901977 & 0.486 & 0.239 & 0.678 & 0  & 0.214 & [0,  0.695] \\ 
rs587409 & 0  & 0.003 & 0.065 & 0.117 & 0.032 & [0.083, 0.132] \\ 
\bottomrule
\end{tabular}
\end{table}

We first conduct 5 parallel runs (with different random seeds) of the \RW{} algorithm, each consisting of 1 million iterations. Then we build a set, denoted by $S_\delta \subseteq [p]$, which includes all variants with posterior probabilities (estimated using the \RW{} output) greater than $\delta$. 
When implementing the addition proposals for \LOT{}, we use~\eqref{eq:def.twadd} with $S = S_\delta$. 
One can also use marginal regression  to build the set $S_\delta$~\citep{fan2008sure}, which would yield very similar results. 
We consider $\delta = 10^{-4}, 5\times 10^{-4}$ and $10^{-3}$. For each choice, we conduct 5 parallel runs of the \LOT{}-1 algorithm. Some summary statistics of the output 
are provided in Table~\ref{table3}.   
For all four algorithms, the mean size of sampled models is $13$. 
According to $\mathrm{ESS}(T_2)$ per minute, \LOT{}-1 algorithms are much more efficient than \RW{} in terms of sampling $\beta$, which may be surprising since \RW{} also has acceptance rate $0.27$ and thence a much larger total number of accepted moves than \LOT{} algorithms. This indicates that \LOT{} can achieve greater sampling efficiency by significantly reducing the autocorrelation in MCMC samples.

Next, we examine the estimate of the posterior inclusion probability (PIP), $\bbE[ \ind_{\gamma}(j) \mid y]$, for each $j \in [p]$. In Table~\ref{table4}, we list the 10 variants with the largest PIPs averaged over all  runs of \LOT{}-1. Among them 5 are known GWAS hits for ocular traits (or ocular disorders) located in 4 different regions. 
For these 5 hits (which we may assume to be true signals), the PIP estimate in each individual run of \RW{} exhibits a very high variability. 
For example, in Table~\ref{table5},  we see that only the 4th run of \RW{} yields a PIP estimate greater than $0.1$ for \textit{rs587409}. Further, if one uses $0.1$ as the threshold, each \RW{} run can miss at least two of the five hits.  This observation suggests that, for large data sets, we often need to run \RW{} for an extremely large number of iterations so that the results can be ``replicable''. 
In contrast, the individual PIP estimates from 15 \LOT{}-1 runs are much more stable. The only exception is the variant \textit{rs4901977}. This is because \textit{rs4901977} is located closely to  \textit{rs10483727}, and thus the two variants are correlated, which makes it challenging to identify both variants at the same time. 

\section{Discussion}\label{sec:discuss}

\subsection{On the swap moves of \LOT{} }\label{sec:swap}
Both the parameter $s_0$ and swap moves are used in our mixing time analysis of \LOT{} for merely technical reasons. As shown in~\citet{yang2016computational}, rapid mixing on the space $\cM(p)$ is usually impossible since sharp local modes can easily occur among very large models, suggesting that the use of $s_0$ is necessary for theoretical analysis. Then, swap moves are introduced to ensure that the chain cannot get trapped at an underfitted model with size $s_0$. 
However, in practice, even if we let $s_0 = p$ and run the chain on $\cM(p)$, the chain is very unlikely to visit those models with size much larger than $s^*$ since they have negligible posterior probabilities~\citep{narisetty2014bayesian}. In other words, assuming that both $s^*$ and $|\gamma^{(0)}|$ are small, Condition~\ref{cond:ywj} actually implies that we will ``observe'' the chain is ``rapidly mixing'' by using only addition and deletion moves. 

The above reasoning suggests that an approximate implementation of informed swap moves will not significantly affect the overall performance of \LOT{}. One way to realize a ``partially informed'' swap move is to treat it as a composition of one addition and one deletion. Given current state $\gamma$, we first use an informed addition move to propose some $\tilde{\gamma} \in \adds(\gamma)$, and then use an informed deletion move to propose $\gamma' \in \dels(\tilde{\gamma})$. The acceptance probability of $\gamma'$ is calculated by
\begin{align*}
    1 \wedge  \frac{\PP(\gamma') \KLOT(\gamma', \tilde{\gamma}) \KLOT(\tilde{\gamma}, \gamma) }{\PP(\gamma) \KLOT(\gamma, \tilde{\gamma}) \KLOT(\tilde{\gamma}, \gamma')  }. 
\end{align*}
One can check that the resulting transition matrix is reversible with respect to $\PP$. In our implementation of \LOT{}, we further impose the constraint that $\gamma' \neq \gamma$ when sampling $\gamma'$ from $\dels(\tilde{\gamma})$ and adjust the Hastings ratio accordingly. 
Note that to implement an addition proposal, we need to calculate both $\KLOT(\gamma, \tilde{\gamma})$ and $\KLOT(\tilde{\gamma}, \gamma)$, which requires evaluating $\PP$ for $p$ models. Similarly, for the deletion proposal, we also need to evaluate $\PP$ for $p$ models.  Hence, in our implementation, each swap proposal involves $2p$ evaluations of $\PP$. This is much more efficient than implementing an informed swap proposal exactly as described in~\eqref{eq:def.K}, which requires evaluating $\PP$ for $2 (p - |\gamma|) |\gamma|$ models.

\subsection{\LOT{} on general discrete state spaces} \label{sec:zanella}
\citet{zanella2020informed} considered ``locally balanced proposals'' for general discrete-state-space problems. Let $\pi$ be a distribution defined on a general discrete state space $\cX$. For each $x$, let $\cN(x) \subset \cX$ denote its neighborhood. 
A locally balanced proposal scheme can be written as
\begin{equation}\label{eq:zanella}
\KZ(x, x') =  \frac{  f \left(  \frac{ \pi (x') }{ \pi (x) } \right) }{ Z_f(x)  }  \ind_{\cN(x)}(x'), \quad  Z_f(x) =  \sum_{y \in \cN(x)}f \left (  \frac{ \pi (y ) }{ \pi (x) } \right), 
\end{equation}
where the ``balancing function'' $f \colon (0, \infty) \rightarrow (0, \infty)$ must satisfy $f(b) = b f(b^{-1})$ for any $b > 0$.  Examples of balancing functions include $f(b) = \sqrt{b}$ and $f(b) = 1 \vee b$.
Consider an MH algorithm with proposal $\KZ$. 
A seemingly desirable property of balancing functions is that the acceptance probability of a proposal move from $x$ to $x'$ is given by 
\begin{equation}\label{eq:general1}
\acc(x, x') =   \min \left\{  1,  \,  \frac{  Z_f(x)  }{ Z_f(x') }   \right\}, 
\end{equation}
for any $x' \in \cN(x)$. If $Z_f(x) \approx Z_f(x')$, this method should work well. Indeed, 
\citet{zanella2020informed} argued that if $  \sup_{x, x'\colon  x' \in \cN(x)} Z_f(x) / Z_f(x') \rightarrow 1$, such a locally balanced proposal is asymptotically optimal. 
But, for problems like variable selection (which was not considered in~\citet{zanella2020informed}), the behavior of the function $x \mapsto Z_f(x)$ is very difficult to predict, and Table~\ref{table1} confirms that for $f(b) = \sqrt{b}$, the informed MH algorithm with proposal~\eqref{eq:zanella} completely fails when the signal-to-noise ratio is sufficiently large. 

Motivated by Condition~\ref{cond:ywj}, consider some $\pi$ that satisfies  the following  condition: there exist $x^* \in \cX$, $\cT \colon \cX \rightarrow \cX$, and $b_0 > 1$ such that for any $x \neq x^*$, $\cT(x) \in \cN(x)$ and $\pi(\cT(x)) / \pi(x) \geq b_0$. 
Define $\omega(x) = \pi(\cT(x)) / \pi(x)$ for each $x \neq x^*$.  
Note that by~\eqref{eq:zanella},  $\KZ(x, \cT(x)) = f(\omega(x)) / Z_f(x)$. 
It follows from~\eqref{eq:general1} that for any $x$ such that $\cT(x) \neq x^*$,   
\begin{equation}\label{eq:general.lit}
    \KZ(x, \cT(x)) \acc(x, \cT(x)) \leq    \frac{  f \left( \omega(x) \right) }{ Z_f( \cT(x) ) } 
    \leq   \frac{ f \left( \omega(x) \right) } { f \left( \omega(\cT(x)) \right) }, 
\end{equation}
since $Z_f(\cT(x)) \geq f ( \omega(\cT(x)) )$. The ratio $ f( \omega(x) ) / f ( \omega(\cT(x)) )$ can be exceedingly small since it is possible that $ \omega( \cT(x) )$ is much larger than $ \omega(x)$.  As we have seen in Example~\ref{ex1}, for variable selection, $\omega(\cT(x)) \gg \omega(x)$ can easily happen if there are correlated covariates and the sample size is large.  
Since collinearity is common for high-dimensional data, this analysis suggests that for general model selection problems, locally balanced MH schemes with proposal given by~\eqref{eq:zanella} may not have good performance when both $n$ and $p$ are large. 

The main idea behind our \LOT{} algorithm can still be applied in this general setting. Pick constants $\overline{f} > \underline{f} > 0$  and modify~\eqref{eq:zanella} by  
\begin{equation*} 
\tilde{\bK}_{\rm{lb}}(x, x') =  \frac{  \underline{f} \vee  f (  \frac{ \pi (x') }{ \pi (x) } )  \wedge \overline{f} }{ \tilde{Z}_f(x)  }  \ind_{\cN(x)}(x'), \quad  \tilde{Z}_f(x) =  \sum_{y \in \cN(x)} \underline{f} \vee  f \left (  \frac{ \pi (y ) }{ \pi (x) } \right)  \wedge \overline{f}. 
\end{equation*}
This modification guarantees that a proposal move from $x$ to $x'$ has acceptance probability  $1$ as long as $\pi(x') / \pi(x)$ is sufficiently large, as shown in the following lemma. 
\begin{lemma}\label{lm:general}
Let $\pi$ be an arbitrary probability distribution on $\cX$. 
Consider the MH algorithm with proposal $\tilde{\bK}_{\rm{lb}}$ given in~\eqref{eq:general.lit} where $f \colon (0, \infty) \rightarrow (0, \infty)$ is an arbitrary non-decreasing function. Suppose there exists $b < \infty$ such that 
 \begin{align*}
    f \left(  b^{-1}  \right) \leq \underline{f},   \text{  and  } b  \geq \frac{ \overline{f} }{\underline{f} } \max_{x \in \cX}|\cN(x)|. 
\end{align*}  
Then, for any $x, x' \in \cN(x)$ such that $\pi(x') / \pi(x) \geq b $, the proposal move from $x$ to $x'$ has acceptance probability $1$. 
\end{lemma} 
\begin{proof}
See Section~\ref{sec:proof.general}. 
\end{proof}

In~\citet{zanella2020informed}, one motivation for the locally balanced proposal was to mimic the behavior of Metropolis-adjusted Langevin algorithms defined on continuous state spaces~\citep{roberts1998optimal}. However, for model selection problems with large sample sizes, the local posterior landscape can change drastically when we move from $x$ to some $x' \in \cN(x)$, which may result in strange behavior of the MH chain (i.e., keep proposing some state $x'$ with $\pi(x') \gg \pi(x)$ and getting rejected). 
One key observation of this work is that once we  truncate the function $f$ in~\eqref{eq:zanella}, the mapping $x \mapsto \tilde{Z}_f(x)$ becomes much ``smoother'' than $x \mapsto Z_f(x)$. 
Since there is almost no difference in computational cost between the two proposals $\KZ$ and $\tilde{\bK}_{\rm{lb}}$, it is apparently  always desirable to use the ``stabilized version'' $\tilde{\bK}_{\rm{lb}}$ in practice. 

\subsection{On the \LIB{}-1 algorithm for variable selection} \label{sec:compare.lb} 
The discussion in Section~\ref{sec:zanella} explains why \LIB{}-2 fails to perform well in our simulation study. Next, consider the \LIB{}-1 algorithm,  which also uses  the  balancing function  $f(b) = \sqrt{b}$ to weight neighboring states. 
The only difference is that in \LIB{}-1, we perform the proposal weighting for addition and deletion moves separately. It may be surprising that this simple modification improves the sampling performance substantially in our simulation studies. 

To explain this, assume Condition~\ref{cond:ywj} holds.  
By Condition~\cond{c2}, as long as $\gamma$ is underfitted, there exists some $\gamma' \in \adds( \gamma)$  such that $\PP(\gamma')/ \PP(\gamma) \geq p^{c_1}$, and thus the proposal probability $\KLIB(\gamma, \gamma')$ is large. 
Further, we have 
\begin{align*}
\PR(\gamma, \gamma') \frac{ \KLIB(\gamma', \gamma)}{\KLIB(\gamma, \gamma')} =  \frac{  \sum_{\tgamma \in \adds(\gamma)} \sqrt{\PR(\gamma, \tgamma) } }{ \sum_{\tgamma \in \dels(\gamma')} \sqrt{  \PR(\gamma', \tgamma) } } \geq  \frac{ p^{ (c_1 - \kappa)/2 } }{ s_0 }, 
\end{align*}
where the inequality follows from $\PR(\gamma', \tgamma)  \leq p^{\kappa}$ for any $\tgamma \in \dels(\gamma')$ and $|\dels(\gamma')| \leq s_0$. 
So if the signal-to-noise ratio is sufficiently large so that  $c_1 > \kappa + 2$, the proposal will always be accepted. 
A similar argument shows that for an overfitted model $\gamma$, a proposal to remove a non-influential covariate will be always accepted if the constant $c_0$ in Condition~\cond{c1} is greater than $1$. 
This heuristic argument explains why, unlike \LIB{}-2, \LIB{}-1 does not get trapped at a model because of extremely small acceptance probabilities of informed proposal moves. 
However, it is not clear whether  \LIB{}-1 can attain a dimension-free mixing time for high-dimensional variable selection, and even if it is possible, it would require stronger assumptions on the true model so that $c_1 > \kappa$. The simulation study in Section~\ref{sec:sim} also shows that \LIB{}-1 under-performs the two \LOT{} algorithms.

\subsection{Closing remarks} \label{sec:compare.other}
Theorem~\ref{th:mh.mix} provides the theoretical guarantee for the use of informed MCMC methods for high-dimensional problems, the proof of which relies on a novel ``two-stage drift condition'' argument. Simulation studies show that \LOT{} is indeed much more efficient than the uninformed version, no matter whether the posterior distribution is multimodal. 
As noted in~\citet{zanella2020informed}, one can further boost \LOT{} using parallel computing~\citep{lee2010utility}: the calculation of $\PP(\gamma')/\PP(\gamma)$ for each $\gamma' \in \cN(\gamma)$ can be easily parallelized. 

One major advantage of \LOT{} is its simplicity, which makes it both theoretically and practically appealing. 
The adaptive MCMC methods proposed by~\citet{griffin2021search} have the usual sensitivities of possibly adapting to wrong information and require running multiple chains in the adaptation phase. 
The tempered Gibbs sampler  of~\citet{zanella2019scalable}, which is one of the most efficient existing MCMC methods (see Supplement B.6 therein), is conceptually very similar to our method in that it selects the coordinate to update using local information of $\PP$. But, as a consequence of this informed updating scheme, the tempered Gibbs sampler  requires the calculation of an importance weight in each iteration, which may reduce the efficiency of the sampler when the weight is unbounded.   
\LOT{} has a provable mixing time bound and, due to its simple design,  can be combined with other MCMC techniques such as tempering, blocking and adaptive proposals. But whether further sophistication enhances the sampler's efficiency needs more investigation, which we leave for future work. 

\section*{Acknowledgements}
GOR was supported by EPSRC grants EP/R018561/1 and EP/R034710/1, JSR was supported by NSERC grant RGPIN-2019-04142, and DV was supported by SERB grant SPG/2021/001322. We thank the submitters and participants of the two dbGaP studies (phs000308.v1.p1 and phs000238.v1.p1), which were funded by NIH.

\bibliographystyle{plainnat}
\bibliography{ref.bib}

\clearpage 
\newpage

\renewcommand{\thelemma}{S\arabic{lemma}}
\renewcommand{\thetheorem}{S\arabic{theorem}}
\renewcommand{\thesection}{S\arabic{section}}
\renewcommand{\theremark}{S\arabic{remark}}
\setcounter{remark}{0}
\setcounter{lemma}{0}
\setcounter{section}{0}
\setcounter{subsection}{0}
\setcounter{theorem}{0}


\begin{center}
\Large{Supplement to ``\TITLE{}''} \\ \medskip
\normalsize{
Quan Zhou, Jun Yang, Dootika Vats, Gareth O. Roberts and Jeffrey S. Rosenthal
}
\end{center}

Section~\ref{sec:prelim} is a brief review on some known results for the drift condition, and the proof of Theorem~\ref{th:drift} is provided in Section~\ref{sec:proof}.  
In Section~\ref{sec:cond.ywj}, we state the main result of~\citet{yang2016computational} and explain how to establish Condition~\ref{cond:ywj} for any fixed constants $c_1, c_2 \geq 0$, using essentially the same assumptions. 
Section~\ref{supp:sim} provides additional results for the simulation studies considered in Section~\ref{sec:sim}. 
Section~\ref{sec:proof.sel} contains all the remaining proofs. 
Section~\ref{sec:data} gives the information about the code and real data used in this work. 

\section{Preliminary results for the drift condition}\label{sec:prelim}
We use the notation introduced in Section~\ref{sec:main}. 
Given a drift condition on the set $\cX \setminus \sC$,  it is well known that the entry time of the chain into $\sC$ has a ``thin-tailed'' distribution.

\begin{lemma}\label{lm:drift0}
Let $(\sX_t)_{t \in \bbN}, \cX, \bP$ be as given in Assumption~\ref{ass:mc}. 
Suppose that there exist  a function $V\colon  \cX \rightarrow [1, \infty)$, a constant $\lambda \in (0, 1)$, and a non-empty set $\sC \in \cE$ such that $  (\bP V)(x) \leq \lambda V(x)$ for every $x \notin \sC$. 
Let  $\tau_\sC = \min\{t \geq 0 \colon \sX_t \in \sC \}$. Then, for any $x \in \cX$, 
\begin{align*}
\E_x[\lambda^{- \tau_{\sC} } ]  \leq V(x),  \quad \text{ and } \; 
\P_x( \tau_{\sC} \geq t )  \leq \lambda^t V(x) , \quad \forall \, t \in \bbN.
\end{align*}
\end{lemma}
\begin{proof}
Let $\sY_t = \lambda^{-t} V(\sX_t)$. The drift condition implies that $\sY_{t \wedge \tau}$ is a supermartingale. 
The results then follow from optional sampling theorem and  Markov's inequality.  
\end{proof}

The following theorem due to~\citets{jerison2016drift} gives a very useful bound on the mixing time  
when we have the generating function of the hitting time of some state $x^*$.
In the original version~\citeps[Theorem 4.5]{jerison2016drift}, it is assumed that the  drift condition holds on $\cX \setminus \{x^*\}$. An inspection of their proof reveals that we only need  $(\bP V) (x^*) < \infty$ and $\E_x[ \lambda^{- \tau^*} ]  \leq V(x)$ (if the single element drift condition holds, then this follows from Lemma~\ref{lm:drift0}). 
For more general results on the relationship between hitting time and  mixing time, see  \citets{aldous1982some, griffiths2014tight, peres2015mixing, anderson2019drift} among many others.  
These results are mostly developed for finite state spaces.   
 
\begin{theorem}\label{th:j2}
Let $(\sX_t)_{t \in \bbN}, \cX, \bP, \pi$ be as given in Assumption~\ref{ass:mc}. 
Suppose there exist a function $V \colon \cX \rightarrow [1, \infty)$,   a constant  $\lambda \in (0, 1)$  and a point $x^* \in \cX$ such that $(\bP V) (x^*) < \infty$ and $\E_x[ \lambda^{- \tau^*} ]  \leq V(x)$   where $\tau^* = \min\{t \geq 0 \colon \sX_t = x^* \}$. 
Then, for every $t  \in \bbN$ and  $x \in \cX$,  
\begin{align*}
 \TV{ \bP^t(x, \cdot )  - \pi} \leq  2  V(x)\lambda^{t + 1}. 
\end{align*}
Further, $ \TV{ \bP^t(x^*, \cdot )  - \pi} \leq   \lambda^{t + 1}$ for every $t  \in \bbN$. 
\end{theorem}

\begin{proof}
See~\citets[Chapter 4.6]{jerison2016drift} for the proof. 
\end{proof}
 
\begin{remark}\label{rmk:consist}
As shown in~\citets[Chapter 4.6]{jerison2016drift}, the assumptions of Theorem~\ref{th:j2} imply that $\pi$ is unique and $\pi(x^*) \geq 1 - \lambda$. 
For high-dimensional model selection problems where $x^*$ is the ``best model'',  this yields the rate of strong model selection consistency. 
\end{remark}

\begin{remark}\label{rmk:jproof}
In Jerison's proof of Theorem~\ref{th:j2}, a critical intermediate step is to  show that $\TV{ \bP^t(x^*, \cdot) - \pi } \leq \P_\pi ( \tau^* > t)$.  
Let $\sX_t$, $\tX_t$ be two Markov chains with transition kernel $\bP$, $\sX_0 = x^*$ and $\tX_0 \sim \pi$. 
By the famous coupling inequality~\citeps{pitman1976coupling,lindvall2002lectures},   we have $\TV{ \bP^t(x^*, \cdot )  - \pi} \leq \P( T > t) $ for $T = \min\{t \geq 0\colon \, \sX_t = \tX_t = x^* \}$. So it only remains to couple $\sX_t, \tX_t$  in such a way that $\tX_t = x^*$ implies $\sX_t = x^*$. 
 \citets{jerison2016drift} finds this coupling (though not explicitly) by using a duality technique, known as ``intertwining of Markov chains''~\citeps{yor1988intertwinings, diaconis1990strong}, and a monotonicity result  of~\citets{lund2006monotonicity}. 
The latter requires that $\bP$ be reversible and have non-negative spectrum. 
\end{remark}

\section{Proof of Theorem~\ref{th:drift}}\label{sec:proof}

The outline of the proof of Theorem~\ref{th:drift} is as follows.
Let $\tau^* = \min\{t \geq 0 \colon \sX_t = x^*\}$ denote the hitting time of the state $x^*$. 
By Theorem~\ref{th:j2} in the supplement, all we need is to bound the generating function for $\tau^*$, $\E_x[ \alpha^{- \tau^*} ]$ for $\alpha \in (0, 1)$.
For our problem, directly bounding the generating function seems difficult.  So we first  find a tail bound instead. 
To this end, we split the path of $(\sX_t)$ into disjoint ``excursions'' in $\sA$ and $\sA^c$ (the length of excursion in $\sA^c$ may be zero). This splitting scheme is the most important step of our proof (see Section~\ref{sec:proof.split}). 
For each excursion in $\sA$, there is some positive probability that the chain can hit $x^*$, and then we can use a union bound to handle the tail probability of $\tau^*$ in the same way as in Theorem 1 of~\citet{rosenthal1995minorization} (see Sections~\ref{sec:proof.tail} and~\ref{sec:proof.prob}). 
Finally, by carefully tuning the parameters in the tail bound for $\tau^*$, we are able to compute its generating function (see Section~\ref{sec:proof.gf}). 

\begin{remark}\label{rmk:proof.two.stage}
\citet[Theorem 2.1]{roberts1999bounds}  give a bound on the generating function for the regeneration times in the drift-and-minorization setting. For that problem whether the chain regenerates depends on the outcome of an independent coin flip, and thus it is possible to condition on the number of coin flips needed to regenerate. 
But in our setting, we cannot bound the generating function of $\tau^*$ by conditioning on the number of excursions in $\sA$ needed for the chain to hit $x^*$, since such conditioning distorts the distribution of $(\sX_t)$.    
\end{remark}

\subsection{Path splitting for $(\sX_t)$} \label{sec:proof.split}
We first find a decomposition of $\bP$. 
Define a transition kernel $\bQ$   by 
\begin{equation*}\label{eq:def.Q}
\bQ(x,  \sC) =  \frac{ \bP(x,  \sC \cap \sA)   }{  \bP(x, \sA) }, \quad \forall \, x \in \sA, \, \sC \in \cE. 
\end{equation*}
The case $x \notin \sA$ is irrelevant to our proof, and one can 
simply let $\bQ(x, \cdot ) =\bP(x, \cdot)$  if $x \notin \sA$.
For $x \in \sA$, the distribution $\bQ(x, \cdot)$ is just $\bP(x, \cdot)$ conditioned on the chain staying in $\sA$. 
Further, condition~\eqref{d5}  implies that $(1 - q) \bQ(x,  \cdot )  \leq \bP(x, \cdot )$.
Hence, there always exists a ``complementary'' transition kernel $\bR$  such that 
\begin{equation}\label{eq:def.cP}
\bP(x,  \cdot ) = q \bR(x, \cdot ) +  (1 - q) \bQ(x, \cdot),  \quad \forall \,  x  \in \cX. 
\end{equation}

Now we re-construct the Markov chain $(\sX_t)_{t \in \bbN}$. 
First, we generate a sequence of i.i.d. Bernoulli random variables,  $(Z_0, Z_1,   \dots)$, such that $Z_i$ is equal to $1$ with probability $q$.  
Starting with $\sX_0 = x \in \cX$, we update the chain as follows. 
\begin{align*}
 \text{If } \sX_t \in \sA^c,    \quad & \text{ generate } \sX_{t + 1} \sim \bP(\sX_t, \cdot).  \\
 \text{If } \sX_t \in \sA, \,  Z_{t + 1}  = 0,    \quad & \text{ generate } \sX_{t + 1} \sim \bQ(\sX_t, \cdot). \\
 \text{If } \sX_t \in \sA, \, Z_{t + 1}  = 1,   \quad & \text{ generate } \sX_{t + 1} \sim \bR(\sX_t, \cdot). 
\end{align*}
It follows from~\eqref{eq:def.cP} that marginally, $(\sX_t)_{t \in \bbN}$ is a Markov chain with transition kernel $\bP$, and we will use $\P_x$ to denote the corresponding probability measure. 
Let $(\cF_t)_{t \in \bbN} = \sigma(\sX_0, \dots, \sX_t, Z_0, \dots, Z_t)$ denote the filtration generated by $(\sX_t, Z_t)_{t \in \bbN}$. Set $\omega_0 = 0$ and then  define the following hitting times with respect to $(\cF_t)$ recursively:
\begin{align*}
\sigma_k  =  \min\{ t \geq  \omega_{k - 1} \colon  \sX_t \in \sA \}, \quad \quad  \omega_k   =  \min\{ t > \sigma_k \colon  Z_t = 1 \}, 
\quad \quad k = 1, 2, \dots. 
\end{align*}
Observe that $\sX_t \in \sA$ if  $\sigma_k \leq t \leq \omega_k - 1$. 
For $k \geq 1$, $\omega_k$ marks the $k$-th time we update the chain using $\bR$, and then we return to the set $\sA$ at time $\sigma_{k + 1}$.  
Note that $\sigma_{k + 1} = \omega_k$ if $\sX_{\omega_k} \in \sA$ (i.e., the ``return'' happens immediately.)
For $k \geq 1$, let $S_{ k -1} =   \sigma_k  - \omega_{k - 1} \geq 0$ and $N_k =  \omega_k - \sigma_k \geq 1$. 
We have 
\begin{equation}\label{eq:s0}
\omega_k  = S_0 +  ( N_1 + S_1 ) + \cdots  + ( N_{k - 1} + S_{k - 1} ) + N_k. 
\end{equation}
  $N_k$ (resp. $S_{k - 1}$) can be seen as the length of  the $k$-th stay of $(\sX_t)$ in $\sA$ (resp. $\sA^c$). 
See Figure~\ref{fig.path} for a graphical illustration of our path splitting scheme. 

\begin{figure}[t!]
\begin{center}
\includegraphics[width=0.95\linewidth]{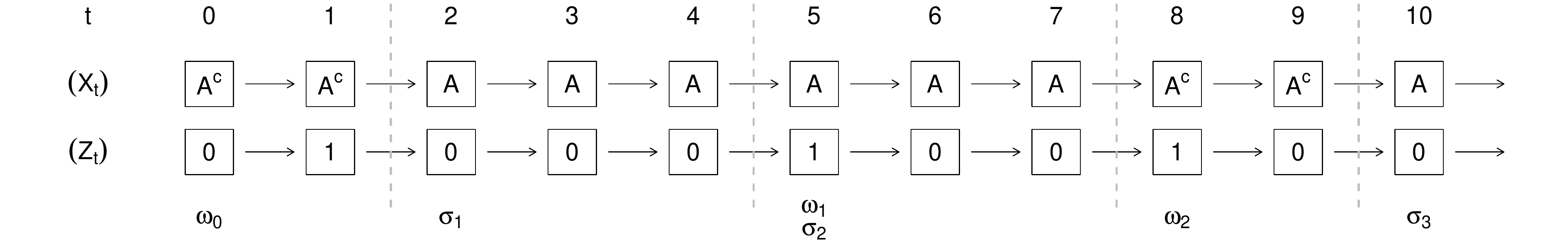} \\
\caption{An example of the evolution of $(X_t, Z_t)$. For $(X_t)$, we only indicate whether it is in $\sA$ or not. 
For $t = 5$ and $t = 8$, we generate $X_t$ from $\bR(X_{t - 1}, \cdot)$. 
Since $X_5 \in \sA$, $\omega_1 = \sigma_2 = 5$. 
For this example, we have $S_0 = 2$, $N_1 = 3$, $S_1 = 0$, $N_2 = 3$ and $S_2 = 2$. }\label{fig.path}
\end{center}
\end{figure}  
  
\subsection{Tail bound for $\tau^* $}  \label{sec:proof.tail}

Let $\sB_k =  \{ \sX_t \neq x^*,  \; \text{ for }   \sigma_k \leq t  \leq \omega_k - 1 \}$ be the event that $(\sX_t)$ does not hit $x^*$ during its $k$-th  stay in $\sA$. Then, for every fixed $t$, 
\begin{align*}
\P_x(    \tau^*  > t \geq \omega_k  ) \leq  \P_x(    \tau^*   \geq \omega_k  ) =  \P_x(  \sB_1 \cap   \cdots \cap \sB_k   ),
\end{align*}
since by $\omega_k$, $(\sX_t)$ has finished  $k$ ``excursions'' in $\sA$. 
As in~\citet[Theorem 1]{rosenthal1995minorization},  we apply the union bound to get 
\begin{equation}\label{eq:th1.union}
\begin{aligned}
\P_x (  \tau^* > t )  \leq \;& \P_x ( \omega_j > t) 
+ \P_x (  \tau^* > t , \omega_j \leq t)    \\
\leq \;& \P_x (    \omega_j  > t )  +   \P_x\left(  \sB_1 \cap   \cdots \cap \sB_j   \right), 
\end{aligned}
\end{equation}
which holds for any positive integer $j$. Therefore, it suffices to bound the two terms $\P_x (    \omega_j  > t )$ and $ \P_x\left(  \sB_1 \cap   \cdots \cap \sB_j   \right)$. The upper bound for $ \P_x\left(  \sB_1 \cap   \cdots \cap \sB_j   \right)$ is deferred to the next subsection. 
 
The key to bounding the tail probability $\P_x (    \omega_j  > t )$ is to show that  $S_0, N_1, S_1, \dots$ all have geometrically decreasing tails, and thus so does $\omega_j$. 
First, note that  $N_k =  \min \{  i  \geq 1 \colon Z_{\sigma_k + i} = 1 \}$, which is just a geometric random variable since $(Z_t)$ is an i.i.d. sequence. 
Therefore, $\E[ u^{ N_k  } ]$ exists for any $u < (1 - q)^{-1}$. 
We choose 
\begin{equation*}\label{eq:def.u}
u = \frac{1}{ 1 - q/2}, 
\end{equation*}
which is less than $\min\{ \lambda_1^{-1}, \lambda_2^{-1} \}$  by condition~\eqref{d5}. 
It is also evident by construction that $N_k$ is  independent of $\cF_{\sigma_k }$, which yields
\begin{equation}\label{eq:th1.N}
  \E_x [ u^{ N_k  } \mid  \cF_{\sigma_k} ]  =  \frac{ u q}{1  - u (1  - q)} = 2, \quad \text{a.s.} 
\end{equation}

Next, consider the random variables $\{S_k \colon k = 0, 1, \dots \}$. We have 
\begin{equation}\label{eq:th1.S0}
\E_x[  u^{S_0}  ] \leq  \E_x[  \lambda_1^{ - S_0}  ] \leq V_1(x), 
\end{equation}
by Lemma~\ref{lm:drift0} and drift condition~\eqref{d1}. 
Similarly,  for $S_k$ with $k \geq 1$,  
\begin{equation}\label{eq:th1.Sk}
\begin{aligned}
\E_x[  u^{S_k} \mid \cF_{\omega_k - 1} ]    \leq \;&  \E_x[ V_1(\sX_{\omega_k} )  \mid \cF_{\omega_k - 1} ]  \\ 
= \;&  \E_x[ V_1(\sX_{\omega_k} )  \mid \cF_{\omega_k - 1}, \sX_{\omega_k} \in \sA ] \P_x(  \sX_{\omega_k} \in \sA \mid   \cF_{\omega_k - 1}) \\
 & + \E_x[ V_1(\sX_{\omega_k} )  \mid \cF_{\omega_k - 1}, \sX_{\omega_k} \in \sA^c ] \P_x(  \sX_{\omega_k} \in \sA^c \mid   \cF_{\omega_k - 1}) \\ 
\leq \;&  \bR(  \sX_{\omega_k - 1},  \sA) + \frac{M}{2} \bR(  \sX_{\omega_k - 1},  \sA^c) 
\leq M / 2, \quad \text{a.s.} 
\end{aligned}
\end{equation}
The second last inequality follows from  condition~\eqref{d3}  and the observation that  $\sX_{\omega_k}$ is generated from $\bR (\sX_{\omega_k - 1}, \cdot)$ and $\sX_{\omega_k - 1} \in \sA$. 
Using~\eqref{eq:s0},~\eqref{eq:th1.N},~\eqref{eq:th1.S0},~\eqref{eq:th1.Sk} and conditioning on  $\cF_{\omega_j - 1}, \cF_{\sigma_j}, \dots, \cF_{\sigma_1}$ recursively (which is allowed since $\omega_k - \sigma_k \geq 1$), we find that 
\begin{equation}\label{eq:th1.gf}
\E_x[  u^{ \omega_j  }  ]  = \E_x[ u^{ S_0 + \cdots + N_j } ]
\leq  2 V_1(x)     M^{j - 1}. 
\end{equation}
The tail probability $ \P_x (    \omega_j  > t ) $ then can be bounded by Markov's inequality.

\subsection{Upper bound for $\P_x( \sB_1 \cap \cdots \cap \sB_j)$}  \label{sec:proof.prob}
Next, we show that  $\P_x( \sB_k \mid \sB_1, \dots, \sB_{k - 1} )\le \rho$ for some $\rho<1$, which implies the upper bound $\P_x (\sB_1 \cap \cdots \cap \sB_j ) \leq  \rho^j$. First, note that $(\sX_t, Z_t)$ forms a bivariate Markov chain and thus 
\begin{align*}
\P_x \left( \sB_k  \mid  \sX_{\sigma_k} =  y, \, Z_{ \sigma_k } =  z   \right)  =  \P_y (   \tau^* \geq  \omega_1 \mid Z_0 = z   ) =  \P_y (   \tau^* \geq  \omega_1),  \quad \text{a.s.} 
\end{align*}
where $\omega_1 = \min\{t \geq 1 \colon Z_t = 1 \}$ because $y = \sX_{\sigma_k} \in \sA$. 
Moreover, $\omega_1 = N_1$ is a geometric random variable independent of $\cF_0$. Conditioning on $N_1$, we find 
\begin{align}
\P_y (   \tau^* \geq  \omega_1) = \P_y (    \tau^* \geq    N_1  ) \nonumber  
 = \;&    \sum\limits_{t = 1}^\infty \P_y ( \tau^* \geq  t \mid  N_1  = t ) (1 - q)^{t - 1} q \nonumber \\ 
 = \;&   \sum\limits_{t = 1}^\infty \P_y (  \tau^* \geq  t \mid  Z_0 = \cdots = Z_{t - 1} = 0 ) (1 - q)^{t - 1} q,  \label{eq:py}
\end{align}
where in the last step we have used that $Z_t$ is independent of $(\sX_0, \dots, \sX_{t - 1})$  and thus the event $\{\tau^* \geq  t\}$. 
On the event $\{ Z_0 = \cdots = Z_{t - 1} = 0 \}$, $(\sX_0, \dots, \sX_{t -1})$ is a Markov chain with transition kernel $\bQ$.  
For $x \in  \sA \setminus \{x^*\}$,  let $q_x = \P(x, \sA^c)$ and  write   
\begin{align*}
(\bP V_2) (x)   =\;&  (1 - q_x)  (\bQ V_2)(x)  +  q_x \,  \E_x[ V_2(\sX_1) \mid \sX_1 \in \sA^c] . 
\end{align*}
Since $(\bP V_2) (x) \leq \lambda_2 V_2(x)$ for some $\lambda_2 < 1$,  we have $(\bQ V_2)(x) \leq  \lambda_2  V_2(x)$  by condition~\eqref{d4}. 
This drift condition enables us to apply Lemma~\ref{lm:drift0} and obtain from~\eqref{eq:py} that 
\begin{equation*}
   \P_y (    \tau^* \geq    N_1  )   \leq    \sum_{t = 1}^\infty   \lambda_2^t V_2(y)   (1 - q)^{t - 1} q    
=   \frac{ V_2(y)  \lambda_2  q }{1 - \lambda_2  (1 - q)}, 
\end{equation*}
for $y \in \sA$. 
Since $\lambda_2, q \in (0, 1)$ and $ K \geq  \sup_{y \in \sA} V_2(y)$, 
\begin{equation}\label{eq:def.rho}
\sup_{y \in \sA}   \P_y (    \tau^* \geq    N_1  )  \leq  \frac{ q  K   }{1 - \lambda_2 } \coloneqq \rho. 
\end{equation}
Note that $\rho  < 1$ by condition~\eqref{d5}. 
By conditioning on $\cF_{ \sigma_k}$, we find that $\P_x( \sB_k \mid \sB_1, \dots, \sB_{k - 1} )$ is also bounded by $\rho$, and thus $\P_x (\sB_1 \cap \cdots \cap \sB_j ) \leq  \rho^j$. 
 
\subsection{Generating function for  $\tau^*$} \label{sec:proof.gf}
We now combine the previous results to bound the generating function of $\tau^*$. Using the upper bound given in~\eqref{eq:th1.union} and  $\P_x (\sB_1 \cap \cdots \cap \sB_j ) \leq  \rho^j$,  we have 
\begin{align*}
\E_x[ \alpha^{-\tau^*}  ]  = \;&  1 + (\alpha^{-1} - 1) \sum_{k=0}^\infty \P( \tau^* > k ) \alpha^{-k} \\
\leq \;&  1 + (\alpha^{-1} - 1) \sum_{k=0}^\infty   \alpha^{-k} \left[ \rho^{j_k} + \P_x ( \omega_{j_k} > k   )  \right],
\end{align*}
for any $\alpha \in (0, 1)$, where $\rho$ is as given in~\eqref{eq:def.rho}. 
We choose $j_k =\lfloor rk + 1\rfloor \geq rk$ for some $r > 0$.
It then follows from~\eqref{eq:th1.gf} and~\eqref{eq:def.rho} that 
\begin{align*}
\E_x[ \alpha^{- \tau^*}  ]  \leq  1 + (\alpha^{-1} - 1) \sum_{k=0}^\infty  \alpha^{-k} \left[ \rho^{r k} + \frac{2 V_1(x)}{  M} \frac{    M^{r k} }{   u^k }  \right]. 
\end{align*}
The right-hand side  might diverge  if either $\rho^{r}$ or $M^{r}/u$ is too large. 
Hence, to obtain the optimal convergence rate, we set 
\begin{align*}
r = \frac{ \log u}{ \log (M / \rho) }, \quad \text{ which yields }
   \rho^{r} = \frac{ M^r }{u} < 1. 
\end{align*}
Since $M/2 \geq 1$, $0 < \rho < 1$ and $1 < u < 2$, we always have $r  \in (0, 1)$.

Finally, we choose $\alpha = (1 + \rho^r ) / 2 < 1$ and find that 
\begin{align*}
\E_x[ \alpha^{-\tau^*}  ]  \leq  1 +   \left(1 + \frac{ 2V_1(x) }{   M } \right) \frac{1 - \alpha}{\alpha - \rho^r} = 2 + \frac{2V_1(x)}{M}. 
\end{align*} 
The proof is then completed by applying Theorem~\ref{th:j2}.

\section{Justification for Condition~\ref{cond:ywj}}\label{sec:cond.ywj}
In this section, we review the high-dimensional assumptions used in~\citet{yang2016computational} to prove Condition~\ref{cond:ywj}. 

\begin{theorem}[\citets{yang2016computational}]\label{th:ywj}
Consider the Bayesian variable selection problem described in Section~\ref{sec:selection}. Suppose the true error variance $\ev = 1$ and the following conditions hold for some finite constants $\CC_0 \geq 0$,  $\CC_1 > 0$,  $\zeta \in (0, 1]$  such that $\CC_1 \zeta \geq 4$. 
\begin{enumerate}[(A)]
\item $\norm{  \X \beta^*}_2^2 \leq g  \log p$,  and $\norm{ \X_{u^* } \beta^*_{u^*}  }_2^2 \leq \CC_0 \log p$ where $u^* = [p] \setminus \gamma^*$.  \label{y1}
\item $\X_j^\top \X_j = n$ for each $j$, and 
$$\min_{\gamma \in \cM(s_0)}  \Lambda_{\rm{min}} (\X_\gamma^\top \X_\gamma) \geq n \zeta,$$ where $\Lambda_{\rm{min}}$ denotes the smallest eigenvalue. \label{y2}
\item  For  $z \sim \N(0, I_n)$, \label{y3}
$$ \E\left[ \max_{\gamma \in \cM(s_0)} \max_{k \notin \gamma}  |   \X_k^\top \OPJ_\gamma  z  | \,  \right] \leq  \frac{1}{2} \sqrt{  \CC_1 \zeta  n  \log p   }.$$ 
\item $\kappa_0 \geq 2$, $\kappa_1 \geq 1 / 2$ and $\kappa = \kappa_0 + \kappa_1 \geq 4 (\CC_0 + \CC_1) + 2$. \label{y4}
\item Let $\Psi(\X) =  \max_{\gamma \in \cM(s_0)} \norm{  (\X_\gamma^\top \X_\gamma)^{-1} \X_\gamma^\top \X_{\gamma^* \setminus \gamma} }_{\rm{op}}^2.$ Then,  \label{y5}
\begin{align*}
\max\left\{ 1, \,   (2 \zeta^{-2} \Psi(\X) + 1 ) s^* \right\} \leq s_0 \leq  \frac{n}{32 \log p} - \frac{\CC_0}{4}. 
\end{align*} 
\item The threshold $\Cbeta$ given in~\eqref{eq:beta.min} satisfies \label{y6}
$$\Cbeta^2 \geq  \frac{128 ( \kappa + \CC_0 + \CC_1)  \log p }{\zeta^2  n }.$$  
\end{enumerate}
Then, with probability $1 -  O(p^{- a})$ for some universal constant $a > 0$, Condition~\ref{cond:ywj} holds for $c_0 = 2$ and $c_1 = 4$. 
\end{theorem}
\begin{proof}
See~\citets[Lemma 4]{yang2016computational}.  Though the original result was stated for $c_0 = 2$ and $c_1 = 3$, replacing $c_1 = 3$ with $c_1 = 4$ does not require any change of their proof. 
\end{proof}

\begin{remark}
As explained in~\citets{yang2016computational}, the assumptions made in Theorem~\ref{th:ywj} are mild. 
In particular, Condition~\eqref{y3} holds for $\CC_1 = O(s_0 / \zeta)$. 
A similar result for the empirical normal-inverse-gamma prior is proved in~\citets[Supplement D]{zhou2021complexity}. 
\end{remark}
 
We note that Theorem~\ref{th:ywj} holds for any other fixed values of $c_0, c_1$ under essentially the same assumptions. To explain the reason, we briefly describe below the main idea of the proof of~\citets{yang2016computational}.

\begin{proof}[Sketch of the proof for Theorem~\ref{th:ywj}]
To simplify the discussion, we assume the constant $\zeta$ in the restricted eigenvalue condition, i.e., Condition~\eqref{y2}, is a universal constant and $\CC_0 = O(1)$. 
For two positive sequences $a_n, b_n$, we write $a_n = \Omega(b_n)$ if $b_n = O(a_n)$, and $a_n = \Theta (b_n)$ if $a_n = O(b_n)$ and $b_n  = O(a_n)$.
Using $s_0 \log p = O(n)$ from Condition~\eqref{y5}  
and concentration inequalities, one can show that models in $\cM(s_0)$ cannot ``overfit'', by which we mean that $y^\top \OPJ_\gamma y = \Omega(n )$.  
Then Condition~\eqref{y1} guarantees that the term $g^{-1} y^\top y = O(\log p)$ in~\eqref{eq:ppr} is negligible, and we  can  write
\begin{align*}
\PR(\gamma, \gamma')   = p^{  \kappa (|\gamma| - |\gamma'| ) } \left\{ 1 +  \frac{y^\top (\PJ_{\gamma'} - \PJ_{\gamma})y}{ y^\top \OPJ_{\gamma'} y + O(\log p)}  \right\}^{ n/2}. 
\end{align*}

Consider $\gamma' = \gamma \cup \{j\}$ for some overfitted $\gamma$ and $j \notin \gamma$. 
Since $\gamma^* \subseteq \gamma$, Condition~\eqref{y1} implies that $ y^\top \OPJ_{\gamma'} y = \Theta(n)$, and  Condition~\eqref{y3} yields that $y^\top (\PJ_{\gamma'} - \PJ_{\gamma})y = O( \CC_1 \log p)$. Hence, 
\begin{align*}
\PR(\gamma, \gamma')     = p^{-\kappa} \left\{ 1 + \frac{ O(\CC_1 \log p) }{  \Theta(n)}  \right\}^{ n/2}. 
\end{align*}
To prove $B(\gamma, \gamma') \leq p^{-c_0}$ for some $c_0 > 0$, we only need $\kappa \geq a_1 \CC_1 + a_2$ for some sufficiently large constants $a_1$ and $a_2$  (and then apply the inequality $1 + x \leq e^x$). This is where Condition~\eqref{y4} is needed. 

Next, consider $\gamma' = \gamma \cup \{j\}$ for some underfitted $\gamma$ and $j \in \gamma^* \setminus \gamma$.  By~\citets[Lemma 8]{yang2016computational}, $j$ can be chosen such that   $y^\top (\PJ_{\gamma'} - \PJ_{\gamma})y  = \Omega( n \Cbeta^2)$.   
(Note that this may not be true for every $j \in \gamma^* \setminus \gamma$.)
Therefore, we can write 
\begin{align*}
\PR(\gamma', \gamma)     = p^{ \kappa} \left\{ 1 - \frac{ \Omega( n \Cbeta^2) }{ y^\top \OPJ_{\gamma} y  }  \right\}^{ n/2}. 
\end{align*}
Here one needs to consider two possible subcases. 
If $y^\top \OPJ_{\gamma'} y = \Theta (n)$, then to prove $B(\gamma', \gamma) \leq p^{ - c_1}$ for some $c_1 > 0$,  we just need $\Cbeta^2 \geq a_3  (\kappa + 1) \log p / n$ for some sufficiently large $a_3$, 
which is guaranteed by Condition~\eqref{y6}. 
If $y^\top \OPJ_{\gamma'} y$ has a larger order than $n$, we need a slightly different argument. 
By~\citets[Lemma 8]{yang2016computational}, we can pick $j$ such that $\gamma' = \gamma \cup \{j\}$ satisfies 
\begin{align*}
\frac{ y^\top (\PJ_{\gamma'} - \PJ_{\gamma})y }{\ty^\top \OPJ_{\gamma} \ty} = \Omega\left( \frac{1}{s^*} \right), 
\end{align*} 
where $\ty = \X_{\gamma^*} \beta^*_{\gamma^*}$ denotes the signal part of $y$. 
Then, $B(\gamma', \gamma) \leq p^{ - c_1}$ would hold if  $\kappa s^* \log p \leq  n / a_4$ for some sufficiently large $a_4$.  
\end{proof}

\section{More results for simulation studies}\label{supp:sim}

\subsection{Local posterior landscape in simulation study I} \label{supp:sim.local}
We first consider the simulation study conducted in Section~\ref{sec:sim1}, which we refer to as simulation study I (and we refer to that conducted in Section~\ref{sec:sim.ess} as simulation study II). 
Let $\true = \{1, \dots, 10\}$ denote the set of all covariates with nonzero regression coefficients in this simulation. 
For each $\gamma$, define 
\begin{align*}
    \fR_{\rm{a0}}(\gamma) =\;& \max_{j \notin  \gamma } \; \log_p \frac{\PP(\gamma \cup \{j\} )}{\PP(\gamma)}, \quad  \quad 
    \fR_{\rm{d0}}(\gamma) = \max_{ j \in \gamma } \;  \log_p \frac{\PP(\gamma \setminus \{j\} )}{\PP(\gamma)}, \nonumber \\  
    \fR_{\rm{a1}}(\gamma) =\;&  \min_{j \in \true \setminus \gamma } \log_p \frac{\PP(\gamma \cup \{j\})}{\PP(\gamma)}, \quad 
    \fR_{\rm{d1}}(\gamma) =  \min_{j \in \gamma \setminus \true} \log_p \frac{\PP(\gamma \setminus \{j\})}{\PP(\gamma)}.  \nonumber 
\end{align*} 
When SNR is sufficiently large,  Condition~\ref{cond:ywj} is very likely to hold with $\gamma^* = \true$, in which case Condition~\cond{c1} and Condition~\cond{c2} can be equivalently expressed as follows. 
\begin{itemize}
    \item  If $\gamma$ is overfitted and $\gamma \neq \gamma^*$, then both $\fR_{\rm{d0}}(\gamma)$ and $\fR_{\rm{d1}}(\gamma)$ should be large, since any non-influential covariate in $\gamma$ should be ``useless''.  
    \item If $\gamma$ is underfitted, then $\fR_{\rm{a0}}(\gamma)$ should be large, which means we are able to add some covariate to $\gamma$, but $\fR_{\rm{a1}}(\gamma)$ may be small since we may not be able to add every missing influential covariate due to the collinearity in the data; 
\end{itemize}
For the reason discussed in Section~\ref{sec:swap}, we do not consider swap moves.

In Figure~\ref{fig:local}, we provide the histograms of $\fR_{\rm{a0}}(\gamma), \fR_{\rm{a1}}(\gamma)$ for ``underfitted'' models and $\fR_{\rm{d0}}(\gamma), \fR_{\rm{d1}}(\gamma)$ for ``overfitted'' models (excluding $\true$) sampled in MCMC under the four settings considered in Table~\ref{table1} with SNR $=3$.  
In this plot, we classify $\gamma$ as ``overfitted'' or ``underfitted'' according to whether $\gamma$ contains $\true$ as a subset; that is, we assume $\true = \gamma^*$. Note that this assumption is not always true, since, by Table~\ref{table1}, in the second setting there are at least two replicates (out of $100$) where $\true$ is not the best model. But since this happens rarely, in this section we will slightly abuse the words ``underfitted'' and ``overfitted'' by using $\true$ as the best model.   
To explain how the data used in Figured~\ref{fig:local} is generated, observe that in each informed iteration, we need to evaluate $\PP$ for all neighboring states. Hence, we can find
$\fR_{\rm{a0}}(\gamma), \fR_{\rm{a1}}(\gamma), \fR_{\rm{d0}}(\gamma), \fR_{\rm{d1}}(\gamma)$ for each $\gamma$ sampled in informed MCMC algorithms without extra computation cost. 
Thus, for each simulated data set, we simply collect all unique models sampled by \LOT{}-1, \LOT{}-2 or \LIB{}-1, and then for each simulation setting, we merge the data for $\fR_{\rm{a0}}, \fR_{\rm{a1}}, \fR_{\rm{d0}}, \fR_{\rm{d1}}$ from all 100 replicates.

It should be noted that the distributions  shown in Figure~\ref{fig:local} could be highly ``biased'', since we only consider models that are sampled by any of the informed MH algorithms.  If $\fR_{\rm{a0}}(\gamma)$ or $\fR_{\rm{d0}}(\gamma)$ is very large for some model $\gamma$,  its posterior probability tends to be very small,  
and thus $\gamma$ is unlikely to be sampled. Hence, we may expect that the actual distribution of $\fR_{\rm{a0}}(\gamma)$ or $\fR_{\rm{d0}}(\gamma)$ on the whole state space has a larger mean than in Figure~\ref{fig:local}.  
Nevertheless, Figure~\ref{fig:local} provides useful insights into Condition~\ref{cond:ywj}, and we now explain why the histograms in Figure~\ref{fig:local} agree well with the theory. Write $y = \ty + \varepsilon$, where $\ty = \X_{\true} \beta^*_{\true}$ denotes the signal part and $\varepsilon$ denotes noise. 

Observe that in all four settings, for any overfitted model $\gamma \neq \true$, $\fR_{\rm{d0}}(\gamma)$ is usually very close to $3.5$, and $\fR_{\rm{d1}}(\gamma)$ is usually greater than $3$, lending support to Condition~\cond{c1}. 
To explain this, fix an arbitrary $j \in \gamma \setminus \true$. 
Clearly, $y^\top \OPJ_{ \gamma \setminus \{j\} } y = \varepsilon^\top \OPJ_{ \gamma \setminus \{j\} }  \varepsilon$, and since we draw $\varepsilon$ from $\N(0, I_n)$  independently of the design matrix, we obtain from~\eqref{eq:ppr} that  
\begin{align*}
    p^{ - \kappa} \leq \PR(\gamma \setminus \{j\}, \gamma) = p^{ - \kappa} 
    \left\{ 1 + \frac{ \chi_1^2 }{g^{-1} y^\top y + \chi^2_{n - |\gamma| - 1} } \right\}^{ n/2}
    \leq p^{-\kappa} \exp\left\{ \frac{n\chi_1^2}{2 \chi^2_{n - |\gamma| - 1}   } \right\}, 
\end{align*}
where $\chi_1^2, \chi^2_{n - |\gamma| - 1}$ denote two independent chi-squared random variables with degrees of freedom $1$ and $n - |\gamma| - 1$ respectively. 
Hence,
$$  \kappa \geq \log_p \PR(\gamma, \gamma \setminus \{j\}) \gtrsim   \kappa - \frac{ n \chi_1^2 }{ 2 \chi^2_{n - |\gamma| - 1}  \log p }.$$  
Assuming $n$ is sufficiently large so that $\chi^2_{n - |\gamma| - 1} / n \approx 1$, this calculation suggests that, when $p = 1,000$,   $\log_p \PR(\gamma, \gamma \setminus \{j\}) \in (\kappa - 0.5, \kappa)$ with probability $\geq 99\%$.  
In simulation study I, we set the ``effective sparsity parameter''  $\kappa = \kappa_0 + \kappa_1 = 3.5$, which explains what we observe in the left column of Figure~\ref{fig:local}. 

The right column of Figure~\ref{fig:local} shows that for any underfitted  $\gamma$, $\fR_{\rm{a0}}(\gamma)$ is usually large, especially in the two cases with independent design, which means that we are able to increase the posterior probability significantly by adding some covariate. This provides evidence for Condition~\cond{c2}. 
It is also interesting to observe that in all four settings,  $\fR_{\rm{a1}}(\gamma)$ is very likely to be negative. As we have argued before Condition~\ref{cond:ywj} in the main text, due to the collinearity in $\X$, whether a specific influential covariate can be added to the current model is very hard to predict. This is one of the key reasons why the high-dimensional variable selection problem is challenging.  

\begin{figure}
    \centering
    \includegraphics[width=0.93\linewidth]{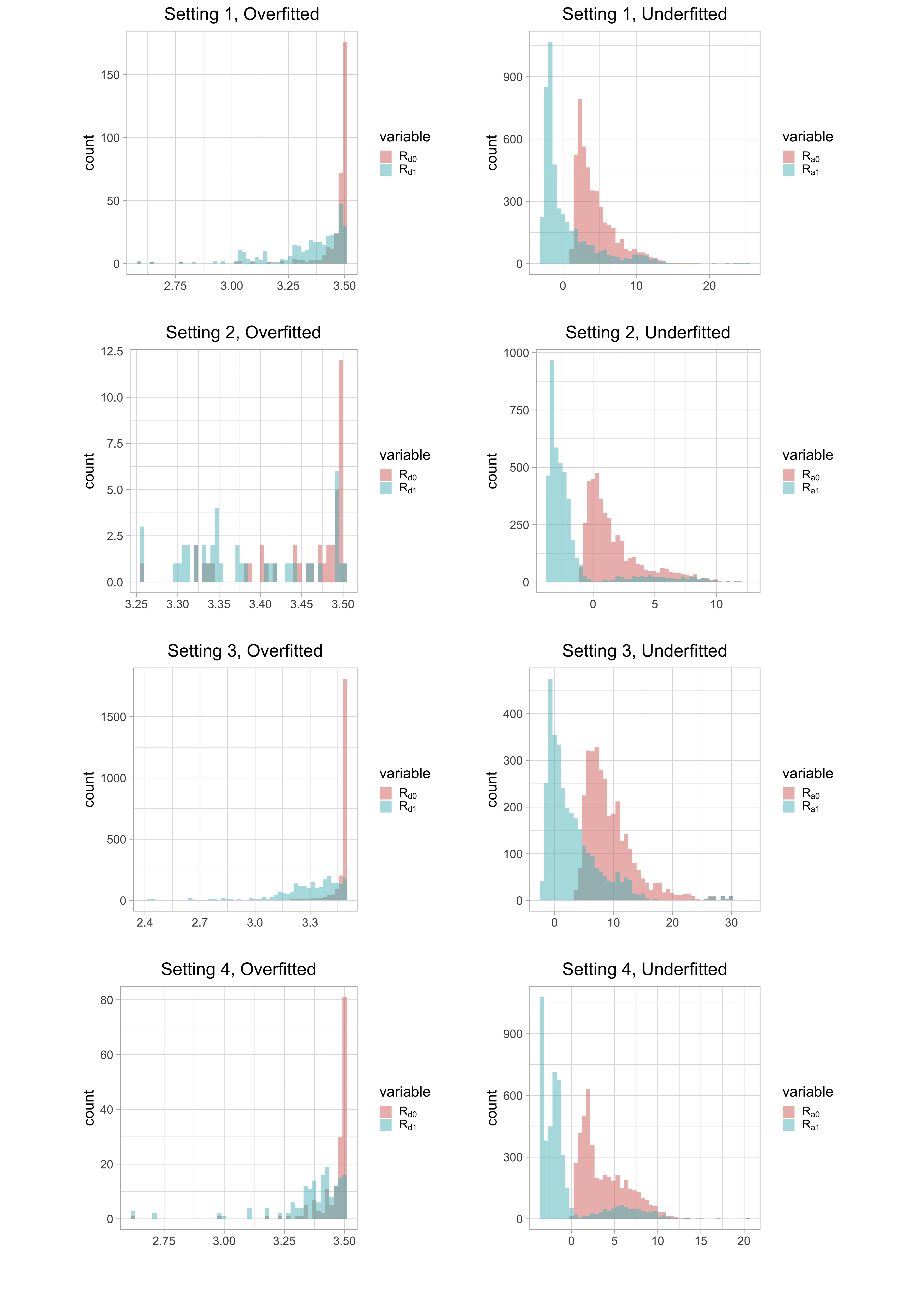}
    \caption{Local posterior landscape in simulation study I with SNR$=3$. Setting 1: $n=500, p=1000$, independent design. 
    Setting 2: $n=500, p=1000$, correlated design. 
    Setting 3: $n=1000, p=5000$, independent design. 
    Setting 4: $n=1000, p=5000$, correlated design.    
    We define ``overfitted'' and ``underfitted'' by assuming that $\true$ is the best model $\gamma^*$. The model $\true$ is excluded from the plots for ``overfitted'' models. 
    \label{fig:local} }
\end{figure}

\subsection{More results for Condition~\ref{cond:ywj}}\label{supp:cond.ywj}
In the previous subsection, we empirically checked whether Condition~\ref{cond:ywj} is satisfied in simulation study I with SNR $=3$. 
But in general, verifying Condition~\ref{cond:ywj} for a given data set is difficult for two reasons.   
First, the order of $|\cM(s_0)|$ is given by $p^{s_0}$, which can easily be astronomical unless $p$ or $s_0$ is very small (this is why we only consider models sampled by an informed MH algorithm in Figure~\ref{fig:local}). 
Second, the ``best model'' $\gamma^*$ in Condition~\ref{cond:ywj} can be hard to determine even for a simulated data set. 
To address the second difficulty, for each $\gamma$, we define 
\begin{align*}
      \fR(\gamma) =  \max\{  \fR_{\rm{a0}}(\gamma), \;  \fR_{\rm{d0}}(\gamma) \} =
      \max_{\gamma' \in \adds(\gamma) \cup \dels(\gamma)} \log_p \frac{\PP(\gamma')}{\PP(\gamma)}. 
\end{align*}
Condition~\ref{cond:ywj}  assumes that for any $\gamma \neq \gamma^*$ such that $|\gamma|$ is not too large, $\fR(\gamma)$ is greater than some constant.  
The distribution of $\fR(\gamma)$ is easier to numerically characterize since we do not need to know whether $\gamma$ is overfitted or underfitted.  
Further, by studying the distribution of $\fR(\gamma)$ we can still check the two most important implications of Condition~\ref{cond:ywj}: whether $\PP$ is unimodal, and whether its tails decay quickly. 

As we have done in the previous subsection, we collect all unique models visited by any informed MH algorithm and, for each simulation setting, merge the data from all replicates. The ``global mode'' $\hat{\gamma}_{\rm{max}}$ (i.e., the one with the largest posterior probability among all models that are sampled by any informed MH sampler) is excluded. See Figure~\ref{fig:Rx1} and Table~\ref{table:Rx1} for the distribution of $\fR(\gamma)$ in simulation study I and Figure~\ref{fig:Rx2} for that in simulation study II. 
Again, the distributions of $\fR(\gamma)$ shown in Figures~\ref{fig:Rx1} and~\ref{fig:Rx2} may be biased since models with large $\fR(\gamma)$ are unlikely to be sampled (clearly, $\PP(\gamma) < p^{-\fR(\gamma)}$).  Those models with very large $\fR$ values appearing in Figure~\ref{fig:Rx1} are collected in the early stage of MCMC when the chain has not reached stationarity (recall that we always initialize the sampler at some randomly generated $\gamma$ with $|\gamma| = 10$).

Observe that in every panel of Figure~\ref{fig:Rx1}, the highest bar always occurs around $3.5$, since almost every overfitted model has $\fR(\gamma) \approx \kappa = 3.5$ for the reason explained in the last subsection. If $\fR(\gamma)$ is significantly smaller than $\kappa$, then  $\gamma$ is probably underfitted, and there are two main reasons why an underfitted $\gamma$ may have a small $\fR(\gamma)$ violating Condition~\ref{cond:ywj}.
First, the effect size of the missing covariate(s) may not be sufficiently large; that is, when we propose adding a covariate to $\gamma$, the gain in R-squared is not large enough to dominate the penalty on model size in~\eqref{eq:ppr}. This explains why in Figure~\ref{fig:Rx1} and Table~\ref{table:Rx1} we see that Condition~\ref{cond:ywj} is more likely to be (approximately) satisfied when SNR $=3$. 
Second, it is also possible that more than one influential covariates with large effect sizes are missing in $\gamma$, but these effects cancel out each other due to the correlation between the covariates. 
This is again consistent with the result shown in Figure~\ref{fig:Rx1} and is also the predominant reason why the posterior distribution on $\cM(s_0)$ may be multimodal.  

In simulation study II, we simulate the data such that $\PP$ tends to be severely multimodal. Since we use $\kappa = 1.5$ in this study, in Figure~\ref{fig:Rx2} we see that there is always a high bar near $1.5$. However, because we initialize the MH algorithms in high-posterior regions, few models collected by the chains have $\fR(\gamma) > 1.5$. Clearly, distributions of $\fR(\gamma)$ in Figure~\ref{fig:Rx2} are very different from those in Figure~\ref{fig:Rx1} (note that $y$-axis is in log scale in Figure~\ref{fig:Rx1} but not in Figure~\ref{fig:Rx2}). This is mainly because in the second simulation study we randomly generate a large number of correlated causal covariates with normally distributed effect sizes, and thus many  models may have small or even negative $\fR$. 

Overall, for the first simulation study, $\PP$ appears to be approximately satisfied in most cases, while for the second simulation study, $\PP$ seems to be severely multimodal and violate Condition~\ref{cond:ywj}.

\begin{figure} 
    \centering
    \includegraphics[width=0.98\linewidth]{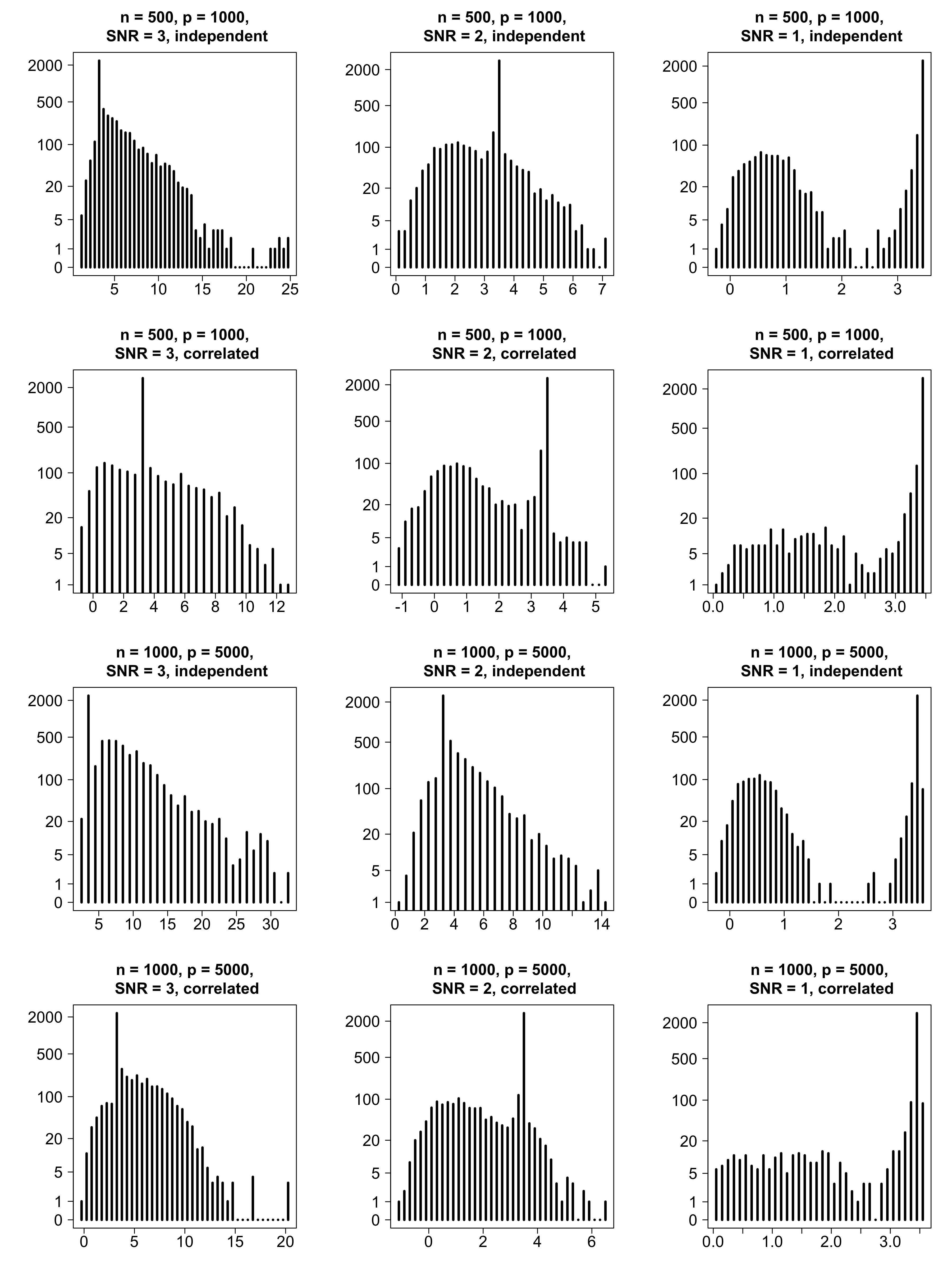}
    \caption{Histograms of $\fR(\gamma)$ in simulation study I. 
    For each simulation setting, we pool together the unique models sampled by \LOT{}-1, \LOT{}-2, and \LIB{}-1 algorithms for all 100 simulated data sets; $\hat{\gamma}_{\rm{max}}$ for each data set is excluded. 
    Note that  $y$-axis is in log scale. 
    \label{fig:Rx1} }
\end{figure}

\begin{figure} 
    \centering
    \includegraphics[width=0.98\linewidth]{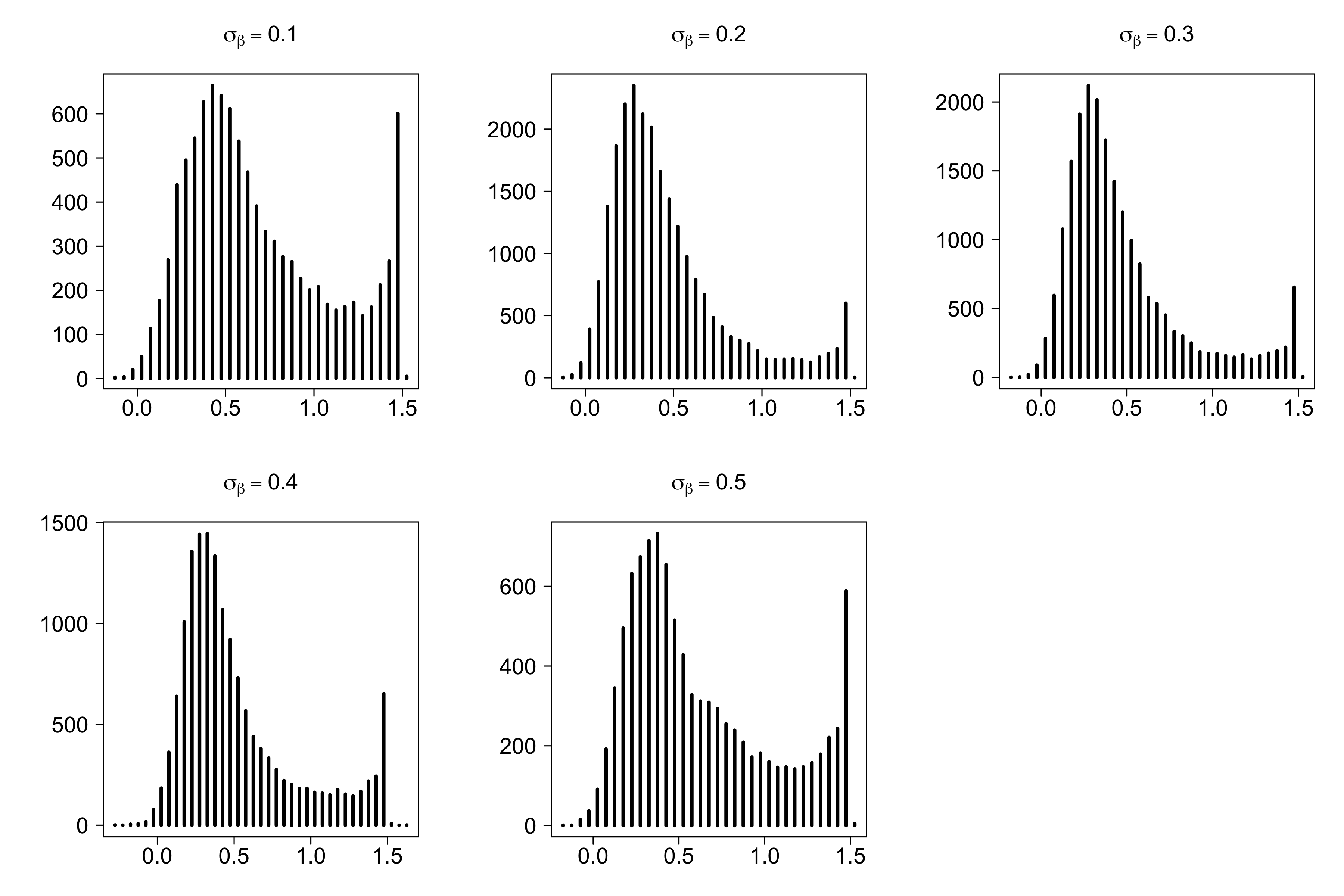}
    \caption{Histograms of $\fR(\gamma)$ in simulation study II. 
    For each choice of $\sigma_\beta$, we pool together the unique models sampled by \LOT{}-1, \LOT{}-2, and \LIB{}-1 algorithms for all 20 simulated data sets; $\hat{\gamma}_{\rm{max}}$ for each data set is excluded. 
    \label{fig:Rx2} }
\end{figure}

\begin{table} 
\caption{\label{table:Rx1} Distribution of $\fR(\gamma)$ in simulation study I. We report the percentage of unique models with $\fR(\gamma) \geq c$ for $c = 1, 2, 3$ using the data shown in Figure~\ref{fig:Rx1}.}
\begin{tabular}{ccccc}
\toprule
      &  &  $\fR(\gamma) \geq 1$ & $\fR(\gamma) \geq 2$ & $\fR(\gamma) \geq 3$ \\ 
    \midrule
    \multirow{3}{*}{\shortstack{$n = 500$, $p = 1000$, \\ independent design}} & SNR$=3$ & 
    100\% & 99.4\% & 96.1\%   \\ 
    & SNR$=2$ & 98.2\% & 87.9\% & 77.3\%   \\ 
    & SNR$=1$ & 83.2\% & 78.3\% & 77.9\%   \\ 
    \midrule
    \multirow{3}{*}{\shortstack{$n = 500$, $p = 1000$, \\ correlated design }} & SNR$=3$ & 92.4\% & 86.9\% & 82.4\%   \\ 
    & SNR$=2$ & 84.2\% & 77.8\% & 75.4\%   \\ 
    & SNR$=1$ & 98.2\% & 95.4\% & 94.1\%   \\ 
    \midrule
        \multirow{3}{*}{\shortstack{$n = 1000$, $p = 5000$, \\ independent design}} & SNR$=3$ & 
    100\% & 100\% & 99.6\%   \\ 
    & SNR$=2$ & 99.9\% & 98.1\% & 92.6\%   \\ 
    & SNR$=1$ & 75.5\% & 73.7\% & 73.6\%   \\ 
    \midrule
    \multirow{3}{*}{\shortstack{$n = 1000$, $p = 5000$, \\ correlated design }} & SNR$=3$ & 99.1\% & 96.7\% & 93.5\%   \\ 
    & SNR$=2$ & 87.5\% & 77.8\% & 72.9\%   \\ 
    & SNR$=1$ & 97.5\% & 94.5\% & 93.4\%   \\ 
\bottomrule
\end{tabular}
\end{table}

\clearpage

\subsection{Number of local modes in simulation study II }\label{supp:sim.nmodes} 
Consider simulation study II. 
In Table~\ref{table2}, we have reported the number of unique local modes (with respect to the single-flip neighborhood relation) sampled by each informed MH sampler in 2000 iterations. 
For comparison, we run \LOT{}-1 for $5 \times 10^5$ iterations under the same initialization. The number of local modes (averaged over 20 data sets) is given in Table~\ref{supp.table:modes}. Except in the case $\sigma_\beta = 0.1$, the number of local modes of $\pi$ seems to be very large and the chain keeps finding new local modes as we run it for more iterations, 
which suggests that it is reasonable to use the number of sampled local modes to describe how fast the sampler explores the posterior.  

\begin{table} 
    \caption{ \label{supp.table:modes} The number of unique local modes sampled by \LOT{}-1, averaged over 20 data sets.}
    \begin{tabular}{cccccc}
    \hline 
    \rule{0pt}{2.5ex}  
      \LOT{}-1 iterations & $\sigma_\beta = 0.1$   & $\sigma_\beta = 0.2$ &  $\sigma_\beta = 0.3$   & $\sigma_\beta = 0.4$ &      $\sigma_\beta = 0.5$ \\ 
    \hline 
    \rule{0pt}{2.5ex}  2K   &  2.30 & 6.20 & 5.05 & 3.85 & 2.75 \\
     500K   & 4.10 & 37.6 & 36.9 & 28.1 & 16.4 \\    
    \hline 
    \end{tabular}
\end{table}

\subsection{Effective sample size estimation in simulation study II }\label{supp:sim.ess}  
Finally, we consider ESS estimation. In Table~\ref{table2}, we calculate ESS estimates using function \texttt{effectiveSize}  in the \texttt{R} package \texttt{coda}~\citep{plummer2006coda}. 
This method, particularly for binary variables, may yield inconsistent estimation since it relies on estimating the asymptotic variance of a fitted univariate  autoregressive process. 
A better method is implemented in function \texttt{multiESS} in the \texttt{R} package \texttt{mcmcse}, which builds nonparametric estimators of the (possibly multivariate) limiting covariance matrix, making only mixing assumptions that are automatically satisfied for discrete space processes~\citep{flegal2017mcmcse}. 

Though $\gamma$ is usually the parameter of interest, it is unrealistic to calculate a multivariate ESS using $\gamma$ directly (treating it as a $p$-dimensional binary vector), since the sample covariance matrix of $\gamma$ is almost always singular. 
We propose to construct a $q$-dimensional summary statistic of $\gamma$ for some small $q$ as follows (we use $q = 5$).  
First, given the sample path of an MCMC algorithm, let $\gamma_{(1)}, \dots, \gamma_{(q)}$ denote the $q$ models with largest posterior probabilities. 
Let $\bar{\gamma} = \cup_{j \in [q]} \gamma_{(j)}$  denote the set of covariates that are included in at least one of the $q$ best models. 
Second, assuming $q$ divides $p$, for $i = 1, \dots, q$, let 
\begin{align*}
    \Gamma_{(i)} = \left\{ \frac{(i-1) p}{q} + 1, \frac{(i-1) p}{q} + 2, \dots, \frac{i p}{q} \right\} 
\end{align*}
and  define a model $\tilde{\gamma}_{(i)}$ by  
\begin{align*}
   \tilde{\gamma}_{(i)} = \left\{ \Gamma_{(i)} \setminus \bar{\gamma} \right\} \cup \gamma_{(i)}.
\end{align*}
Finally, given an arbitrary $\gamma \subset [p]$,  construct a $q$-dimensional statistic $F(\gamma)$ by 
\begin{align*}
    F(\gamma) = ( \, \mathrm{HD}(\gamma, \tilde{\gamma}_{(1)}), \, \dots, \,     \mathrm{HD}(\gamma, \tilde{\gamma}_{(q)}) \, ), 
\end{align*}
where $\mathrm{HD}(\gamma, \gamma') = | \gamma \triangle \gamma'|$ denotes the Hamming distance between two models. 
Hence, the $i$-the element of $F(\gamma)$ is the distance from $\gamma$ to the $i$-th ``reference model'' $\tilde{\gamma}_{(i)}$. 

This construction of $F(\gamma)$ can certainly be modified in many ways. For example, one can sample $\gamma_{(1)}, \dots, \gamma_{(q)}$ from the MCMC sample path. 
However, it is worth explaining why we choose to build $\tilde{\gamma}_{(i)}$ by combining $\gamma_{(i)}$ and $\Gamma_{(i)}$. 
Since $p$ is often very large, to create a low-dimensional statistic of $\gamma$,  we probably want to only focus on those ``influential'' variables. This can be achieved by calculating the distance from $\gamma$ to some high-posterior-probability models $\gamma_{(1)}, \dots, \gamma_{(q)}$. 
But it is very likely that $\gamma_{(1)}, \dots, \gamma_{(q)}$ are highly correlated with each other, so the sample covariance matrix of the statistic 
$( \, \mathrm{HD}(\gamma, \gamma_{(1)}), \, \dots, \,     \mathrm{HD}(\gamma, \gamma_{(q)}) \, )$ tends to be ill-conditioned, making the corresponding ESS estimate unstable. 
Adding $\Gamma_i$ to $\gamma_{(i)}$ has three benefits. 
First, compared with $\mathrm{HD}(\gamma, \gamma_{(i)})$, the statistic $\mathrm{HD}(\gamma, \tilde{\gamma}_{(i)})$ has a larger variance and thus is more indicative of the mixing behavior of the chain. 
Second, even if $\gamma_{(i)}$ and $\gamma_{(j)}$ only differ by one covariate, the correlation between $\mathrm{HD}(\gamma, \tilde{\gamma}_{(i)})$ and $\mathrm{HD}(\gamma, \tilde{\gamma}_{(j)})$  is usually not high due to the use of $\Gamma_i$ and $\Gamma_j$.  Third, our definition of $F(\gamma)$ makes use of information from all $p$ variables. 
   
Table~\ref{table2.supp} displays the ESS estimates calculated by the \texttt{R} package \texttt{mcmcse}. Let $T_1, T_2$ be as defined in the main text. 
We use $T_3$ to denote the multivariate statistic $F(\gamma)$ defined above with $q = 5$ and reference models being the best $5$ models from the sample path of \LOT{}-1 with $500K$ iterations, and let $T_4 = (T_1, T_2, T_3)$ be a $7$-dimensional statistic. 
All the four ESS statistics show strong evidence that \LOT{} algorithms mix much faster than \RW{}, and interestingly, ESS$(T_2)$ and ESS$(T_3)$ behave very similarly (up to a scaling factor) except when $\sigma_\beta = 0.5$. 
Note that for $T_1$ and $T_2$, ESS estimates are different from those in Table~\ref{table2} since the estimation method is different, and the estimates in Table~\ref{table2.supp} are actually more favorable to \LOT{} algorithms. 
  
\begin{table} 
 \caption{\label{table2.supp} Effective sample size results for simulation study II.  ESS($T_i$) is the estimated effective sample size of the statistic $T_i$ calculated by function \texttt{multiESS} in the \texttt{R} package \texttt{mcmcse}.    
All estimates are averaged over 20 data sets.}
\begin{tabular}{cccccc}
\toprule
  &  &  \RW{} & \LOT{}-1 & \LOT{}-2 & \LIB{}-1   \\
 & Number of iterations &  200,000 & 2,000 & 2,000 & 2,000 \\  
\midrule
 \multirow{4}{*}{ \shortstack{ $\sigma_\beta = 0.1$ \\ Mean model size = 6.1 }} & ESS($T_1$)/Time  & 0.513 & 18.3 & 13.4 & 4.8\\ 
  & ESS($T_2$)/Time  & 2.51 & 28.5 & 25.5 & 10.2\\ 
  & ESS($T_3$)/Time  & 1.74 & 20.6 & 23.2 & 10.4\\ 
  & ESS($T_4$)/Time  & 2.2 & 25.8 & 29.4 & 14.9\\ 
\cmidrule{1-6} 
 \multirow{4}{*}{ \shortstack{ $\sigma_\beta = 0.2$ \\ Mean model size = 26.4 }} & ESS($T_1$)/Time  & 0.347 & 4.09 & 3.47 & 2.07\\ 
  & ESS($T_2$)/Time  & 2.49 & 20.3 & 19 & 11.2\\ 
  & ESS($T_3$)/Time  & 0.727 & 8.25 & 7.51 & 4.67\\ 
  & ESS($T_4$)/Time  & 0.977 & 11.5 & 10.8 & 6.83\\ 
\cmidrule{1-6} 
 \multirow{4}{*}{ \shortstack{ $\sigma_\beta = 0.3$ \\ Mean model size = 50.2 }} & ESS($T_1$)/Time  & 0.244 & 2.24 & 2.28 & 1.19\\ 
  & ESS($T_2$)/Time  & 2.42 & 18.8 & 16.2 & 8.68\\ 
  & ESS($T_3$)/Time  & 0.651 & 5.2 & 4.91 & 2.84\\ 
  & ESS($T_4$)/Time  & 0.915 & 6.81 & 7.2 & 4.48\\ 
\cmidrule{1-6} 
 \multirow{4}{*}{ \shortstack{ $\sigma_\beta = 0.4$ \\ Mean model size = 63.9 }} & ESS($T_1$)/Time  & 0.238 & 2.26 & 1.78 & 0.859\\
  & ESS($T_2$)/Time  & 1.98 & 15.1 & 13.6 & 6.86\\ 
  & ESS($T_3$)/Time  & 0.628 & 4.42 & 4.48 & 2.87\\ 
  & ESS($T_4$)/Time  & 0.927 & 6.23 & 5.99 & 4.1\\ 
\cmidrule{1-6} 
 \multirow{4}{*}{ \shortstack{ $\sigma_\beta = 0.5$ \\ Mean model size = 71.6 }} & ESS($T_1$)/Time  & 0.37 & 3.13 & 2.84 & 1.13\\ 
  & ESS($T_2$)/Time  & 1.71 & 13.5 & 13.1 & 4.96\\ 
  & ESS($T_3$)/Time  & 1.2 & 5.28 & 6.44 & 3.18\\ 
  & ESS($T_4$)/Time  & 1.43 & 6.66 & 7.94 & 4.28\\ 
\bottomrule
\end{tabular}
\end{table}  

\clearpage 
\newpage 
\section{Proofs for variable selection and \LOT{} } \label{sec:proof.sel}

\subsection{Proof for Example~\ref{ex1}}\label{proof.ex1}
\begin{proof}

First, a straightforward calculation gives $y^\top y =  (1 + 2\nu )n$ and 
\begin{align*}
    y^\top P_{ \gamma } y = \left\{ \begin{array}{ll}
        0, &   \text{ if } \gamma = \emptyset \text{ or } \{i\} \text{ for } i = 3, 4, \dots, p, \\
        \nu^2 n,  & \text{ if } \gamma = \{1\} \text{ or } \{2\}, \\
       2 \nu n, \quad & \text{ if } \gamma = \{1, 2\}. 
    \end{array}\right. 
\end{align*}
By~\eqref{eq:ppr}, we have 
\begin{align*}
  &  \frac{ \PP(\{ 3 \} )}{\PP(\emptyset)} =  \cdots = \frac{ \PP(\{ p \} )}{\PP(\emptyset)} = p^{ - \kappa}, \\ 
  &  \frac{ \PP(\{ 1 \} )}{\PP(\emptyset)} =  \frac{ \PP(\{ 2 \} )}{\PP(\emptyset)} = p^{ - \kappa} \left( 1 - \frac{ \nu^2}{(1 + g^{-1}) (1 + 2\nu)} \right)^{-n / 2},  \\ 
 & \frac{ \PP(\{ 1, 2\} )}{\PP(\{1\})} =  \frac{ \PP(\{1, 2 \} )}{\PP(\{2\})} = p^{ - \kappa} \left( 1 - \frac{\nu (2 - \nu)}{ (1 + g^{-1})(1 + 2\nu) - \nu^2 } \right)^{-n / 2}. 
\end{align*}
Using the bound $e^{- n x/ (1 - x)} \leq (1 - x)^n \leq e^{-n x}$, we get 
\begin{align*}
    p^{ - \kappa} e^{n a_1 / \nu}  \leq  \frac{ \PP(\{ 1 \} )}{\PP(\emptyset)} \leq & \; p^{ - \kappa}e^{n a_2} ,   \text{ and }
    \frac{ \PP(\{ 1 \} )}{\PP(\{1, 2\})} \leq  p^{ \kappa} e^{- n a_3 }. 
\end{align*}
where 
\begin{align*}
    a_1 = \frac{ \nu^3 / 2}{(1 + g^{-1}) (1 + 2\nu)}, \quad a_2 =  \frac{ \nu^2 / 2 }{(1 + g^{-1}) (1 + 2\nu) - \nu^2}, \quad 
    a_3 =  \frac{ \nu (2 - \nu)/2 }{(1 + g^{-1}) (1 + 2\nu) - \nu^2}. 
\end{align*} 
Since we assume $p, \nu, g, \kappa$ are all fixed, we get 
\begin{align*}
    \bK_\nu(\emptyset, \{1\} \cup \{2\})  
    \geq \frac{ 2 e^{n  a_1} }{ p + 2 e^{n a_1} } = 1 - O(e^{- n a_1}). 
\end{align*}
As shown in the main text, the transition probability $\bP_\nu(\emptyset, \{1\})$ is bounded by 
\begin{align*}
    \left\{ \frac{ \PP( \{1\} ) }{ \PP(\emptyset)} \right\}^{1 - \nu}  \left\{ \frac{ \PP( \{1\}) }{ \PP( \{1, 2\}) } \right\}^{ \nu} 
    \leq p^{-\kappa ( 1 - 2\nu)} e^{n a_2 (1 - \nu) - n a_3 \nu } = p^{-\kappa ( 1 - 2\nu)} e^{ - n a_2}.  
\end{align*}
Hence, $\bP_\nu(\emptyset, \{1\} \cup \{2\}) = O(e^{- n a_2})$. 
\end{proof}

\subsection{Proof of Lemma~\ref{lm:R1}}\label{proof.lemmas}
\begin{proof} 
Consider part~\eqref{lm1.1} first. 
Since $y^\top \OPJ_\gamma y \in [0, y^\top y]$, we have
\begin{align*}
1 \leq  \frac{   g^{-1} y^\top y +  y^\top  \OPJ_\gamma   y }{  g^{-1} y^\top y }   \leq 1 + g, 
\end{align*}
and thus $V_1(\gamma) \in [1, e]$. The bounds for $V_2$ are evident since $\gamma \in \cM(s_0)$. 
Part~\eqref{lm1.2} follows from the fact that $y^\top \OPJ_{\gamma \cup \{j\} } y \leq y^\top \OPJ_{\gamma  } y $
since $\OPJ_\gamma$ projects a vector onto the space orthogonal to the column space of $\X_\gamma$. 
For part~\eqref{lm1.3},  we use the following two inequalities, 
\begin{equation}\label{ineq.exp}
e^{-  x} \leq 1 - \frac{x}{2},  \quad \quad  e^{  x} \leq 1 + 2 x, \quad \quad \forall \, x \in [0, 1]. 
\end{equation}
Then, for $k \in \gamma \setminus \gamma^*$, the bound on $R_2(\gamma, \gamma \setminus \{k\}) $ follows from the first inequality above, and similarly,  for $j \in (\gamma \cup \gamma^*)^c$,  the bound on $R_2(\gamma, \gamma \cup \{j\})$ follows from the second. 
\end{proof}
 
\subsection{Drift condition for overfitted models}
In this subsection, we prove the drift condition provided in Proposition~\ref{th:overfit}.
Recall the weighting functions defined in~\eqref{eq:def.weights}  in which the corresponding normalizing constants can be expressed by   
\begin{equation*}\label{eq:Zs}
\begin{aligned}
\Zadd(\gamma)  = \;&    \sum_{\gamma' \in \adds(\gamma)} \left(  \PR(\gamma, \gamma')  \wedge p^{c_1} \right),  \\ 
\Zdel(\gamma)  = \;&    \sum_{\gamma' \in \dels(\gamma)} \left(  1 \vee \PR(\gamma, \gamma')   \wedge p^{c_0} \right),  \\ 
\Zswap(\gamma)  = \;&    \sum_{\gamma' \in \swaps(\gamma)} \left(   p s_0 \vee \PR(\gamma, \gamma')   \wedge p^{c_1} \right). 
\end{aligned}
\end{equation*} 
Under Condition~\ref{cond:ywj},  we can bound $\Zstar(\gamma)$ for an overfitted $\gamma$ as follows. 
\begin{lemma}\label{lm:overfit1}
Suppose   Condition~\ref{cond:ywj} holds and $\gamma \in \cM(s_0)$ is an overfitted model. 
\begin{enumerate}[(i)]
\item $\Zadd(\gamma) \leq p^{1 - c_0 }$.   \label{o1}
\item For any $k \in \gamma$, $\wdel( \gamma \setminus \{k\}  \mid \gamma  ) = p^{c_0}$ if $k \in \gamma \setminus \gamma^*$,  and $1$ if $k \in \gamma^*$.   \label{o2}
\item $\Zdel(\gamma) = (|\gamma| - s^*) p^{c_0}  + s^*$.  \label{o3}
\item For any $j \notin \gamma$ and $k \in \gamma^*$,  $\PR(\gamma,  (\gamma \cup \{j\}) \setminus \{k\}) \leq p^{- (c_0 + c_1)}$.  \label{o4}
\end{enumerate}
\end{lemma}
\begin{proof} 
For any overfitted model $\gamma$ with $|\gamma| \leq s_0$,   by Condition~\cond{c1}, 
\begin{align*}
\Zadd(\gamma) = \sum_{\gamma' \in \adds(\gamma)} \left( \PR(\gamma, \gamma') \wedge p^{c_1} \right) \leq |\adds(\gamma)|  p^{-c_0} \leq p^{1 - c_0}, 
\end{align*}
since any $\gamma'$ in $\adds(\gamma)$ is obtained by adding a non-influential covariate to $\gamma$. 
 
To prove part~\eqref{o2}, note that Condition~\cond{c1} implies that for any $k \in \gamma \setminus \gamma^*$, we have 
$$\PR(\gamma, \gamma \setminus \{k\}) = \PR(\gamma \setminus \{k\}, \gamma)^{-1} \geq p^{ c_0},$$ 
since $\gamma \setminus \{k\}$ is still overfitted. 
If $k \in \gamma^*$,  then $\gamma' = \gamma \setminus \{k\}$ is underfitted and $\Ga = \{k\}$. 
Hence,  by Condition~\cond{c2},  $\PR(\gamma, \gamma') \leq p^{ - c_1} < 1$. 
Part~\eqref{o3} follows from~\eqref{o2}. 
 
Consider a swap move. For any $j \notin \gamma$ and $k \in \gamma^* \subset \gamma$, we have 
\begin{align*}
 \PR(\gamma,  (\gamma \cup \{j\}) \setminus \{k\}) = \PR(\gamma, \gamma \cup \{j\}) \PR(\gamma \cup \{j\}, (\gamma \cup \{j\}) \setminus \{k\}) \leq p^{- (c_0 + c_1)}, 
\end{align*}
since $ \PR(\gamma, \gamma \cup \{j\})  \leq p^{ - c_0}$ by Condition~\cond{c1} and $\PR(\gamma \cup \{j\}, (\gamma \cup \{j\}) \setminus \{k\})  \leq p^{- c_1}$ by Condition~\cond{c2}.  
Part~\eqref{o4} then follows. 
\end{proof}

The drift condition provided in Proposition~\ref{th:overfit} follows from the following lemma. 
\begin{lemma}\label{lm:overfit.all}
Suppose that Condition~\ref{cond:ywj} holds for some $c_0 \geq 2$ and $c_1 \geq 1$. 
For any overfitted model $\gamma$ such that $\gamma \neq \gamma^*$ and $|\gamma| \leq s_0$,  
\begin{align*}
\sum_{\gamma' \in \adds(\gamma)} R_2 (\gamma, \gamma') \MHP(\gamma, \gamma') \leq \;&   \frac{1}{s_0 p^{c_0 - 1}}  , \\
\sum_{\gamma' \in \dels(\gamma)} R_2 (\gamma, \gamma') \MHP(\gamma, \gamma')  \leq \;& -  \frac{1  }{4 s_0} + O \left(   \frac{1}{s_0 p^{c_0 - 1}} \right), \\
 \sum_{\gamma' \in \swaps(\gamma)}  R_2 (\gamma,  \gamma') \MHP(\gamma,   \gamma')  \leq \;& p^{-c_0}. 
\end{align*}
\end{lemma}
\begin{proof} 
Consider addition first. Since $\gamma$ is overfitted, we can only add non-influential covariates. 
For any $j \notin \gamma$, it follows from~\eqref{eq:mh2} and Condition~\cond{c1}  that 
$$\MHP(\gamma, \gamma \cup \{j\})  \leq \frac{ \PR(\gamma, \gamma \cup \{j\}) }{2} \leq  \frac{1}{ 2 p^{c_0} }.$$ 
Thus, using Lemma~\ref{lm:R1}\eqref{lm1.3} we obtain that 
\begin{align*}
\sum_{\gamma' \in \adds(\gamma)} R_2 (\gamma, \gamma') \MHP(\gamma, \gamma')
\leq   \sum_{\gamma' \in \adds(\gamma)}  \frac{ 1 }{s_0 p^{c_0}} = 
\frac{ p - |\gamma| }{s_0 p^{c_0}} \leq \frac{1}{s_0 p^{c_0 - 1}}. 
\end{align*}

Consider deletion moves. Observe that $V_2$ only changes if we remove a non-influential covariate. 
For any $k \in \gamma \setminus \gamma^*$,  
$\wadd( \gamma \mid \gamma \setminus \{k\}) = \PR(\gamma \setminus \{k\}, \gamma) \leq p^{- c_0}$ by Condition~\cond{c1}. 
Apply Lemma~\ref{lm:overfit1}\eqref{o1} to get 
\begin{align*}
 \PR(\gamma, \gamma \setminus \{k\}) \KLOT(\gamma \setminus \{k\}, \gamma)  =\;& \PR(\gamma, \gamma \setminus \{k\}) \frac{ \wadd( \gamma \mid \gamma \setminus \{k\}) }{2  \Zadd(\gamma \setminus \{k\})} \\
 =\;& \frac{ 1}{2  \Zadd(\gamma \setminus \{k\})} \geq \frac{p^{c_0 - 1}}{2}  \geq 1. 
\end{align*}
Thus, by~\eqref{eq:mh2}, 
$ \MHP(\gamma, \gamma \setminus \{k\})  =  \KLOT(\gamma, \gamma \setminus \{k\}).$ 
Applying Lemma~\ref{lm:R1}\eqref{lm1.2}, we find that 
\begin{align*}
  - R_2 (\gamma, \gamma \setminus \{k\}) \MHP(\gamma, \gamma \setminus \{k\}) 
\geq    \frac{  \KLOT(\gamma, \gamma \setminus \{k\})  }{2 s_0}
=     \frac{ p^{c_0}}{ 4 s_0 \left[  (|\gamma| - s^* ) p^{c_0} + s^*  \right] }. 
\end{align*}
Since $c_0 \geq 2$ and there are $(|\gamma| - s^*)$ non-influential covariates that we may remove, 
\begin{align*}
 - \sum_{\gamma' \in \dels(\gamma)} R_2 (\gamma, \gamma') \MHP(\gamma, \gamma') 
 \geq \frac{(|\gamma| - s^*)     (4 s_0)^{-1} }{  |\gamma| - s^* + s^* p^{-c_0}}
 = \frac{1  }{4 s_0} + O\left( \frac{1}{s_0 p^{ c_0 - 1 } }\right). 
\end{align*}
In the last step we have used that $|\gamma| - s^* \geq 1$, $s^* <  p$, and $(1 + x)^{-1} \sim 1 - x$ for $x = o(1)$. 

For the swap moves, note that $V_2$ only changes if we swap an influential covariate $k \in \gamma^*$ with a non-influential covariate $j \notin \gamma$.
The total number of such pairs  is $ (p - s_0) s^*$.   Let $\gamma' = (\gamma \cup \{j\}) \setminus \{k\}$ denote the resulting model.  
By Lemma~\ref{lm:R1}\eqref{lm1.3},  Lemma~\ref{lm:overfit1}\eqref{o4} and~\eqref{eq:mh2}, 
\begin{align*}
   R_2 (\gamma,  \gamma') \MHP(\gamma,   \gamma') 
\leq  R_2 (\gamma,  \gamma') \frac{ \PR(\gamma,  \gamma') }{2}
\leq \frac{ 1 }{  s_0}  \PR(\gamma,  \gamma')  
 \leq \;& \frac{ 1 }{ s_0 p^{c_0 + c_1}}    . 
\end{align*}
Since $(p - s_0) s^* \leq p s_0$,
\begin{align*}
 \sum_{\gamma' \in \swaps(\gamma)}  R_2 (\gamma,  \gamma') \MHP(\gamma,   \gamma')  \leq   \frac{1}{p^{c_0 + c_1 - 1}}
 \leq \frac{1}{p^{c_0}}, 
\end{align*}  
which concludes the proof. 
\end{proof}

\subsection{Drift condition for underfitted models}
In this subsection, we prove the drift condition provided in Proposition~\ref{th:underfit}. We first prove two auxiliary results. 
\begin{lemma}\label{lm:underfit1}
Suppose   Condition~\ref{cond:ywj} holds and $\gamma \in \cM(s_0)$ is an underfitted model. 
\begin{enumerate}[(i)]
\item $ \Zadd(\gamma) \geq p^{c_1}$.  \label{u1}
\item $s_0 \leq  \Zdel(\gamma) \leq  s_0 p^{c_0}$.  \label{u2}
\item If $|\gamma| = s_0$, then  $p^{c_1} \leq \Zswap(\gamma) \leq  s_0 p^{c_1 + 1}$.  \label{u3}
\end{enumerate}
\end{lemma}
\begin{proof} 
By Condition~\cond{c2},   there exists some $j^*$ such that  $\PR( \gamma, \, \gamma \cup\{j ^* \} ) \geq p^{c_1}$, which proves part~\eqref{u1}. 
Part~\eqref{u2} follows from the definition of $\wdel$ and that $|\dels(\gamma)| = |\gamma| \leq s_0$. 
A similar argument using Condition~\cond{c3} and $|\swaps(\gamma)| \leq p s_0$ proves part~\eqref{u3}.  
\end{proof}

\begin{lemma}\label{lm:ineq.R1}
Suppose that  $\PR(\gamma, \gamma') \geq p^a$ for some $a \in \bbR$ and define
\begin{align*}
b =  \frac{    a + \kappa ( |\gamma'| - |\gamma|)  }{n \kappa_1}. 
\end{align*}
If $b \in [0, 1]$, then $- R_1(\gamma, \gamma') \geq  b / 2$. 
If $b \in [-1, 0]$, then $R_1(\gamma, \gamma')  \leq  - 2 b$. 
\end{lemma}
 \begin{proof} 
First, by~\eqref{eq:post} and the definition of $V_1$,  we have 
\begin{align*}
\log \left\{ p^{\kappa ( |\gamma'| - |\gamma| )} \PR(\gamma, \gamma') \right\}  
= - \frac{n \log(1 + g)}{2} \log \frac{  V_1(\gamma') }{  V_1(\gamma) }. 
\end{align*}
Using  $1 + g = p^{2 \kappa_1}$ and $R_1(\gamma, \gamma') = V_1(\gamma') / V_1(\gamma) - 1$,  we obtain that 
\begin{align*}
\frac{ \log \PR(\gamma, \gamma') }{\log p} =  \kappa (|\gamma| - |\gamma'|)   -  n \kappa_1     \log [ 1 + R_1(\gamma, \gamma')  ]. 
\end{align*}
Then,   $\PR(\gamma, \gamma') \geq p^a$ implies that 
\begin{align*}
-R_1(\gamma, \gamma') \geq  1  - \exp\left\{ - \frac{a + \kappa ( |\gamma'| - |\gamma|)}{n \kappa_1}   \right\} 
= 1 - e^{ - b}.  
\end{align*}
Applying the two inequalities in~\eqref{ineq.exp} yields the result. 
\end{proof}

The drift condition provided in Proposition~\ref{th:underfit} follows from the following lemma. 

\begin{lemma}\label{lm:underfit.all}
Suppose Condition~\ref{cond:ywj} holds and $\gamma \in \cM(s_0)$ is underfitted. 
\begin{enumerate}[(i)]
\item We always have 
\begin{align*}
 0 \leq \sum_{\gamma'  \in \dels(\gamma)} R_1 (\gamma, \gamma') \MHP(\gamma, \gamma')
 \leq  \frac{ |\gamma| (e  - 1) }{2 p^{c_1}}. 
\end{align*}
\item If $c_1 \geq  ( c_0 + 1 ) \vee 2$ and $ \kappa + c_1  \leq    n \kappa_1$, 
\begin{align*}
-  \sum_{\gamma'  \in \adds(\gamma)} R_1 (\gamma, \gamma') \MHP(\gamma, \gamma')
 \geq    \frac{ \kappa + c_1  }{8 n \kappa_1}. 
\end{align*}
\item If  $|\gamma| = s_0$,  $ 4 \leq c_1  \leq n  \kappa_1$, $ n = O(p)$, $\kappa = O(s_0)$ and $s_0 \log p = O(n)$,  
\begin{align*}
-  \sum_{\gamma'  \in \swaps(\gamma)} R_1 (\gamma, \gamma') \MHP(\gamma, \gamma')
\geq  \frac{c_1 }{8 n \kappa_1} + o\left(   \frac{1}{n \kappa_1} \right).  
\end{align*}
\end{enumerate}
\end{lemma}
 
 \begin{proof}[Proof of part (i) (deletion)]
By Lemma~\ref{lm:underfit1}\eqref{u1},  we have
\begin{align*}
 \PR(\gamma, \gamma \setminus \{k\})  \KLOT(  \gamma \setminus \{k\}, \gamma)  = \frac{ 
  \PR(\gamma, \gamma \setminus \{k\}) \wadd( \gamma \mid  \gamma \setminus \{k\} ) }{2 \Zadd(\gamma \setminus \{k\}  ) } \leq \frac{1}{2 p^{c_1} },
\end{align*}
since $B(\gamma, \gamma') \wadd(\gamma' \mid \gamma) \leq 1$ for any $\gamma' \in \adds(\gamma)$. 
It then follows from~\eqref{eq:mh2}  that 
\begin{align*}
R_1 (\gamma, \gamma \setminus \{k\}) \MHP(\gamma, \gamma \cup \{k\})  
\leq \; &   \frac{ R_1 (\gamma, \gamma \setminus \{k\})    }{2p^{c_1}}. 
\end{align*}
By Lemma~\ref{lm:R1}\eqref{lm1.1}, $R_1 (\gamma, \gamma \setminus \{k\})  \leq  e - 1$, from which the asserted bound follows.   
\end{proof}
  
\begin{proof}[Proof of part (ii) (addition)]
Define  a set of ``good'' addition moves as 
$$\cG = \cG(\gamma) = \{ \gamma \cup \{j\} \colon j \notin \gamma, \,     \PR(\gamma, \gamma \cup \{j\}) \geq p^{c_1 - 1} \}.$$ 
Our goal is to bound $- R_1 (\gamma, \gamma')   \MHP(\gamma, \gamma')$ summed over $\gamma' \in \cG$. 

By Condition~\cond{c2}, $\cG$ contains at least one element, which we denote by $\cT(\gamma)$,  such that $\PR(\gamma, \cT(\gamma)) \geq p^{c_1}$. By Lemma~\ref{lm:ineq.R1} and the assumption that $c_1 + \kappa \leq n \kappa_1$,  
\begin{equation}\label{eq:def.A0}
  - R_1 (\gamma, \cT(\gamma))  \geq  \frac{  \kappa + c_1  }{2 n \kappa_1}   \eqqcolon A. 
\end{equation}
Using Lemma~\ref{lm:ineq.R1} again and the assumption that $c_1 \geq 2$ (which implies $2(\kappa + c_1 - 1) \geq \kappa + c_1$),  we find that  
\begin{equation} \label{eq:A1}
  - R_1 (\gamma, \gamma')  \geq  \frac{   \kappa + c_1 - 1  }{2 n \kappa_1} \geq  \frac{A}{2}, \quad  \forall \, \gamma' \in \cG. 
\end{equation}

To bound $\MHP(\gamma, \gamma')$ for $\gamma' \in \cG$, observe that for any $\gamma' \in \adds(\gamma)$, 
\begin{align*}
     \KLOT(\gamma', \gamma) = \frac{ \wdel(\gamma' \mid \gamma ) }{2 \Zdel(\gamma')} \geq \frac{1}{2 \Zdel(\gamma')} \geq \frac{1}{2 s_0 p^{c_0} } 
\end{align*}
by Lemma~\ref{lm:underfit1}\eqref{u2}. 
It then follows from~\eqref{eq:mh2}  that for any $\gamma' \in \adds(\gamma)$, 
\begin{align*}
 \MHP(\gamma, \gamma')   = \;& \min \{ \KLOT(\gamma, \gamma'), \, \PR(\gamma, \gamma') \KLOT(\gamma', \gamma) \}  \\
\geq \;&    \min \left\{
\frac{ \wadd(  \gamma' \mid \gamma ) }{2 \Zadd(\gamma)}, \, \frac{\PR(\gamma, \gamma')}{2 s_0 p^{c_0}} \right\}   \\
\geq \;& \frac{\wadd(  \gamma' \mid \gamma ) }{2  } \min \left\{ \frac{ 1}{  \Zadd(\gamma)}, \, \frac{1}{ s_0 p^{c_0}} \right\}. 
\end{align*}
By Lemma~\ref{lm:underfit1}\eqref{u1} and the assumption $c_1 \geq c_0 + 1$,  we have $ \Zadd(\gamma) \geq p^{c_1} \geq s_0 p^{c_0}$.
Using the above displayed inequality, we  obtain that 
\begin{equation}\label{eq:under.add1}
 \MHP(\gamma, \gamma')      \geq   \frac{ \wadd( \gamma' \mid \gamma ) }{ 2 \Zadd(\gamma) }, \quad  \forall \, \gamma' \in \adds(\gamma).
\end{equation} 

Define $\cG' = \cG \setminus \{\cT(\gamma) \}$, which may be empty. 
Let $ W =  \sum_{\gamma' \in \cG'}\wadd(  \gamma' \mid \gamma).$ 
Then, 
\begin{equation}\label{eq:under.Z}
\begin{aligned}
\Zadd(\gamma) =\;& \sum\nolimits_{\gamma' \in \adds(\gamma)} \wadd( \gamma' \mid \gamma )   \\
= \;& p^{c_1} + W + \sum\nolimits_{\gamma' \in \adds(\gamma) \setminus  \cG(\gamma) } \wadd( \gamma' \mid \gamma )    \\
\leq  \;& W + 2 p^{c_1}, 
\end{aligned}
\end{equation}
since for any $\gamma' \in \adds(\gamma) \setminus  \cG(\gamma)$, we have $\PR(\gamma, \gamma') < p^{c_1 - 1}$.  
By Lemma~\ref{lm:R1}\eqref{lm1.2}, $R_1(\gamma, \gamma') \leq 0$ for any $\gamma' \in \adds(\gamma)$,  which implies that 
\begin{align*}
   \sum_{\gamma' \in \adds(\gamma)}  R_1 (\gamma, \gamma')   \MHP(\gamma, \gamma')     
   \leq      \sum_{\gamma' \in \cG} R_1 (\gamma, \gamma')   \MHP(\gamma, \gamma'). 
\end{align*}
Some algebra using~\eqref{eq:def.A0},~\eqref{eq:A1},~\eqref{eq:under.add1} and~\eqref{eq:under.Z} yields that 
 \begin{equation}\label{eq:similar}
 \begin{aligned}
    - \sum_{\gamma' \in \cG} R_1 (\gamma, \gamma')   \MHP(\gamma, \gamma')       
\geq  \;&  \frac{Ap^{c_1}}{2 \Zadd(\gamma)}  +  \sum_{\gamma' \in \cG'}  \frac{  A   \wadd(\gamma' \mid \gamma )    }{4 \Zadd(\gamma)} \\
=  \;&   \frac{  A (2 p^{c_1} +  W)  }{4 \Zadd(\gamma)}   \geq   \frac{A}{4}, 
\end{aligned}
\end{equation}
which concludes the proof. 
\end{proof}

\begin{proof}[Proof of part (iii) (swap)]
First, we use an argument similar to the proof of part (ii) to analyze those ``good'' moves. 
Define 
\begin{align*}
\cG_1(\gamma) =\left \{\gamma' \in \swaps(\gamma) \colon  \PR(\gamma, \gamma') \geq  \frac{ p^{c_1} }{ p s_0 }  \right\}. 
\end{align*}
By Condition~\cond{c3}, there exists  $\cT(\gamma) \in \cG_1(\gamma)$ such that $\PR(\gamma, \cT(\gamma)) \geq p^{c_1}$.  
By Lemma~\ref{lm:ineq.R1}, 
\begin{align*}
- R_1(\gamma, \cT(\gamma))   \geq \frac{c_1 }{2 n \kappa_1}. 
\end{align*}
Similarly, for any $\gamma' \in \cG_1(\gamma)$, 
\begin{align*}
- R_1(\gamma, \gamma') \geq   \frac{ (c_1 - 1)  -  (\log s_0)/(\log p) }{2 n \kappa_1} \geq   \frac{ c_1 - 2  }{2 n \kappa_1} 
\geq \frac{c_1 }{4 n \kappa_1}, 
\end{align*}
since $c_1 \geq 4$. 
By Lemma~\ref{lm:underfit1}\eqref{u3}, $\Zswap( \gamma')  \leq s_0 p^{c_1 + 1}$. 
Further, since $c_1 \geq 4$, for any $\gamma' \in \cG_1(\gamma)$, $\PR(\gamma, \gamma')  \geq \wswap( \gamma' \mid \gamma)$.  Applying~\eqref{eq:mh2}, we obtain that
\begin{align*}
   \MHP(\gamma, \gamma')         \geq  \wswap( \gamma' \mid \gamma)  \min \left\{  \frac{1}{2  \Zswap(\gamma)}, \; \frac{  p s_0}{2  s_0 p^{c_1 + 1} } \right\} 
   =  \frac{ \wswap( \gamma' \mid \gamma )  }{ 2  \Zswap(\gamma) }, \quad    \forall \, \gamma' \in \cG_1. 
\end{align*} 
By a calculation  similar to~\eqref{eq:under.Z} and~\eqref{eq:similar}, we find that 
\begin{align*}
\Zswap(\gamma) \leq 2 p^{c_1} + \sum_{\gamma' \in \cG_1 \setminus \{ \cT(\gamma)\} } \wswap( \gamma' \mid \gamma), 
\end{align*}
which yields 
\begin{equation}\label{eq:swap.G1}
- \sum_{\gamma' \in \cG_1 }  R_1 (\gamma, \gamma' )   \MHP(\gamma, \gamma' )     \geq  \frac{ c_1 }{8 n \kappa_1}. 
\end{equation}

Let $\cG_2(\gamma) = \{\gamma' \in \swaps(\gamma) \colon   \PR(\gamma, \gamma')  < 1 \}$. 
We claim that any $\gamma' \in \cG_2$ is still underfitted. 
If $\gamma' \in \cG_2$ is overfitted, $\gamma' = (\gamma \cup \{j\}) \setminus \{k\}$  for some $j \in \Ga$ and $k \in \Gb$. 
By Lemma~\ref{lm:overfit1}\eqref{o4},  $\PR(\gamma', \gamma) \leq p^{- (c_0 + c_1)} $, which yields the contradiction. 
Since $\gamma' \in \cG_2$ is underfitted, $\Zswap(\gamma') \geq p^{c_1}$ by Lemma~\ref{lm:underfit1}\eqref{u3}. 
Consequently, 
\begin{align*}
 \MHP(\gamma, \gamma')     \leq \PR(\gamma, \gamma') \KLOT(\gamma', \gamma) \leq \PR(\gamma, \gamma')\frac{  \wswap(\gamma \mid \gamma') }{2  p^{c_1}}. 
\end{align*}
If $\PR(\gamma, \gamma') < (p s_0)^{-1}$, then $\PR(\gamma, \gamma')  \wswap(\gamma \mid \gamma')  \leq 1$.  By Lemma~\ref{lm:R1}\eqref{lm1.1},  $R_1(\gamma, \gamma') \leq e - 1$, which yields that 
\begin{align*}
R_1(\gamma, \gamma') \MHP(\gamma, \gamma')  \leq  \frac{e - 1}{2 p^{c_1}}. 
\end{align*}
If $\PR(\gamma, \gamma') \in [ (p s_0)^{-1}, 1)$,  $\wswap(\gamma \mid \gamma') = p s_0$, and thus $\MHP(\gamma, \gamma') \leq s_0 / 2 p^{c_1 - 1} $ by Lemma~\ref{lm:underfit1}\eqref{u3}. 
Apply Lemma~\ref{lm:ineq.R1} to get 
\begin{align*}
R_1(\gamma, \gamma') \MHP(\gamma, \gamma')  \leq \frac{2 \log (ps_0) }{ n \kappa_1 \log p}    \MHP(\gamma, \gamma')  \leq \frac{  s_0 \log (ps_0) }{ n \kappa_1  p^{c_1 - 1} \log p}      \leq  \frac{ 2 s_0  }{ n \kappa_1  p^{c_1 - 1}  }.
\end{align*}
Using the assumptions $c_1 \geq 4$, $n = O(p)$ and $\kappa = O(s_0)$,  we get 
\begin{equation}\label{eq:swap.G2}
R_1(\gamma, \gamma') \MHP(\gamma, \gamma') = O \left(    \frac{s_0}{ n \kappa_1  p^3}  \right).
\end{equation}  

If $\gamma' \in \swaps(\gamma) \setminus \cG_2(\gamma)$, then $R_1(\gamma, \gamma') \leq 0$. 
Hence, it follows from~\eqref{eq:swap.G1} and~\eqref{eq:swap.G2} that 
\begin{align*}
  \sum_{\gamma'  \in \swaps(\gamma)} R_1 (\gamma, \gamma') \MHP(\gamma, \gamma')
\leq \;&    \sum_{\gamma' \in \cG_1  }  R_1 (\gamma, \gamma' )   \MHP(\gamma, \gamma' )     
 + \sum_{\gamma' \in \cG_2} R_1 (\gamma, \gamma' )   \MHP(\gamma, \gamma' )       \\
\leq \;&  - \frac{c_1  }{8 n  \kappa_1} +  O \left(    \frac{s_0^2}{ n \kappa_1  p^2}  \right). 
\end{align*}
Finally, notice that $s_0 \log p =O(n)$ and $n = O(p)$ imply that $s_0 = o(p)$. 
Hence, the asymptotic term in the last step is $o ( (n \kappa_1)^{-1} )$. 
\end{proof}

\subsection{Proof of Theorem~\ref{th:mh.mix}}\label{proof.mix}
\begin{proof}
To apply the results provided in Section~\ref{sec:main}, we need to consider the lazy version of the transition matrix $\MHP$,  $\bP_{\rm{lazy}} = (\MHP + \bI) / 2$. But this  is equivalent to dividing all the proposal probabilities by $2$. 
Hence, by Propositions~\ref{th:overfit} and~\ref{th:underfit}, for the lazy chain we have 
\begin{align*}
\lambda_1 = 1 - \frac{c_1}{16 n \kappa_1} + o( (n \kappa_1)^{-1} ), \quad \quad
\lambda_2  = 1 - \frac{1}{8 s_0} + o ( s_0^{-1}).
\end{align*}
For sufficiently large $n$, we can assume that 
\begin{align*}
\lambda_1 \leq 1 - \frac{c_1}{20 n \kappa_1}, \quad \quad
\lambda_2  \leq 1 - \frac{1}{10 s_0}. 
\end{align*}
Let $q$ be the probability of the chain escaping from the set $\cO(\gamma^*, s_0)$. 
By Condition~\cond{c2} and Lemma~\ref{lm:overfit1}\eqref{o4}, $q \leq p^{-c_0} \leq p^{-2}$. 
Since $\kappa_1 \leq \kappa = O(s_0) = o(p)$, we have $q = o(  (1 - \lambda_1)  \wedge (1  - \lambda_2) )$. 
By Lemma~\ref{lm:R1}\eqref{lm1.1},  $ \sup_{\gamma \in \cO} V_2(\gamma) \leq e$ and $\sup_{\gamma \in \cM(s_0)} V_1(\gamma) \leq e$. 
Thus, we may assume $\bP_{\rm{lazy}}$ satisfies that assumptions of Corollary~\ref{th:hd} with $K = e$, $M = 2e$ and $C = 6 / e$ (other values of $C$ will yield slightly different constants in the bound).
The asserted upper bound on the mixing time then follows from a routine calculation using Corollary~\ref{coro:hd}. 
\end{proof}

\subsection{Proof of Lemma~\ref{lm:general} } 
\label{sec:proof.general}
\begin{proof}
Assume $\pi(x') / \pi(x) \geq b $ and $x' \in \cN(x)$. 
To show the acceptance probability of the proposal move from $x$ to $x'$ is $1$, it suffices to prove that 
\begin{align*}
    \frac{ \pi(x') }{ \pi(x)}\frac{ \tilde{\bK}_{\rm{lb}}(x', x) }{ \tilde{\bK}_{\rm{lb}}(x, x')}  \geq 1. 
\end{align*}  
We simply bound $\tilde{\bK}_{\rm{lb}}(x, x')$ by $1$. 
Since $f$ is non-decreasing and $\pi(x) / \pi(x') \leq b^{-1}$, the assumption $f( b^{-1} ) \leq \underline{f}$ implies that 
$f \left( \frac{\pi(x)}{\pi(x')}  \right) \leq \underline{f}.$
Hence, by the definition of $\tilde{\bK}_{\rm{lb}}$, we have 
\begin{align*}
    \tilde{\bK}_{\rm{lb}}(x', x) =\;& \frac{ \underline{f} }{ \tilde{Z}_f(x')  } \geq \frac{ \underline{f} }{ |\cN(x')| \overline{f} }. 
\end{align*}
The claim then follows from the assumption that $\underline{f} b \geq \overline{f}  |\cN(x')| $. 
\end{proof}

\section{Data and code availability}\label{sec:data}
The two GWAS data sets used in Section~\ref{sec:gwas} are the Primary Open-Angle Glaucoma Genes and Environment (GLAUGEN) Study (accession no. phs000308.v1.p1) and the National Eye Institute Glaucoma Human Genetics Collaboration (NEIGHBOR) Consortium Glaucoma Genome-Wide Association Study (accession no. phs000238.v1.p1). 
Both can be obtained from dbGaP (\url{https://www.ncbi.nlm.nih.gov/gap/}). The genotype data of both studis were generated using the Illumina Human660W-Quad\_v1\_A beadchip. 
The code used for simulation studies described in Section~\ref{sec:sim} is available at \url{https://web.stat.tamu.edu/~quan/code/lit\_mh.tgz}. 

\bibliographystyles{plainnat}
\bibliographys{ref.bib}

\end{document}